\newif\iffull
\fulltrue

\iffull
\documentclass[nonacm]{acmart}
\else
\documentclass[acmsmall, 10pt]{acmart}
\setcopyright{none}
\fi

\usepackage{scalerel}
\usepackage{accents}
\usepackage{xspace}

\newcommand{\rulename}[1]{\textsc{#1}}
\newcommand{\Rule}[4][]{\ensuremath{\inferrule*[right={(#2)},#1]{#3}{#4}}}

\newcommand{\RuleNolabel}[3][]{\ensuremath{\inferrule*[#1]{#2}{#3}}}

\newcommand{\spawn}{\texttt{spawn}\xspace}

\newcommand{\termspawn}{create}
\newcommand{\termSpawn}{Create}
\newcommand{\termsync}{sync}
\newcommand{\termSync}{Sync}

\newcommand{\bnfdef}{::=}
\newcommand{\bnfalt}{\mid}

\newcounter{ccount}

\renewcommand{\paragraph}[1]{%
{\normalfont\normalsize\bfseries #1}%
}%

\newcommand{\altdiv}{\mathbin{\mathrlap{\mathrlap{\sim}:}\phantom{\sim}}}

\newcommand{\calcname}{\ensuremath{\lambda^{4}_s}}
\newcommand{\langname}{\calcname}

\newcommand{\ctx}{\Gamma}

\theoremstyle{definition}
\newtheorem{defn}{Definition}

\newcommand{\ectx}{\cdot}

\newcommand{\dom}[1]{\mathit{dom}(#1)}

\newcommand{\defeq}{\triangleq}

\newcommand{\priowork}[1]{W_{#1}}
\newcommand{\prioworkof}[2]{\priowork{#2}(#1)}
\newcommand{\psnlt}[1]{\not\prec #1}
\newcommand{\longsp}[2]{S_{#2}(#1)}

\newcommand{\anc}[2]{#1 \sqsupseteq #2}
\newcommand{\sanc}[2]{#1 \sqsupseteq^s #2}
\newcommand{\wanc}[2]{#1 \sqsupseteq^w #2}
\newcommand{\nanc}[2]{#1 \not\sqsupseteq #2}

\newcommand{\compwork}[2]{\ensuremath{\hspace{0.25em}\mathop{\mathclap{\nuparrow}\mathclap{\downarrow}}\hspace{0.25em}#2}}

\newcommand{\uprio}[2]{\mathit{Prio}_{#1}(#2)}

\newcommand{\strengthen}[2]{\hat{#2}_{#1}}

\newcommand{\tscomp}{\cdot}

\newcommand{\gthread}[3]{\ensuremath{#1 \xhookrightarrow[#2]{} #3}}

\newcommand{\gthreads}{\mathcal{T}}
\newcommand{\spawns}{E^c}
\newcommand{\syncs}{E^s}
\newcommand{\reads}{E^w}
\newcommand{\dagq}[4]{\ensuremath{(#1, #2, #3, #4)}}

\newcommand{\resptimeof}[1]{\ensuremath{T(#1)}}

\newcommand{\kw}[1]{\mbox{\texttt{#1}}}

\newcommand{\cdsqbracks}[1]{\kw{[}{#1}\kw{]}}
\newcommand{\cd}[1]{{\lstinline!#1!}}

\newcommand{\kwunit}{\kw{unit}}
\newcommand{\kwnat}{\kw{nat}}

\newcommand{\kwcvt}[2]{\ensuremath{\kw{cv}\cdsqbracks{#1; #2}}}
\newcommand{\kwmutext}[1]{\ensuremath{\kw{mutex}\cdsqbracks{#1}}}

\newcommand{\kwreft}[1]{\ensuremath{#1~\kw{ref}}}

\newcommand{\prio}{\rho}

\newcommand{\ple}[2]{#1 \preceq #2}
\newcommand{\plt}[2]{#1 \prec #2}

\newcommand{\nple}[2]{#1 \not\preceq #2}
\newcommand{\pnlt}[2]{#1 \not\prec #2}

\newcommand{\worlds}{R}
\newcommand{\prios}{R}

\newcommand{\kwnumeral}[1]{\overline{#1}}
\newcommand{\kwn}{\kwnumeral{n}}

\newcommand{\kwassn}{s}
\newcommand{\kwtriv}{\langle \rangle}

\newcommand{\kwspawn}[4]{\ensuremath{\kw{spawn}\cdsqbracks{#1}\{#4\}}}
\newcommand{\kwlet}[3]{\ensuremath{\kw{let}~#1 = #2~\kw{in}~#3}}

\newcommand{\kwnewref}[2]{\ensuremath{\kw{ref}\cdsqbracks{#1}~#2}}
\newcommand{\kwref}[1]{\ensuremath{\kw{ref}\cdsqbracks{#1}}}
\newcommand{\kwderef}[1]{\ensuremath{!#1}}
\newcommand{\kwassign}[2]{\ensuremath{#1 := #2}}

\newcommand{\cvname}{\ensuremath{\alpha}}
\newcommand{\mutname}{\ensuremath{\beta}}
\newcommand{\cvmutname}{\ensuremath{\eta}}
\newcommand{\kwwait}[1]{\ensuremath{\kw{wait}~#1}}
\newcommand{\kwsignal}[1]{\ensuremath{\kw{signal}~#1}}
\newcommand{\kwcv}[1]{\ensuremath{\kw{cv}\cdsqbracks{#1}}}
\newcommand{\kwmutex}[1]{\ensuremath{\kw{mutex}\cdsqbracks{#1}}}
\newcommand{\kwnewcv}[1]{\ensuremath{\kw{newcv}[#1]}}
\newcommand{\kwnewmutex}[1]{\ensuremath{\kw{newmutex}[#1]}}

\newcommand{\kwwithlock}[2]{\ensuremath{\kw{with}~#1~#2}}
\newcommand{\kwwithlocks}[2]{\ensuremath{\kw{with*}~#1~#2}}
\newcommand{\kwwithlocksp}[4]{\ensuremath{\kw{with*}\cdsqbracks{#3; #4}~#1~#2}}
\newcommand{\kwpromote}[2]{\ensuremath{\kw{promote}~#1~\kw{to}~#2}}
\newcommand{\kwskip}{\ensuremath{\kw{skip}}}
\newcommand{\kwif}[3]{\ensuremath{\kw{if}~#1~\kw{then}~#2~\kw{else}~#3}}
\newcommand{\kwwhile}[2]{\ensuremath{\kw{while}~#1~#2}}

\newcommand{\instr}{i}
\newcommand{\block}{b}

\newcommand{\fresh}{~\kw{fresh}}

\newcommand{\graph}{g}
\newcommand{\gnode}[2]{(#1, #2)}

\newcommand{\scompsym}{\oplus}
\newcommand{\scomp}[1]{\scompsym_{#1}}
\newcommand{\egraph}{\emptyset}
\newcommand{\edgespawn}[2]{(#1, #2)^c}
\newcommand{\edgesync}[2]{(#1, #2)^s}
\newcommand{\edgeweak}[2]{(#1, #2)^w}
\newcommand{\lastv}[2]{\mathit{LastV}_{#1}(#2)}

\newcommand{\sig}{\Sigma}

\newcommand{\esig}{\cdot}
\newcommand{\emem}{\emptyset}

\newcommand{\hastype}[2]{#1 \mathop{:} #2}

\newcommand{\sigrtype}[2]{#1 \mathord{\sim} #2}
\newcommand{\sigcvtype}[2]{#1 \mathord{@} #2}

\newcommand{\etyped}[5][\worlds]{#3 \vdash^{#1}_{#2} #4 : #5}

\newcommand{\ityped}[8][\worlds]{#3; #4 \vdash^{#1}_{#2} #5 \altdiv #6 \mathord{@} #7 \dashv #8}
\newcommand{\styped}[7][\worlds]{#3; #4 \vdash^{#1}_{#2} #5 \mathord{@} #6 \dashv
                                     #7}
\newcommand{\mtyped}[2]{#1 \vdash #2}
\newcommand{\cfgtyped}[3][\worlds]{#2 \vdash^{#1} #3}
\newcommand{\meetc}[3][\worlds]{#2 \vdash^{#1} #3}

\newcommand{\sstyped}[6][\worlds]{#3 \vdash^{#1}_{#2} #4 \mathbin{@} #5 \dashv #6}
\newcommand{\stackaccepts}[8][\worlds]{\vdash^{#1}_{#2} #3 \mathbin{\vartriangleleft :} (#4, #5, #6) \leadsto (#7, #8)}
\newcommand{\stackacceptsc}[7][\worlds]{\vdash^{#1}_{#2} #3 \mathbin{\blacktriangleleft :} (#4, #5) \leadsto (#6, #7)}

\newcommand{\fctx}{\Psi}
\newcommand{\fracp}{\ensuremath{\pi}}
\newcommand{\splitsto}[3]{#1 \rightarrow #2 \otimes #3}
\newcommand{\fctxwf}[1]{\kw{Splittable}(#1)}
\newcommand{\fowned}{\kw{Owned}}
\newcommand{\fshared}{\kw{Shared}}
\newcommand{\fnone}{\kw{None}}
\newcommand{\permof}[3]{#1(#2, #3)}

\newcommand{\tpcp}{\otimes}
\newcommand{\tp}{\mu}

\newcommand{\mem}{\sigma}

\newcommand{\memrent}[3]{(#1, #2, #3)}
\newcommand{\mement}[4]{#1 \mapsto \memrent{#2}{#3}{#4}}

\newcommand{\memupd}[2]{#1[#2]}

\newcommand{\stack}{k}
\newcommand{\estack}{\epsilon}
\newcommand{\ssend}[2]{#1 \mathbin{\vartriangleright} #2}
\newcommand{\sreturn}[2]{#1 \mathbin{\vartriangleleft} #2}
\newcommand{\scsend}[2]{#1 \mathbin{\blacktriangleright} #2}
\newcommand{\screturn}[1]{#1 \mathbin{\blacktriangleleft}}
\newcommand{\shole}{\text{\---}}
\newcommand{\scp}[2]{#1:: #2}
\newcommand{\stackstate}{K}
\newcommand{\stframe}{f}

\newcommand{\twlconfig}[4]{#1; #2; #3; #4}
\newcommand{\lconfig}[5]{#1; #2; #3; #4 \mid #5}

\newcommand{\cthread}[4]{\ensuremath{#1 \xhookrightarrow[#2, #3]{} #4}}

\newcommand{\mstep}{\mathbin{\mathbf{\Rightarrow}}}

\newcommand{\waiters}{\ensuremath{W}}
\newcommand{\waitemp}{[]}
\newcommand{\waitcons}[2]{#1::#2}
\newcommand{\waitupd}[3]{#1[#2 \mapsto #3]}

\newcommand{\locked}{\ensuremath{L}}
\newcommand{\lunlocked}{\bot}

\newcommand{\prioceil}[5]{\ensuremath{#1 \underset{#2, #3, #4}{\leadsto} #5}}

\newcommand{\secput}[2]{\section{#2}\label{sec:#1}}
\newcommand{\secref}[1]{Section~\ref{sec:#1}}

\newcommand{\figref}[1]{Figure~\ref{fig:#1}}

\renewcommand{\eqref}[1]{Equation~(\ref{eq:#1})}

\newif\ifnotes
\relax

\newcounter{remark}[section]

\newcommand{\code}[1]{\lstinline!#1!}

\newcommand{\punt}[1]{}

\newcommand{\mylineskip}{0.3cm}
\newcommand{\langfigsize}{\small}

\notesfalse 

\usepackage{relsize}
\usepackage{mathtools}
\usepackage{mathpartir}
\usepackage{listings, multicol, wrapfig}
\usepackage{tikz}
\usepackage{code}

\usepackage[T1]{fontenc}
\usepackage{microtype}

\usepackage{flushend}
\begin{document}

\title{Responsive Parallelism with Synchronization}         

\author{Stefan K. Muller}
\email{smuller2@iit.edu}
\affiliation{Illinois Institute of Technology\country{USA}}

\author{Kyle Singer}
\email{kdsinger@wustl.edu}
\affiliation{Washington University in St. Louis\country{USA}}

\author{Devyn Terra Keeney}
\email{dkeeney2@hawk.iit.edu}
\affiliation{Illinois Institute of Technology\country{USA}}

\author{Andrew Neth}
\email{aneth@hawk.iit.edu}
\affiliation{Illinois Institute of Technology\country{USA}}

\author{Kunal Agrawal}
\email{kunal@wustl.edu}
\affiliation{Washington University in St. Louis\country{USA}}

\author{I-Ting Angelina Lee}
\email{angelee@wustl.edu}
\affiliation{Washington University in St. Louis\country{USA}}

\author{Umut A. Acar}
\email{umut@cmu.edu}
\affiliation{Carnegie Mellon University\country{USA}}

\begin{abstract}
Many concurrent programs assign priorities to threads to improve
responsiveness. When used in conjunction with synchronization mechanisms
such as mutexes and condition variables, however, priorities can lead
to {\em priority inversions}, in which high-priority threads are delayed by
low-priority ones. Priority inversions in the use of mutexes are easily
handled using dynamic techniques such as {\em priority inheritance}, but
priority inversions in the use of condition variables are not well-studied and
dynamic techniques are not suitable.

In this work, we use a combination of static and dynamic techniques to prevent
priority inversion in code that uses mutexes and condition variables. A type
system ensures that condition variables are used safely, even while dynamic
techniques change thread priorities at runtime to eliminate priority inversions
in the use of mutexes. We prove the soundness of our system, using a model of
priority inversions based on cost models for parallel programs.
To show that the type system is practical to implement, we encode it within the
type systems of Rust and C++, and show that the restrictions are not overly
burdensome by writing sizeable case studies using these encodings, including
porting the Memcached object server to use our C++ implementation.

\end{abstract}

\maketitle
\renewcommand{\shortauthors}{S. K. Muller, K. Singer, D. T. Keeney, A. Neth,  K. Agrawal, I. Lee, and U. A. Acar}

\section{Introduction}\label{sec:intro}

For decades, software applications have used concurrency to perform
tasks simultaneously on multiple threads.
In many such applications, particularly those that interact with a user or
the outside environment, different tasks have different resource and latency
requirements.
For example, in an email client, an event loop processing user input must run
frequently to ensure responsiveness of the user interface, but a background
thread that uses spare cycles to compress stored emails has no such
requirement.
Many systems for multithreading allow programmers to associate \emph{priorities}
with threads, corresponding to these requirements: a highly latency-sensitive
thread might run at high priority, while a less-sensitive background thread
would run at low priority.

For almost as long as priorities have been used in concurrent programs,
programmers and researchers have observed the problem of \emph{priority
  inversions}.
While many formulations of the problem exist with slight variations, the
essence is this: a higher-priority thread finds itself waiting for a
lower-priority thread to release some resource or establish some condition.
In most cases, the wait will be brief and such
\emph{bounded priority inversions} are sometimes not considered harmful in
themselves.
However, if the lower-priority thread in this scenario is preempted by a
long-running thread of intermediate priority, this could become an
\emph{unbounded} priority inversion, which can severely impact responsiveness
or prevent the program from making progress entirely.
As such, it is desirable to be able to detect and prevent priority inversions.

When priority inversions are caused by contention on simple synchronization
primitives such as locks or mutexes, these can be fixed
using dynamic mechanisms such as
\emph{priority inheritance} or \emph{priority ceiling}~\citep{srl-priority-1990},
both of which temporarily raise the priority of the low-priority
thread so it can run and release the resource.
However, the problem of priority inversion is much less understood in the
context of \emph{condition variables}, a powerful and general
synchronization mechanism that can be used to encode many other mechanisms
such as semaphores and monitors~\cite{hansen1973operating,hoare1974monitors,hansen1975programming}.
Indeed, it is impossible to use dynamic techniques analogous to priority
inheritance or priority ceiling to detect and prevent priority inversions
caused by condition variables.
To see why this is the case, consider a program consisting of a main thread~$A$,
which spawns threads~$B$ and (some time later)~$C$.
If thread~$B$ \emph{waits} on a condition variable (an operation which blocks
until another thread performs a corresponding \emph{signal} operation), there
is no way for the thread scheduler to know which thread will perform the
signal operation: indeed, thread~$C$ may not even have been spawned
at the time of the ``wait''.
Thus, it is not even possible to determine on the fly
whether this blocking constitutes a
priority inversion.

In this paper, we consider a static approach to the problem:
our solution is a type system that soundly determines if a priority inversion
might occur.
This information can then be used either to reject programs with potential
inversions, or as a warning to the programmer that a particular portion of
the code might require additional reasoning and/or debugging to ensure that
an inversion cannot occur.
We study the type system in the context of an imperative calculus
with static and dynamic semantics that handle
both condition variables and
locks, preventing priority inversions using a combination of static
and dynamic techniques: we use the dynamic priority ceiling protocol
for mutexes and static techniques for condition variables.

The type system statically associates a priority with each mutex and condition
variable.
For mutexes, this is the priority ceiling---the highest priority at which it
might be acquired.
For condition variables, it is a priority higher than (or equal to)
threads that might wait on it and lower than (or equal to) threads that might
signal it.
For condition variables, the fact that such a priority exists implies that
any signaling thread will be higher-priority than any waiting thread, which
would appear to rule out priority inversions.
This is not, however, sufficient: consider
the three-thread scenario above where threads~$B$ and~$C$ are high-priority,
but~$A$ is low-priority:~$B$ must wait for lower-priority thread~$A$ to
create thread~$C$, which will signal and awake~$B$.
Both threads that use the condition variable are high-priority, but there is
still a priority inversion.
We prevent such inversions by adding a notion of
ownership~\citep{crarywamo99, swm00, Cyclone02, Boyland03, afm07-l3}
of condition variables and allowing code to {\em promote}, or raise the
priority of, condition variables under certain conditions.
In the above example, thread~$A$ could only promote the condition variable to
high priority, allowing thread~$B$ to wait on it, {\em after} creating
thread~$C$, thus avoiding the priority inversion.
The type system {\em does not} statically restrict the use of mutexes,
other than to ensure that the annotated priority ceiling is correct,
but does ensure that the use of the priority ceiling protocol would not
cause runtime type errors in code that mixes the use of mutexes and
condition variables.

Proving that the type system is sound in detecting priority inversions
requires a formal model of priority inversions that is rich enough to describe
scenarios such as the three-thread priority inversion above.
Formalisms such as that of \citet{bms-formal-1993}, which is based on
relationships between threads at snapshots in time, cannot encode an inversion
between a waiting thread and a not-yet-created signaling thread.
Instead, we build on a recent line of work~\cite{mah-responsive-2017, mah-priorities-2018, MullerSiGoAcAgLe20}
that represents parallel programs with priorities as Directed Acyclic Graphs
(DAGs) and shows that, if a DAG meets certain ``well-formedness'' conditions,
the corresponding program can be scheduled efficiently by a simple scheduling
principle that greedily observes priorities.
This yields a definition of priority inversions in terms of the impact they
have on performance: a priority inversion is
any interaction among threads that delays high-priority
threads due to low-priority ones, and we
can be assured that no such delay occurs for well-formed DAGs.
%
%
This prior work also considers type systems for preventing priority inversions,
but for programs that only synchronize at the completion of threads; this is
a much weaker model than that allowed by mutexes and
condition variables.
We extend the DAG model, well-formedness definitions, and scheduling proofs,
to encode concurrent programs that use mutexes and condition variables,
and show that any DAG arising from a well-typed program is well-formed and
therefore free of priority inversions.

Despite its power, the type system we propose in this paper uses some
well-established ideas, albeit in novel ways.
Indeed, we are able to encode approximations of its restrictions (with some
limitations) using advanced features of the C++ and Rust type systems.
To show that the restrictions of the type system do not overly hinder
productivity, we have used the encodings to develop sizable case
study programs, including a large real-world interactive application, the 
Memcached object caching server~\cite{Memcached09} (v1.5.13), which consists of
about 20,100 lines of C code.

\iffull
We begin with an overview of the key ideas of the type system through
examples (Section~\ref{sec:overview}), followed by a formal presentation
of the type system (Section~\ref{sec:lang}) and cost model
(Section~\ref{sec:dag}).
In Section~\ref{sec:proof}, we define the dynamic semantics of the core
calculus, including facilities for the priority ceiling protocol, and prove
type safety and well-formedness (absence of priority inversions).
We then discuss the C++ and Rust implementations and case studies,
discuss related work, and conclude.
\fi

\section{Overview and Examples}\label{sec:overview}

%
In this section, we present via several examples the complexities of
preventing priority inversions involving mutexes and
condition variables and describe the main ideas
behind our proposed approach, which we formalize in
Section~\ref{sec:lang}.

\subsection{Condition Variables}

\begin{figure}
\small
\begin{minipage}{0.49\textwidth}
\begin{lstlisting}
void <Low>f(CV cv, int *result) {
  ...
  *result = ...;
  signal(cv);
}
void <High>main() {
  CV cv = new CV();
  int result = 0;
  spawn<Low>(f(CV, &result));
  ...
  wait(cv);
}
\end{lstlisting}
\caption{An example priority inversion.}
\label{fig:fut-inv}
\end{minipage}\hfill%
\begin{minipage}{0.49\textwidth}
\begin{lstlisting}
void <Low>f(CV<High> cv, int *result) {
  ...
  *result = ...;
  signal(cv); //ill-typed
}
void <High>main() {
  CV<High> cv = new CV<High>();
  int result = 0;
  spawn<Low>(f(CV, &result));
  ...
  wait(cv);
}
\end{lstlisting}
\caption{The priority inversion is now a type error.}
\label{fig:fut-terr}
\end{minipage}
\end{figure}

We consider condition variables with three operations (in addition to
the operation that constructs a new CV): \texttt{wait},
\texttt{signal}, and~\texttt{promote}.
The~\texttt{wait} operation causes the calling thread to become blocked
on the CV, and the~\texttt{signal} operation resumes one thread that is
currently blocked on the CV\footnote{We discuss later in the paper how
  a~\texttt{broadcast} operation that resumes all blocked threads can be
  added as a straightforward extension.}
The~\texttt{promote} operation is discussed later in this section.
%
The code in Figure~\ref{fig:fut-inv} uses an asynchronous thread to run a
function~\texttt{f} and write its result into a
reference~\texttt{result}, after which it signals a condition
variable~\texttt{cv} to indicate the result is available.\footnote{This may
  be seen as an encoding of {\em futures} or {\em promises}, a popular
  mechanism for expressing parallelism.}
The main thread then does some other work and, when the result of the
asynchronous computation is needed, waits on the condition variable.

Unfortunately, this code
suffers from a priority inversion: the \texttt{High}-priority
\texttt{main} thread waits on the \texttt{Low}-priority thread
running~\texttt{f}.
It is impossible to detect this priority inversion statically
using only priority annotations on threads as in prior work~\citep{MullerSiGoAcAgLe20},
because
when~\texttt{wait} is called, the type system does not know which thread
will call~\texttt{signal}.
Instead, we assign a priority to the CV
itself and restrict the use of CVs
by threads.
%
%
Initially, we require that:
\begin{enumerate}
\item To wait on a CV with priority $\prio$, a thread's
  priority must be less than or equal to $\prio$, and

\item To signal a CV with priority $\prio$, a thread's
  priority must be greater than or equal to $\prio$.
\end{enumerate}
This transitively ensures that any signaling thread is higher-priority
than any waiting thread.
The two conditions above become the first two restrictions enforced by
the type system we formalize in Section~\ref{sec:lang}; we will add to this
list as we consider more examples of priority inversions.

In Figure~\ref{fig:fut-terr}, we define the condition variable
~\texttt{cv} with priority~\texttt{High}, because the \texttt{main}
thread waits on it.
Now, the priority inversion manifests as a type error, because the
\texttt{Low} priority thread executing~\texttt{f} cannot signal a
high-priority condition variable.
Note that switching~\texttt{cv} to~\texttt{Low} priority would not fix
this, because then the \texttt{High}-priority ~\texttt{main} would
wait on a low-priority CV.


\begin{figure}
\small



\begin{minipage}[b]{0.49\textwidth}
\begin{lstlisting}
void <High>prod(CV<High> cv, Buffer *buf) {
  while(true) {
    ...
    append(buf, x);
    signal(cv);
  }
}

void <High>cons(CV<High> cv, Buffer *buf) {
  while(true) {
    if empty(buffer) {
      wait(cv);
    }
    x = pop(buf);
    ...
  }
}

void <Low>main() {
  CV<High> cv = new CV();
  Buffer *buf;

  spawn<High>(cons(cv, buf));
  spawn<High>(prod(cv, buf)); // ill-typed
}
\end{lstlisting}
\caption{Initial attempt at the producer-consumer program, which has
  a priority inversion (and a type error, in our full system).}
\label{fig:pc-terr}
\end{minipage}\hfill%
\begin{minipage}[b]{0.49\textwidth}
\begin{lstlisting}
void <High>prod(CV<Low> cv, Buffer *buf) {
  while(true) {
    ...
    append(buf, x);
    signal(cv);
  }
}

void <Low>main() {
  CV<Low> cv = new CV();
  Buffer *buf;

  spawn<High>(prod(cv, buf));
  CV<High> cv2 = promote<High>(cv);
  // Would now be ill-typed to spawn prod
  spawn<High>(cons(cv2, buf));
}
\end{lstlisting}
\caption{The producer-consumer example, without type errors or priority
  inversions.}
\label{fig:pc-fixed}
\end{minipage}
\end{figure}

\paragraph{A subtle priority inversion.}
Much of the complexity of the typing restrictions comes from a particular
pattern of priority inversion, which we illustrate with an instance
of the classic ``producer-consumer''
pattern.
%
Figure~\ref{fig:pc-terr} shows the pseudocode, where the \texttt{main}
thread runs at priority \texttt{Low}, and spawns a ``producer'' thread
that places messages into a buffer, and a ``consumer''
that removes messages from the buffer and processes them.
Both the producer and consumer run at priority \texttt{High}.
If the consumer gets ahead of the producer and
finds the buffer empty, we wish for it to sleep until the buffer is
non-empty.\footnote{For simplicity and to highlight the novel aspects
  of our type system, we assume the buffer is unbounded,
  so the producer never needs to block. Our system can also represent
  the classic case where the producer blocks when the buffer is
  full.}
To accomplish this, both threads also share a condition variable~\texttt{cv},
which must have priority~\texttt{High} because it is waited on by a
high-priority thread.
If the consumer finds the buffer empty, it waits on~\texttt{cv}.
The producer signals~\texttt{cv} when it adds an element to the buffer, thus
waking the consumer if it is blocked.

This code example satisfies conditions~1 and~2 governing the priorities of
threads and condition variables: in the example,
all relevant threads and condition variables are high priority.
But, this code nevertheless contains a subtle priority inversion.
Consider running the program by timesharing on a single processor.
When the main thread spawns the consumer, it will immediately begin running
(because it is high-priority), and will run until it finds the buffer empty
and waits on~\texttt{cv}.
At this point, the consumer, which is high-priority, is now waiting
not for the high-priority producer, but for the low-priority
\texttt{main} thread, which still needs to run to create the producer.
%
%
In our time-shared system, this wait could become effectively
unbounded if, for example, a long-running medium priority thread
(not shown in the pseudocode) preempts the low-priority \texttt{main} thread.

The problem is that the low-priority
\texttt{main} thread holds on to the high-priority~\texttt{cv}.
%
%
When this happens, it is possible for a thread waiting on the CV
to be blocked via the low-priority holding thread.
To prevent this, we add an additional restriction (which we will later
restate more formally)
governing the interactions between threads and condition variables:
\begin{enumerate}
  \setcounter{enumi}{2}
\item (Draft) If a thread has priority~$\prio$, then it
  {\em and its descendants} may not signal any condition variable that has
  priority greater than~$\prio$.
\end{enumerate}
In the example of Figure~\ref{fig:pc-terr}, the spawn of the producer
is ill-typed because the low-priority \texttt{main} thread attempts to
pass a high-priority condition variable to a thread that intends to
signal it (it will become clear later why the type error is on the
spawn rather than the signal).

It would seem that we are now stuck: the priority of the condition
variable must be high (because the high-priority consumer is waiting on
it), but then it cannot be passed by the low-priority main thread.
%
To make progress, we will allow threads to {\em promote} condition
variables by assigning them a higher priority under certain
conditions.
%
%
Note that when a low-priority thread promotes a condition variable, it
restricts its own power over that condition variable, e.g., it can no longer signal it.
Figure~\ref{fig:pc-fixed} shows how the~\texttt{main} thread can
promote the condition variable in the producer-consumer example.
The condition variable now starts at low priority and is passed as
such to the producer (this is acceptable, because the high-priority
producer may signal a low-priority condition variable).
After spawning the producer, the main thread promotes the condition
variable to high priority.
The~\texttt{promote} operation returns a new handle~\texttt{cv2} to the
same underlying condition variable as~\texttt{cv}, but which is typed at
priority~\texttt{High}.
This handle is passed to the consumer, which requires a CV at priority
\texttt{High}.
Note that the ordering of the spawns is important here: after
promoting the condition variable to high priority, {\em any} spawn of a
thread that signals the CV would be ill-typed due to restriction 3.
This ensures that no producer thread can be spawned at this point,
which prevents the priority inversion and type-checks.

\paragraph{The promote operation.}
Before discussing the full set of typing restrictions on the use and promotion
of condition variables, we give a high-level summary of the restrictions on
the condition variable API and common patterns for using it.

First, threads that signal a CV must be higher-priority than threads that
wait on the CV.
In many uses of CVs in our case studies, the same threads signal and wait
on CVs at different times (e.g., the same thread might act as a producer
and as a consumer at different times in a program).
In this case, the type system will enforce that all threads that use the
CV are at the same priority.
There is also a pattern in which the signaling threads are at a strictly
higher priority (e.g., a high-priority interaction thread in a server
application sends logging messages to a low-priority background thread in
a producer-consumer fashion).

The second point of concern is the order and conditions in which the
threads that use a CV are spawned.
This must be controlled to avoid the priority inversion of
Figure~\ref{fig:pc-terr}, in which a low-priority thread spawns the
consumer, which then blocks waiting for the producer to be
spawned.
If all of the threads that use the CV are at the same priority~$\prio$ and
all of these threads call both~\lstinline{wait} and~\lstinline{signal},
these threads should all be spawned by a thread at their
priority or higher (a common mistake would be to spawn them all from the
low-priority initial thread).
One solution that doesn't involve promotion is to spawn a temporary thread at
priority~$\prio$ (or higher) which spawns all of the threads that use the CV.
If the users of the CV are split between high-priority ``producers''
and low-priority ``consumers'', one can, as in Figure~\ref{fig:pc-fixed},
construct the CV at a lower priority, spawn the producers, then promote the
CV and spawn the consumers.

\paragraph{Typing restriction details.}
We now discuss the typing restrictions of the API at a more detailed level.
A full understanding of these details is not necessary for most users of
the condition variable API, but this discussion serves as a lead-in to the
formal type system.
Before proceeding, we must still discuss one additional important feature of
promotion.
Suppose that, in Figure~\ref{fig:pc-fixed}, the producer thread instead
ran at priority \texttt{Low}.
This still obeys all restrictions
discussed so far: the producer would signal a condition variable at
its own priority, and the main thread would still have the ability to pass
the condition variable to the producer.
Of course, this code would have a priority inversion, as the high-priority
consumer would be waiting on the low-priority producer.
To prevent this, we further fortify our restrictions and disallow
promoting a condition variable to a priority~$\prio$ if any
concurrently running thread may signal it at a priority lower
than~$\prio$.

To make this more formal, we equip condition variables with a notion of
{\em ownership} (e.g.,~\citep{crarywamo99, swm00, Cyclone02, Boyland03, afm07-l3}).
Each thread has some level of ownership of a condition variable at every
priority: conceptually, these levels are ``none'', ``owned'', and ``shared.''
It is an invariant that if one thread owns a condition variable at priority~$\prio$,
no other thread owns or shares it at~$\prio$ (though another thread might
own or share it at a different priority~$\prio'$).
Any number of threads can share a condition variable at a priority as long as no thread owns
it at that priority.
We can now re-state the restrictions governing condition variables and threads
in terms of ownership (restriction 1 is unchanged):
\begin{enumerate}
  \item To wait on a CV with priority $\prio$, a thread's
  priority must be less than or equal to $\prio$, and
\item To signal a CV with priority~$\prio$, a thread must own or
  share the CV at the thread's priority.
\item To pass {\em any} ownership of a CV at {\em any} priority
  to another thread, a thread must own or share the CV at
  its own priority.
\item No thread may own or share a CV at a priority lower than
  that CV's priority.
\item To promote a CV to a priority~$\prio$, a thread must
  {\em own} the CV at every priority lower than~$\prio$.
\end{enumerate}
A number of important facts follow directly from these restrictions.
First, these restrictions prevent the hypothetical priority inversion described
above: if the main thread promotes~\texttt{cv} from Low to High, it must
(by restriction 5)
own~\texttt{cv} at Low, which is impossible (by restriction 1 and the definition
of ownership) if the already-spawned low-priority producer signals it.
Second, the restated restriction 2, together with~4, implies the original
restriction~2: If a thread has a priority lower than that of the condition
variable, it cannot own or share it at the thread priority
and therefore cannot signal it.
Third, it is a direct result of restriction 4 that, after promoting a
condition variable to priority~$\prio$, the promoting thread must also give up
ownership at all lower priorities.
Finally, the formal statement of restriction 3 implies the informal version
from earlier because a thread cannot own or share a higher-priority condition
variable at the thread's own (lower) priority, and therefore cannot ``hold
onto'' it in a meaningful way, i.e., pass it to other threads that might
signal it.

\subsection{Mutexes}
Mutexes are used to enforce
mutual exclusion in critical sections of code, that is, ensure that
only one thread runs a critical section at a time.
The language we consider in this paper syntactically wraps critical sections
with the syntax~$\kwwithlock{v}{s}$.
If~$v$ is a mutex, this construct attempts to acquire the mutex; if successful,
it executes~$s$, ensuring that no other thread can acquire~$v$ during the
execution of~$s$.
When completed, it releases the mutex.
If the mutex referenced by~$v$ is already held by another thread, the
construct blocks until it can acquire the mutex and then proceeds to run~$s$.
\footnote{In many lower-level languages, this would be equivalent to
\texttt{lock v; s; unlock v}.
Our syntax ensures that locks and unlocks are paired and that the mutex is
always unlocked by the thread that locked it.
The syntax also allows a type checker to see the entire scope of a critical
section, which will be important for our type system.}

\begin{figure}
  \centering
  \begin{minipage}{0.33\textwidth}
    \begin{lstlisting}
Mutex mut;
void <Low> f() {
  with (mut) { ... }
}
\end{lstlisting}
  \end{minipage}\hfill%
  \begin{minipage}{0.33\textwidth}
    \begin{lstlisting}[firstnumber=last]
void <High> g() {
  with (mut) { ... }
}
\end{lstlisting}
  \end{minipage}\hfill%
  \begin{minipage}{0.33\textwidth}
    \begin{lstlisting}[firstnumber=last]
void <Medium> h() {
  while (1) { ... }
}
\end{lstlisting}
  \end{minipage}
  \caption{A priority inversion with three threads and a mutex.}
  \label{fig:mut-prio-inv}
\end{figure}

The classic case of a priority inversion with mutexes is illustrated in
Figure~\ref{fig:mut-prio-inv}.
Suppose the low priority thread runs first and acquires the mutex~\texttt{mut}.
The high priority thread then runs and attempts to acquire~\texttt{mut} but
blocks.
The medium thread then preempts the low priority thread and prevents it from
running for a long or indefinite period of time, thus delaying the high
priority thread, which is blocked on it.
(Note that, under the definition we use in this paper, a priority inversion
would still be present even in the absence of the medium priority thread;
without it, however, the impact on latency
would be bounded by the length of low's critical section.)

We could, as we did with condition variables, add static restrictions on the
use of mutexes to prevent priority inversions.
Such a restriction would require that any thread attempting to acquire a
given mutex~\texttt{mut} have a single priority~$\prio_{\mathtt{mut}}$ specific
to that mutex (to see why this is necessary, consider threads at two priorities
A and B potentially contending on a mutex: if the lower-priority thread is
successful, the higher-priority thread is waiting for it and there is a
priority inversion; absent any static guarantees on which thread will succeed,
A and B must be the same priority).
This is not as heavy-handed a restriction as it may seem at first
glance---a programmer could determine, for each mutex, the highest priority
at which a thread might try to acquire it (often called the {\em priority
  ceiling}), and any critical sections would run at the priority ceiling
for the mutex.\footnote{Our language does not allow for threads to change
  their own priority, but this can be effectively done by spawning a new
  thread to run the critical section and signal a condition variable, which is waited on
  by the main thread.}
This is similar to the approach suggested
by~\citet{lr-mesa-1980}.

Most implementations, however, take one of a number
of more efficient approaches to prevent priority inversions
dynamically and we model one such approach in this paper.
Like the solution proposed above, the approach we will refer to as the
{\em priority ceiling protocol}  associates with each
mutex a priority ceiling, that is, the highest priority of a thread that might
try to acquire it.
This priority is specified at the time the mutex is created.
A thread~$T$ at a lower priority may acquire the mutex, but if a higher priority
thread tries to acquire the mutex during~$T$'s critical section,~$T$ is
dynamically promoted to the priority ceiling for the remainder of its critical
section, thus preventing any threads at intermediate priorities from delaying
the high-priority thread.

\begin{figure}
  \begin{minipage}{0.49\textwidth}
  \begin{lstlisting}
Mutex<High> mut;
CV<Low> cv;
void <High>producer () {
  while(true) {
    ...
    with(mut) {
      append(buf, x)
      signal(cv)
    }
  }
}
\end{lstlisting}
  \end{minipage}\hfill%
  \begin{minipage}{0.49\textwidth}
  \begin{lstlisting}
void <Low>consumer () {
  while(true) {
    with(mut) {
      while (empty(buffer))
        wait(cv);  // Ill-typed
      x = pop(buf);
    }
    ...
  }
}
\end{lstlisting}
  \end{minipage}
  \caption{A more realistic version of the producer-consumer problem
    with mutexes and condition variables.}
  \label{fig:mut-cv}
\end{figure}

Although we use dynamic techniques to handle priority inversions involving
mutexes, the type system still comes into play in two ways.
First, to ensure soundness, the type system enforces that the priority of the
mutex is indeed its priority ceiling, that is, that every thread using the
mutex is at a lower priority than the mutex's priority.
The second type system modification is necessary
in order to properly handle code that uses both mutexes and condition variables,
and we demonstrate its necessity by way of another example.
Figure~\ref{fig:mut-cv} shows a more realistic implementation of the
producer-consumer code (this time with the producer higher-priority than the
consumer) which uses the mutex~\texttt{mut} both to guard queue
operations, which may not be atomic, and ensure that the queue is still
nonempty when the consumer calls~\texttt{pop} (this is necessary if there are
multiple consumers).\footnote{POSIX CVs would also require
  that~\texttt{mut} is passed to~\texttt{wait}, which would release the mutex
  while the thread blocks on~\texttt{cv} and reacquire it when awoken by a
  signal. This behavior causes no issues for safety---indeed, it is safe for
  threads to temporarily release locks at any point within \texttt{with}
  blocks, as long as the \texttt{with} block contains all points at which
  the lock may be held. Because this issue is therefore orthogonal to the
  soundness of the type system, we omit it from the model for simplicity.}
Because the producer has priority~\texttt{High}, the priority of~\texttt{mut}
must be~\texttt{High}.
However, this means that the consumer's critical section, which ordinarily
runs at~\texttt{Low}, may be dynamically raised to~\texttt{High} if the
producer attempts to acquire the mutex during the critical section.
If this occurs, the consumer may then call~\texttt{wait} on the low-priority
condition variable~\texttt{cv} while running at high priority; this is a
violation of type safety.

In order to prevent the type safety violation described above, we require that
every critical section (the~$s$ of~$\kwwithlock{v}{s}$) type-check at both
its normal priority and the priority ceiling of the mutex~$v$,
as the critical section may run at either or both of these
priorities.
In Figure~\ref{fig:mut-cv}, this means that the consumer code
is ill-typed because the critical section cannot check at~\texttt{High}
due to the~\texttt{wait(cv)}.
(The code could be made type-correct by changing the
priority of~\texttt{cv} to~\texttt{High}.)
In summary, we add two
restrictions to the list (in addition to the restrictions above for
condition variables):
\begin{enumerate}
  \setcounter{enumi}{5}
\item Any critical section for a mutex at priority~$\prio$ must run at
  a priority less than or equal to~$\prio$
  (i.e.,~$\prio$ is the priority ceiling of the mutex).
\item Any critical section for a mutex at priority~$\prio$ must type check
  at both its own priority and~$\prio$.
\end{enumerate}

\secput{lang}{Type System for Responsiveness}

In this section, we formalize the ideas of Section~\ref{sec:overview} by
presenting a core calculus {\langname} with threading, condition variables
and mutexes.
%

\subsection{Syntax}

\begin{figure}
\langfigsize
\[
\begin{array}{l l l}
  \prio \in \prios = \mathit{Priorities} & a, b, c \in \mathit{Threads}\\
  \cvname \in \mathit{Condition Variables} &
  \mutname \in \mathit{Mutexes} &
  \cvmutname \bnfdef \cvname \bnfalt \mutname
  \end{array}
\]
\[
\begin{array}{l l r l }


  \mathit{Types} & \tau & \bnfdef &
  \kwunit
  \bnfalt  \kwnat
  \bnfalt \kwreft{\tau}
  \bnfalt \kwcvt{\cvname}{\prio}
  \bnfalt \kwmutext{\prio}
\\

  \mathit{Values} & v & \bnfdef &
  x \bnfalt \kwtriv
  \bnfalt \kwnumeral{n}
  \bnfalt \kwcv{\cvname}
  \bnfalt \kwmutex{\mutname}
  \bnfalt \kwref{\kwassn}
  \\
  
  \mathit{Instructions} & \instr & \bnfdef &
  v
  \bnfalt \kwspawn{\prio}{\tau}{\vec{\cvname}}{s}
  \bnfalt \kwnewref{\tau}{v}
  \bnfalt \kwderef{v}
  \bnfalt  \kwassign{v}{v}
  \\
  & & \bnfalt & \kwnewcv{\prio}
  \bnfalt \kwwait{v}
  \bnfalt \kwsignal{v}
  \bnfalt \kwpromote{v}{\prio}
  \bnfalt \kwnewmutex{\prio}

  \\

  \mathit{Statements} & s & \bnfdef &
  \kwlet{x}{\instr}{\block}
  \bnfalt \kwwithlock{v}{\block}
  \bnfalt \kwif{v}{\block}{\block}
  \bnfalt \kwwhile{v}{\block}
  \bnfalt s; s
  \bnfalt \kwskip


  
\end{array}
\]
\caption{Syntax of {\calcname}}
\label{fig:syn}
\end{figure}

Figure~\ref{fig:syn} presents the syntax of {\langname}.
Priorities~$\prio$ are drawn from a totally ordered set~$\prios$
with total order~$\preceq$.
Both the set and the order are fixed for the
duration of the program but may be arbitrary.
We use metavariables~$a, b, c$, and variants to denote threads.
Each condition variable and mutex has a unique such name which
we will use in the semantics to track acquisition, signaling, etc.
We use~$\cvname$ and variants as the names of condition variables,
$\mutname$ for mutexes and~$\cvmutname$ for either a condition variable or
a mutex.
Base types are~$\kwunit$ and natural numbers.
We also have the type~$\kwreft{\tau}$ of mutable references to values of
type~$\tau$.
There are also types for handles to condition variables and mutexes.
Both types indicate the priority of the handle (for mutexes, this is the
priority of the mutex; condition variables may have several handles at
different priorities).
The type of condition variable handles also contains the name of the
condition variable with which the handle is associated.

The rest of the language is in three levels: values~$v$, instructions~$\instr$
and statements~$s$.
Values consist of variables as well as other expressions that do not evaluate:
the unit value; natural numbers; and handles to condition variables, mutexes,
and references.
Instructions are in ``2/3-cps'' form: they contain only values
as subcomponents, simplifying the presentation of the semantics.
Complex expressions can be built by let-binding intermediate results, so
this causes no loss of generality.
Instructions perform a single operation on values and return a new value.
%
%
The instruction~$\kwspawn{\prio}{\tau}{}{s}$ spawns a new thread at
priority~$\prio$ to run the statement~$s$.
Instructions for manipulating global state are~$\kwnewref{\tau}{v}$
which creates a new reference of type~$\tau$, initialized to~$v$;
$\kwderef{v}$, which gets the value of the reference~$v$;
and~$\kwassign{v_1}{v_2}$ which assigns~$v_2$ to the reference~$v_1$.
The instruction~$\kwnewcv{\prio}$ creates a new condition variable and
returns a handle to it at priority~$\prio$; this handle can then be used to
wait, signal, or promote the condition variable.

Statements handle control flow.
The statement $\kwlet{x}{\instr}{s}$
binds~$x$ to the result of instruction~$\instr$
and proceeds with~$s$.
Otherwise, statements contain only values and substatements.
The mutex-guarded critical section~$\kwwithlock{v}{s}$ is also a statement.
Statements also consist of conditionals, while loops, concatenation of
statements, and the empty statement~$\kwskip$.

\subsection{Static Semantics}
\iffull
\begin{figure}
\langfigsize
    \centering
  \def \MathparLineskip {\lineskip=\mylineskip}
\begin{mathpar}

\Rule{var}
{
  \strut
}
{
  \etyped{\sig}{\ctx, \hastype{x}{\tau}}{x}{\tau}
}
\and
\Rule{\kwunit I}
{
\strut
}
{
\etyped{\sig}{\ctx}{\kwtriv}{\kwunit}
}
\and
\Rule{\kwnat I}
{
\strut
}
{
\etyped{\sig}{\ctx}{\kwn}{\kwnat}
}
\and
\Rule{RefVal}
     {
       \strut
     }
     {
       \etyped{\sig,\sigrtype{\kwassn}{\tau}}{\ctx}{\kwref{\kwassn}}
              {\kwreft{\tau}}
     }
\and
\Rule{Mutex}
     {
       \strut
     }
     {
       \etyped{\sig, \sigcvtype{\mutname}{\prio}}{\ctx}
              {\kwmutex{\mutname}}{\kwmutext{\prio}}
     }
\and
\Rule{CV}
     {
     }
     {
       \etyped{\sig,\sigcvtype{\cvname}{\prio}}{\ctx}{\kwcv{\cvname}}
              {\kwcvt{\cvname}{\prio}}
     }
\end{mathpar}
\caption{Value typing rules.}
\label{fig:val-statics}
\end{figure}
\fi
\begin{figure*}
\langfigsize
\centering
\def \MathparLineskip {\lineskip=\mylineskip}
  \begin{mathpar}
\Rule{Spawn}
{
  \styped{\sig}{\ctx}{\fctx}{s}{\prio'}{\fctx''}\\
  \forall \cvname, \prio''. \permof{\fctx}{\cvname}{\prio''} \neq \fnone
  \rightarrow
  \permof{\fctx}{\cvname}{\prio} \neq \fnone
}
{
  \ityped{\sig}{\ctx}{\fctx', \fctx}
  {\kwspawn{\prio'}{\tau}{\vec{\cvname}}{s}}
  {\kwunit}{\prio}
  {\fctx'}
}
\and
\iffull
\Rule{NewRef}
     {
       \etyped{\sig}{\ctx}{v}{\tau}
     }
     {
       \ityped{\sig}{\ctx}{\fctx}{\kwnewref{\tau}{v}}{\kwreft{\tau}}{\prio}{\fctx}
     }
\and
\Rule{Deref}
     {
       \etyped{\sig}{\ctx}{v}{\kwreft{\tau}}
     }
     {
       \ityped{\sig}{\ctx}{\fctx}{\kwderef{v}}{\tau}{\prio}{\fctx}
     }
\and
\Rule{Update}
     {
       \etyped{\sig}{\ctx}{v_1}{\kwreft{\tau}}\\
       \etyped{\sig}{\ctx}{v_2}{\tau}
     }
     {
       \ityped{\sig}{\ctx}{\fctx}{\kwassign{v_1}{v_2}}{\kwunit}{\prio}{\fctx}
     }
\and
\fi     
\Rule{Wait}
{
  \etyped{\sig}{\ctx}{v}{\kwcvt{\cvname}{\prio'}}\\
  \meetc[\worlds]{\ctx}{\ple{\prio}{\prio'}}
}
{
  \ityped{\sig}{\ctx}{\fctx}{\kwwait{v}}{\kwunit}{\prio}{\fctx}
}
\and
\Rule{Signal}
{
  \etyped{\sig}{\ctx}{v}{\kwcvt{\cvname}{\prio'}}\\
  \permof{\fctx}{\cvname}{\prio} \neq \fnone
}
{
  \ityped{\sig}{\ctx}{\fctx}
  {\kwsignal{v}}{\kwunit}
  {\prio}{\fctx}
}
\and
\Rule{Promote}
{
  \etyped{\sig}{\ctx}{v}{\kwcvt{\cvname}{\prio_1}}\\
  \meetc[\worlds]{\ctx}{\ple{\prio_1}{\prio_2}}\\
  \forall \prio_0. (\ple{\prio_1}{\prio_0} \land \nple{\prio_2}{\prio_0}) \rightarrow
  \permof{\fctx}{\cvname}{\prio_0} = \fowned\\
  \forall \prio_0. \nple{\prio_2}{\prio_0} \rightarrow
  \permof{\fctx'}{\cvname}{\prio_0} = \fnone\\
  \forall \prio_0. \ple{\prio_2}{\prio_0} \rightarrow
  \permof{\fctx'}{\cvname}{\prio_0} = \permof{\fctx}{\cvname}{\prio_0}\\
  \forall \cvname' \neq \cvname, \prio_0.
  \permof{\fctx'}{\cvname'}{\prio_0} = \permof{\fctx}{\cvname'}{\prio_0}\\
}
{
  \ityped{\sig}{\ctx}{\fctx}
  {\kwpromote{v}{\prio_2}}
  {\kwcvt{\cvname}{\prio_2}}
  {\prio}
  {\fctx'}
}
\and
\Rule{NewCV}
{
  \ple{\prio'}{\prio}\\
  \cvname\fresh\\
  \forall \prio_0, \ple{\prio'}{\prio_0}. \permof{\fctx'}{\cvname}{\prio_0} = \fowned\\
  \forall \prio_0, \nple{\prio'}{\prio_0}. \permof{\fctx'}{\cvname}{\prio_0} = \fnone\\
  \forall \cvname' \neq \cvname. \forall \prio_0.
  \permof{\fctx'}{\cvname'}{\prio_0} =
  \permof{\fctx}{\cvname'}{\prio_0}
}
{
  \ityped{\sig}{\ctx}{\fctx}{\kwnewcv{\prio'}}{\kwcvt{\cvname}{\prio'}}
  {\prio}{\fctx'}
}
\iffull
\and
\Rule{NewMutex}
     {
       \strut}
     {
       \ityped{\sig}{\ctx}{\fctx}{\kwnewmutex{\prio'}}{\kwmutext{\prio'}}
              {\prio}{\fctx}
     }
\fi
\end{mathpar}

\caption{Instruction typing rules.}
\label{fig:inst-statics}
\end{figure*}

\begin{figure*}
\langfigsize
\centering
\def \MathparLineskip {\lineskip=\mylineskip}
  \begin{mathpar}
     \Rule{Let}
         {\ityped{\sig}{\ctx}{\fctx}{\instr}{\tau}{\prio}{\fctx'}\\
           \styped{\sig}{\ctx, \hastype{x}{\tau}}{\fctx'}{s}{\prio}{\fctx''}}
         {\styped{\sig}{\ctx}{\fctx}{\kwlet{x}{\instr}{s}}{\prio}{\fctx''}}
    \and
    \Rule{Skip}
         {\strut}
         {\styped{\sig}{\ctx}{\fctx}{\kwskip}{\prio}{\fctx}}
    \and
    \Rule{WithLock}
     {
       \etyped{\sig}{\ctx}{v}{\kwmutext{\prio'}}\\
       \styped{\sig}{\ctx}{\fctx}{s}{\prio'}{\fctx'}\\
       \styped{\sig}{\ctx}{\fctx}{s}{\prio}{\fctx'}\\
       \meetc{\ctx}{\ple{\prio}{\prio'}}
     }
     {
       \styped{\sig}{\ctx}{\fctx}{\kwwithlock{v}{s}}{\prio}{\fctx'}
     }
     \and
    \Rule{If}
         {\etyped{\sig}{\ctx}{v}{\kwnat}\\
           \styped{\sig}{\ctx}{\fctx}{s_1}{\prio}{\fctx'}\\
           \styped{\sig}{\ctx}{\fctx}{s_2}{\prio}{\fctx'}}
         {\styped{\sig}{\ctx}{\fctx}{\kwif{v}{s_1}{s_2}}{\prio}{\fctx'}}
    \and
    \Rule{While}
         {\etyped{\sig}{\ctx}{v}{\kwnat}\\
           \styped{\sig}{\ctx}{\fctx}{s}{\prio}{\fctx}}
         {\styped{\sig}{\ctx}{\fctx}{\kwwhile{v}{s}}{\prio}{\fctx}}
    \and
    \Rule{Seq}
         {\styped{\sig}{\ctx}{\fctx}{s_1}{\prio}{\fctx'}\\
           \styped{\sig}{\ctx}{\fctx'}{s_2}{\prio}{\fctx''}\\
         }
         {
           \styped{\sig}{\ctx}{\fctx}{s_1; s_2}{\prio}{\fctx''}
         }
  \end{mathpar}
  \caption{Statement typing rules.}
  \label{fig:stmt-statics}
\end{figure*}


We now present the type system for {\langname}, which, together with the
dynamic priority ceiling mechanism for mutexes, statically ensures that
programs are free of priority inversions.
%
%
The typing judgment for values is of the
form~$\etyped[\prios]{\sig}{\ctx}{v}{\tau}$.
The judgment has two parameters:
$\prios$ is the totally ordered set of priorities.
The signature~$\sig$ maps reference cells to the types they contain
(these entries will be written~$\sigrtype{\kwassn}{\tau}$),
and mutexes and condition variables to their priorities
($\sigcvtype{\cvname}{\prio}$ and $\sigcvtype{\mutname}{\prio}$).
There will be at most one entry in a signature for each reference cell and
mutex.
Condition variables, however, can appear in a signature multiple times with
multiple priorities, allowing different handles to the condition variable to
have different priorities.
The context~$\ctx$, as usual, maps variables to types.
\iffull
The rules for the judgment, in Figure~\ref{fig:val-statics}, are relatively
straightforward.
\fi

The typing judgment for instructions,
$\ityped[\prios]{\sig}{\ctx}{\fctx}{\instr}{\tau}{\prio}{\fctx'}$
is more complex.
It uses~$\prios$,~$\sig$, and~$\ctx$ as before.
The typing of instructions and statements, unlike values, depends on the
priority at which the instruction is run; this priority is indicated
as~$\prio$ in the judgment.
The judgment also tracks ownership of condition variables which, as motivated
in Section~\ref{sec:overview}, can be in one of three states:
$\fnone$, $\fshared$, and~$\fowned$.
The mapping~$\fctx$ maps pairs of condition variable names and
priorities to a ``permission'' level,
written~$\permof{\fctx}{\cvname}{\prio} = \fracp \in \{\fnone, \fshared, \fowned\}$.
Instructions may produce and consume permissions, so the
judgment contains the mapping before~($\fctx$) and after~($\fctx'$)
the instruction.
Finally, the judgment indicates that the instruction produces a value of
type~$\tau$.

In the~\rulename{Spawn} rule, $\fctx',\fctx$ means that permissions are
split between~$\fctx'$ and~$\fctx$.
To split permission levels, we define the
judgment~$\splitsto{\fracp_1}{\fracp_2}{\fracp_3}$,
defined in in Figure~\ref{fig:split-perm}, which indicates that the
permission in~$\fracp_1$ is split between~$\fracp_2$ and~$\fracp_3$.
Formally, if~$\fctx'' = \fctx',\fctx$,
then for all~$\cvname$ and~$\prio$,
$\splitsto{\permof{\fctx''}{\cvname}{\prio}}
{\permof{\fctx'}{\cvname}{\prio}}{\permof{\fctx}{\cvname}{\prio}}$.
The current thread keeps~$\fctx'$ and~$\fctx$ is passed to the new thread.
The second premise requires that if~$\fctx$ shares or owns
a condition variable~$\cvname$ at any priority, the spawning thread shares
or owns it at its own priority.
This enforces restriction~3 of the list in Section~\ref{sec:overview}.

\begin{figure}
\[
\begin{array}{c c c}
\splitsto{\fowned}{\fowned}{\fnone} &
\splitsto{\fowned}{\fnone}{\fowned} &
\splitsto{\fowned}{\fshared}{\fshared}\\
\splitsto{\fshared}{\fnone}{\fshared} &
\splitsto{\fshared}{\fshared}{\fnone} &
\splitsto{\fshared}{\fshared}{\fshared}\\
& \splitsto{\fnone}{\fnone}{\fnone}
\end{array}
\]
\caption{Rules for splitting permission levels.}
\label{fig:split-perm}
\end{figure}

\iffull
The rules~\rulename{NewRef}, \rulename{Deref}, and \rulename{Update} do not
interact with priorities or permissions, but assign types to
operations on references in standard ways.
\fi
The~\rulename{Wait} operation requires that the subexpression have the
type of a handle at priority~$\prio'$ to the condition variable~$\cvname$.
As motivated earlier (restriction 1), waiting requires that the priority of
the current thread is lower than the priority of the handle.
Waiting does not require any ownership of the condition variable.
The~\rulename{Signal} rule requires that the current thread own
or share the CV~$\cvname$ (restriction 2).
Recall that signaling also requires that the thread's priority be greater than
or equal to
that of the condition variable, but requiring non-none permission
already ensures this, as there is no permission at lower priorities
(restriction 4).

The~\rulename{Promote} rule gets the CV name and priority from the type of the
CV handle, and checks that the current priority of the handle is lower
than~$\prio_2$, the priority to which it is being promoted.
As motivated earlier (restriction 5), promotion requires~$\fowned$ permission
at all
priorities not greater than~$\prio_2$; this is checked by the third premise.
The remaining premises define~$\fctx'$, the returned permission
mapping.
Permission for~$\cvname$ at all priorities not greater than~$\prio_2$ is removed
(to preserve restriction 4), otherwise permissions are preserved.
Finally,~\rulename{NewCV} creates a new condition variable name for the
returned handle and returns a new context~$\fctx'$ with ownership of the
new condition variable initialized to~$\fowned$ or~$\fnone$ at appropriate priorities.

%
%

The judgment for statements
is~$\styped[\prios]{\sig}{\ctx}{\fctx}{s}{\prio}{\fctx'}$, and is
largely the same as that of instructions except that statements do not
return a value.
The only rule that directly interacts with priorities is~\rulename{WithLock}.
This rule ensures that the critical section is well-typed at both the
thread's current priority and the priority ceiling of the mutex (because the
critical section may be raised to this priority at runtime): this corresponds
to restriction~7 of Section~\ref{sec:overview}.
Regardless of priority, the critical section must be typable with the same
contexts~$\fctx$ and~$\fctx'$, which are threaded through.
The final premise enforces restriction~6, that is, that the priority of the
current thread is less than or equal to the mutex's priority ceiling.
The remaining rules thread priority and ownership through substatements.
The~\rulename{Let} rule also adds the bound variable~$x$ to the context
with its appropriate type when typing the statement~$s$.
Note that in~\rulename{If}, both branches must have the same input and output
permissions.
The permissions for the body of a while loop must be invariant, as enforced
by the~\rulename{While} rule; this is, however, fairly permissive due to our
coarse-grained permissions.
For example, $\fshared$ permissions can be split arbitrarily within a loop
as long as they remain~$\fshared$.


\iffull
\section{A DAG Model for Responsive Synchronization}\label{sec:dag}

In this section, we set up the formalization of priority inversions that we
will use to prove the correctness of the type
system.
The formalization proceeds as follows: we develop a model for
representing the parallelism and synchronization of {\calcname} programs using
a graph.
Prior work~\citep{mah-responsive-2017, mah-priorities-2018, MullerSiGoAcAgLe20} has
established results about the efficient schedulability of programs
represented by such graphs, provided the graphs are ``well-formed'',
which corresponds to the absence of priority inversions.
We present a new generalization of these techniques and apply it to {\calcname},
including a definition of well-formedness for programs with mutexes and
CVs.
It is this definition that we will use to show that well-typed {\calcname}
programs are free of priority inversions.
We begin by reviewing preliminaries related to DAG-based parallelism
and scheduling models and proceed to extend these results to {\calcname}.

\subsection{Preliminaries: DAG Models}
We model executions of parallel programs using
Directed Acyclic Graphs, or DAGs, in which
vertices represent units of computation (to simplify later
concepts, we will assume without loss of generality that each vertex represents
a single, consistent unit of computation time, like a processor cycle)
and edges represent the dependences between portions of the program.
If~$u$ is an ancestor of~$u'$ in the DAG, notated~$\anc{u}{u'}$,
then either~$u = u'$ or~$u$ must be executed before~$u'$.
If~$\nanc{u}{u'}$ and~$\nanc{u'}{u}$, that is, there is no path of dependences
between the two vertices in the graph, then the computations represented
by~$u$ and~$u'$ may be run in parallel.

Because {\calcname} allows arbitrary
synchronization, it is possible for programs to contain
{\em deadlocks}, where two or more threads depend on each other in a cyclic
fashion; such a condition will manifest as a cycle in the dependence graph.
However, because little of interest can be said about the running time of
programs with deadlocks, our results will focus on non-deadlocking programs,
whose graphs are acyclic.
Thus, we will use the terms ``graph'' and ``DAG'' somewhat interchangeably.

A DAG represents a particular execution of
a parallel program.
In a nondeterministic program, the pattern of synchronizations
may depend on many factors outside our reasoning and so
a particular program may give rise to many possible DAGs, one for each
possible execution.


\paragraph{DAG Notation and Terminology.}
We will use the metavariable~$\graph$ (and variants thereof) to denote graphs
and will notate a graph as a
quadruple~$\dagq{\gthreads}{\spawns}{\syncs}{\reads}$.
In this notation,~$\gthreads$ is a set of {\em threads} and the other three
components are sets of inter-thread edges.
As in {\langname}, we name threads using~$a, b, c,$ and variants.
We write~$\gthread{a}{\prio}{u_1 \tscomp u_2 \tscomp \dots \tscomp u_n}$
for a single thread~$a$ of priority~$\prio$ with vertices~$u_1, \dots, u_n$.
The edges~$(u_1, u_2), (u_2, u_3), \dots, (u_{n-1}, u_n)$ are
implied; we will refer to these as {\em thread} edges.
We will often use metavariables~$s$ and~$t$ for the first and last vertices,
respectively, of a thread, but otherwise use~$u$ and variants for vertices.
The sets~$\spawns$,~$\syncs$ and~$\reads$
contain {\em \termspawn}, {\em \termsync} and {\em weak} edges, respectively.
{\termSpawn} edges are of the form~$(u, a)$, where~$u$ is a vertex and~$a$ is a
thread created or spawned off by the computation~$u$: as an edge, this may
be seen as an edge from~$u$ to the first vertex of thread~$a$.
Edges in~$\syncs$ and~$\reads$ are of the form~$(u_1, u_2)$, representing
an edge from~$u_1$ to~$u_2$.
{\termSync} edges represent synchronization between threads.

Weak edges~\citep{MullerSiGoAcAgLe20} capture happens-before relationships
between computations that occur at runtime but are not the result of
explicit synchronization.
As an example, weak edges were originally used to track happens-before
relationships induced by global memory: if the computation~$u$
writes a value into memory and that value is read by the computation~$u'$,
there may be a weak edge~$(u, u')$.
A path, or sequence of edges, is {\em strong} if all of its edges are
strong (i.e., not weak).
If~$\anc{u}{u'}$ and all paths from~$u$ to~$u'$ are strong, then~$u$ is
a {\em strong ancestor} of~$u'$, written~$\sanc{u}{u'}$.
If some path from~$u$ to~$u'$ contains a weak edge, then~$u$ is a
{\em weak ancestor} of~$u'$, written~$\wanc{u}{u'}$.
We will continue to write~$\anc{u}{u'}$ if it is not important whether the
paths are weak or strong.

\paragraph{Schedules and Response Time.}
Given a set of processors and a series of time steps,
a {\em schedule} of a DAG is an assignment of vertices to processors
at each time step; this corresponds to executing the program represented by
the DAG on a parallel machine.
The schedule must respect the dependences indicated by the DAG.
A vertex is {\em ready} if all of its strong ancestors have already been
executed; a schedule may only assign ready vertices.
A schedule is {\em admissible} if it also reflects the happens-before
relationships indicated by weak edges;
in an admissible schedule, if a vertex is ready, all of its weak parents
(that is, the sources of its incoming weak edges) have already been executed.

%
We wish to find a schedule of a DAG that minimizes the {\em response time}
of a particular thread~$\gthread{a}{\prio}{s\tscomp \dots \tscomp t}$,
which we define to be the number of steps (inclusive) from when~$s$ becomes
ready to when~$t$ is executed.
%
%
Results such as Brent's Theorem~\citep{brent74} and its descendants
give approximately-optimal bounds for schedules meeting certain requirements.
The requirement we will use is that the schedule is {\em prompt}~\citep{mah-priorities-2018, mah-responsive-2017}.
At any time step, a prompt schedule first assigns all ready nodes at the
highest priority, then the next highest and so on.
%
%
Processors are only left idle if no ready vertices remain.

\paragraph{Work and Span.}
Bounds on schedules are typically stated in terms of two quantities which can
be computed from the DAG: the {\em work} is the total amount of computation,
 and the {\em span} is a measure of the
critical-path length.
We bound response time in terms of quantities closely related to these.
The {\em competitor work}~$\prioworkof{\compwork{}{a}}{\psnlt{\prio}}$
of a thread~$a$ of priority~$\prio$ is the number of vertices (amount of work)
at priority not
less than~$\prio$ which may be executed in parallel with vertices of~$a$.
These are, informally, the vertices with which vertices of thread~$a$ may
have to ``compete'' for processor time.
Competitor work is defined formally as follows:
\[
\prioworkof{\compwork{}{a}}{\psnlt{\prio}} \defeq
|\{u \in \graph \mid \nanc{u}{s} \land \nanc{t}{u} \land
\pnlt{\uprio{\graph}{u}}{\prio}\}|\]

We also define a metric analogous to span, but this definition deviates from
prior work, and so we defer it to later in the section.

\subsection{Graph Models for Mutexes and Condition Variables}

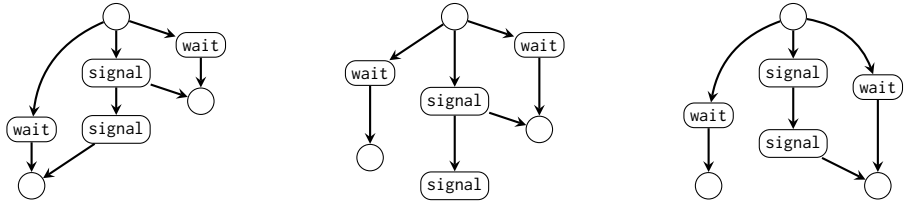
\begin{figure}
  \footnotesize
  \begin{tikzpicture}[scale=0.75]
    \tikzstyle{node}=[draw,rectangle,rounded corners,minimum width=14pt];
    \tikzstyle{enode}=[draw,circle,minimum width=10pt];
    \tikzstyle{edge}=[draw,thick,-stealth];

    \node[enode] (0) at (0, 0) {};
    \node[node] (w1) at (1.5, -0.5) {\texttt{wait}};
    \node[enode] (n1) at (1.5, -1.5) {};
    \node[node] (s1) at (0, -1) {\texttt{signal}};
    \node[node] (s2) at (0, -2) {\texttt{signal}};
    \node[node] (w2) at (-1.5, -2) {\texttt{wait}};
    \node[enode] (n2) at (-1.5, -3) {};
    
    \path[edge] (0)--(w1);
    \path[edge] (0) to [bend right] (w2);
    \path[edge] (0)--(s1);
    \path[edge] (w1)--(n1);
    \path[edge] (w2)--(n2);
    \path[edge] (s1)--(s2);
    \path[edge] (s1)--(n1);
    \path[edge] (s2)--(n2);

    \node[enode] (1) at (6, 0) {};
    \node[node] (1w1) at (7.5, -0.5) {\texttt{wait}};
    \node[enode] (1n1) at (7.5, -2) {};
    \node[node] (1s1) at (6, -1.5) {\texttt{signal}};
    \node[node] (1s2) at (6, -3) {\texttt{signal}};
    \node[node] (1w2) at (4.5, -1) {\texttt{wait}};
    \node[enode] (1n2) at (4.5, -2.5) {};
    
    \path[edge] (1)--(1w1);
    \path[edge] (1)--(1w2);
    \path[edge] (1)--(1s1);
    \path[edge] (1w1)--(1n1);
    \path[edge] (1w2)--(1n2);
    \path[edge] (1s1)--(1s2);
    \path[edge] (1s1)--(1n1);

    \node[enode] (2) at (12, 0) {};
    \node[node] (2w1) at (13.5, -1.25) {\texttt{wait}};
    \node[enode] (2n1) at (13.5, -3) {};
    \node[node] (2s1) at (12, -1) {\texttt{signal}};
    \node[node] (2s2) at (12, -2.25) {\texttt{signal}};
    \node[node] (2w2) at (10.5, -1.75) {\texttt{wait}};
    \node[enode] (2n2) at (10.5, -3) {};
    
    \path[edge] (2) to [bend left] (2w1);
    \path[edge] (2) to [bend right] (2w2);
    \path[edge] (2)--(2s1);
    \path[edge] (2w1)--(2n1);
    \path[edge] (2w2)--(2n2);
    \path[edge] (2s1)--(2s2);
    \path[edge] (2s2)--(2n1);
  \end{tikzpicture}
  \caption{Several graphs representing the same program with condition variables.
    Vertices are ordered vertically in order of execution.}
  \label{fig:cv-dags}
\end{figure}

We now describe, at a high level, how we use the graph representation of
the previous subsection to represent programs with mutexes and condition
variables.
The representation will be made formal in Section~\ref{sec:proof}.
%
Because both mutexes and condition variables involve synchronization using
first-class data, rather than control flow, it is not possible to represent
a given piece of code using one definite graph.
%
Instead, a graph will represent one particular execution of a program and is
necessarily dependent on scheduling decisions made at runtime.

\begin{figure}
  \begin{minipage}{0.27\textwidth}
    \footnotesize
    \centering
  \begin{tikzpicture}[scale=0.6]
    \tikzstyle{node}=[draw,rectangle,rounded corners,minimum width=14pt];
    \tikzstyle{enode}=[draw,circle,minimum width=10pt];
    \tikzstyle{edge}=[draw,thick,-stealth];
    \tikzstyle{wedge}=[draw,dashed,thick,-stealth];

    \node[enode] (0) at (0, -1) {};
    \node[node] (l1) at (0, -2) {\texttt{lock1}};
    \node[node] (l2) at (2, -3.5) {\texttt{lock2}};
    \node[node] (s1) at (0, -3) {$s$};
    \node[node] (s2) at (2, -5) {$s$};
    \node[node] (u1) at (0, -4) {\texttt{unlock1}};
    \node[node] (u2) at (2, -6) {\texttt{unlock2}};
    
    \path[edge] (0)--(l1);
    \draw[edge] (0) to [bend left] (l2);
    \path[edge] (l1)--(s1);
    \path[edge] (s1)--(u1);
    \path[edge] (l2)--(s2);
    \path[edge] (s2)--(u2);
    \path[edge] (u1)--(s2);
    \path[wedge] (l1)--(s2);
  \end{tikzpicture}
  \caption{Two threads contending on a lock.}
  \label{fig:mutex-dag}
  \end{minipage}\hfill%
  \begin{minipage}{0.35\textwidth}
    \footnotesize
    \centering
    \begin{tikzpicture}[scale=0.6]
    \tikzstyle{node}=[draw,rectangle,rounded corners,minimum width=14pt];
    \tikzstyle{enode}=[draw,circle,minimum width=10pt];
    \tikzstyle{edge}=[draw,thick,-stealth];
    \tikzstyle{wedge}=[draw,dashed,thick,-stealth];

    \node[enode] (0) at (0, -1) {};
    \node[node] (l1) at (0, -2) {\texttt{lock1}};
    \node[node] (l2) at (4, -4) {\texttt{lock2}};
    \node[node] (s1) at (0, -3) {$s_1$};
    \node[node] (s1') at (1.5, -4) {$s_2$};
    \node[node] (s2) at (4, -5) {$s$};
    \node[node] (u1) at (1.5, -5) {\texttt{unlock1}};
    \node[enode] (end1) at (0, -6) {};
    \node[node] (u2) at (4, -6) {\texttt{unlock2}};
    
    \path[edge] (0)--(l1);
    \draw[edge] (0) to [bend left] (l2);
    \path[edge] (l1)--(s1);
    \path[edge] (s1)--(s1');
    \path[edge] (s1')--(u1);
    \path[edge] (u1)--(end1);
    \path[edge] (l2)--(s2);
    \path[edge] (s2)--(u2);
    \path[edge] (u1)--(s2);
    \path[edge] (s1)--(end1);
    \path[wedge] (l1)--(s2);
    \path[wedge] (s1')--(s2);
  \end{tikzpicture}
    \caption{Thread 1 is promoted to the priority ceiling.}
  \label{fig:mutex-dag-ceil}
  \end{minipage}\hfill%
  \begin{minipage}{0.25\textwidth}
    \footnotesize
    \centering
  \begin{tikzpicture}[scale=0.6]
    \tikzstyle{node}=[draw,rectangle,rounded corners,minimum width=14pt];
    \tikzstyle{enode}=[draw,circle,minimum width=10pt];
    \tikzstyle{edge}=[draw,thick,-stealth];
    \tikzstyle{wedge}=[draw,dashed,thick,-stealth];

    \node[enode] (0) at (0, -1) {};
    \node[node] (l1) at (0, -2) {\texttt{lock1}};
    \node[node] (l2) at (2, -3.5) {\texttt{lock2}};
    \node[node] (s1) at (0, -3) {$s$};
    \node[node] (s2) at (2, -5) {$s$};
    \node[node] (u1) at (0, -4) {\texttt{unlock1}};
    \node[node] (u2) at (2, -6) {\texttt{unlock2}};
    
    \path[edge] (0)--(l1);
    \draw[edge] (0) to [bend left] (l2);
    \path[edge] (s1)--(u1);
    \path[edge] (l2)--(s2);
    \path[edge] (s2)--(u2);
    \path[edge] (u1)--(s2);
    \path[edge] (l2)--(s1);
  \end{tikzpicture}
  \caption{The strengthening of Figure~\ref{fig:mutex-dag}.}
  \label{fig:mutex-dag-strong}
  \end{minipage}
\end{figure}

This dynamicity can be seen clearly by considering how we would model waiting
on, and signaling, a condition variable.
Suppose (the instruction represented by) vertex~$u$ signals a condition
variable, unblocking some thread: let~$u'$ represent the vertex in the
newly-unblocked thread that becomes ready as a result of the signal.
We represent the dependence between the signal and wait operations by
adding the sync edge~$(u, u')$ to~$\syncs$ in the graph.
At runtime, it is easy to determine what threads are blocked on a particular
condition variable, and we will do so when formalizing the operational
semantics.
However, it is not possible in general to determine this statically, because
which~\texttt{wait} operations have run before a given~\texttt{signal}
operation depends on the precise runtime interleaving of the threads involved,
and can be different from one run of a program to another.
As an illustration, Figure~\ref{fig:cv-dags} shows three possible DAGs
that might arise from the program
\begin{lstlisting}
  CV cv; spawn {wait(cv)}; spawn {wait(cv)}; signal(cv); signal(cv);
\end{lstlisting}

The graph representations of programs with mutexes similarly depend on the
order in which threads acquire mutexes.
If a vertex~$u$ releases a
mutex and unblocks vertex~$u'$, this induces a sync edge~$(u, u')$.
In addition, we add a weak edge from the vertex that successfully acquires
the mutex to~$u'$, indicating the happens-before relation that the
successful acquire, necessarily, ran before the unsuccessful thread.
Such a DAG is illustrated in Figure~\ref{fig:mutex-dag}, where~\texttt{lock}
indicates the beginning of a~$\kwwithlock{v}{s}$ critical section
and~\texttt{unlock} indicates the end of a critical section.

The model can also represent programs in which
priority inversions are prevented using dynamic mechanisms such as
priority inheritance or priority ceiling; we will present a graph
for priority ceiling since that is the mechanism we focus on in this
paper, but other mechanisms can be represented in a similar way.
Figure~\ref{fig:mutex-dag-ceil} represents the same scenario as above but
now supposes that the thread that acquires the mutex is
running at a priority lower than that of both the second thread and the
priority ceiling, so that when~\texttt{lock2} runs, the first thread is
promoted to the priority ceiling of the mutex.
We represent this by spawning another thread (at the
priority ceiling) to complete the first thread's critical section
(represented by~$s_2$ in the figure).
Another weak edge is added from this thread to the (still unsuccessful) thread.
When the critical section finishes, control of the first thread returns to
the original thread at the original priority, represented by a sync edge
from the release (\texttt{unlock1}) operation to the next operation in the
original thread.

\subsection{Response Time and Synchronization}

The goal of this section is to bound the response time of a thread~$a$ in an
admissible, prompt schedule of a
graph of a program with mutexes and condition variables.
Below, we will use~$s$ to denote the first vertex of~$a$ and~$t$ to denote
its last vertex.
Intuitively, the bound should only depend on the amount of work at priorities
not less than the priority of~$a$; otherwise, the schedule is not properly
prioritizing high-priority work.
The competitor work, defined above, already only includes such work.
We must therefore ensure that our notion of span, corresponding to the
critical path of~$a$, also includes only high-priority computation.
To this end, we restrict our attention to {\em well-formed}
DAGs~\citep{mah-priorities-2018, MullerSiGoAcAgLe20},
in which lower-priority work does not fall on the
critical path of higher-priority threads.
This corresponds to ruling out priority inversions.

In the presence of weak edges, it is possible for lower-priority work that does
not actually fall on the critical path (because of happens-before relations)
to \textit{appear} to fall on the critical path, according to straightforward
definitions of well-formedness.
As an example, in Figure~\ref{fig:mutex-dag}, the vertex~\texttt{lock1}
appears to be on the critical path of the right thread: it 
is an ancestor of~\texttt{unlock2} (the last vertex of the right thread)
and not an ancestor of~\texttt{lock2} (the first vertex of the right thread).
However,~\texttt{lock2} will never have to wait for~\texttt{lock1} to complete
(which is, for all practical purposes, the definition of ``being on the
critical path'') because the weak edge indicates that~\texttt{lock1} has
already completed when~\texttt{lock2} becomes ready.
In order to allow such DAGs, \citet{MullerSiGoAcAgLe20} introduced
{\em strengthening}, to transform a graph with weak edges into
one in which all vertices actually on the critical path of thread~$a$
appear as strong ancestors of~$t$.
If~$\strengthen{a}{\graph}$ denotes the strengthened graph,
the so-called {\em $a$-span}~$\longsp{\compwork{}{a}}{a}$ is then the
longest strong path in~$\strengthen{a}{\graph}$ ending at~$t$
consisting of vertices that are not ancestors of~$s$.
This path is the critical path of strong ancestors that may have not been
executed before~$s$ is ready but must be completed before~$a$ can complete.

We follow a similar approach, but note that synchronization primitives
yield a different set of DAGs than prior work.
%
As such, the DAGs in this paper will not necessarily meet the definition of
well-formedness from that work even though these DAGs still do not contain
priority inversions.

We observe that, in order to use the~$a$-span of a thread in its response time
bound, it suffices for the strengthening of a graph to meet two conditions:
\begin{itemize}
\item The strengthening may not remove dependencies that were present in the
  original graph:
  i.e., if, at any step of an admissible schedule, a vertex~$u$ is
  ready in~$\strengthen{a}{\graph}$, then it is ready in~$\graph$.
\item The~$a$-span (or in general, any strong path ending at~$t$
  in~$\strengthen{a}{\graph}$ consisting of non-ancestors of~$s$)
  must not contain work lower-priority than~$a$.
\end{itemize}
%
Different programming models may give rise to different types of DAGs,
which require different strengthenings to achieve the above
conditions.
In turn, the original graph will have to meet different well-formedness
conditions in order for the strengthening to meet the two conditions above.
Any DAG model using weak edges can thus define its own notion of strengthening
and corresponding notion of well-formedness.
An appropriate bound on response time 
will hold for graphs in this model as long as the definitions of
strengthening and well-formedness meet the conditions above.
Theorem~\ref{thm:brent-weak-gen} formalizes this general result.
The intuition behind the bound is that, at every time step, by the definition
of a prompt schedule, we either complete~$P$ units of competitor work or (if
not enough high-priority work is ready) reduce the length of the~$a$-span
by~1.

\begin{theorem}\label{thm:brent-weak-gen}
  Let~$\graph$ be a graph and~$\strengthen{a}{\graph}$
  be its $a$-strengthening for a
  thread~$\gthread{a}{\prio}{s \tscomp \dots \tscomp t} \in \graph$.
  Assume the following two properties hold of~$\graph$
  and~$\strengthen{a}{\graph}$.
  \begin{enumerate}
  \item If, at any step of an admissible schedule, a vertex~$u$ is
    ready in~$\strengthen{a}{\graph}$, then it is ready in~$\graph$.
  \item
    If there exists a strong path from~$u$ to~$t$ in~$\strengthen{a}{\graph}$
    and~$\nanc{u}{s}$ in~$\strengthen{a}{\graph}$,
    then~$\ple{\prio}{\uprio{\graph}{u}}$.
  \end{enumerate}
  Then, for any admissible prompt schedule of~$\graph$ on~$P$ cores,
  \[
  \resptimeof{a} \leq 
  \frac{1}{P}\left[
  \prioworkof{\compwork{}{a}}{\psnlt{\prio}} +
      (P-1)\longsp{\compwork{}{a}}{a}\right]
      \]
\end{theorem}
\begin{proof}
  The proof is effectively identical to that of Theorem 2.1
  of~\citep{MullerSiGoAcAgLe20}, which assumes only these two properties of
  the graph and its strengthening.
\end{proof}

We now present the definitions of well-formedness
(Definition~\ref{def:wellformed}) and strengthening
(Definition~\ref{def:strengthen}) we will use in this paper, and show
that they meet the above conditions, guaranteeing that the bound holds.
Informally, in addition to the expected restrictions on priorities,
we require that any synchronization edge that
results from thread~$a$ waiting on a lock held by thread~$b$ is
preceded by a weak edge indicating that thread~$b$ acquired the lock before
thread~$a$ attempted to acquire it.


    
\begin{defn}\label{def:wellformed}
A DAG~$\graph = \dagq{\gthreads}{\spawns}{\syncs}{\reads}$ is well-formed if
for all
threads~$\gthread{a}{\prio}{s \tscomp{} \dots \tscomp{} t} \in \gthreads$,
\begin{enumerate}
\item If~$\sanc{u}{t}$ and $\nanc{u}{s}$ then~$\ple{\prio}{\uprio{\graph}{u}}$.
\item If~$\sanc{u}{t}$ and~$\nanc{u}{s}$
  and there is a strong edge~$(u', u)$, where~$\wanc{u'}{t}$,
  then there is a weak edge~$(u', u'')$ where~$\sanc{u''}{t}$
  and~$\nanc{u''}{s}$ 
  and~$u''$ is not the first vertex in its thread.
  Furthermore, the edge from the parent of~$u''$ is the only strong edge
  to~$u''$.
  %
\end{enumerate}
\end{defn}

Because of the weak edge whose existence is guaranteed by well-formedness,
ancestors of thread~$b$ are not on the
critical path of thread~$a$.
The strengthening makes this explicit by removing ancestors of the vertex~$u$
in thread~$b$.
In order to not leave it disconnected, we make~$u$ a strong child of the vertex
in thread~$a$ \textit{immediately before} the vertex that tried to attempt the
lock; we will call this vertex~$u''_p$.
Doing so does not increase the bound; by the weak edge,~$u$'s parent
must execute before the child of~$u''_p$.
From the perspective of the response time bound, the worst case is then
that~$u$ is not ready until immediately after~$u''_p$ executes, which is
made explicit by the added edge.
As an example, Figure~\ref{fig:mutex-dag-strong} shows the strengthening of the
graph in Figure~\ref{fig:mutex-dag}.
Note that~\texttt{lock1} no longer appears on the second thread's critical
path, but the schedule is the same.

\begin{defn}\label{def:strengthen}
  Let~$\graph$ be a well-formed graph with a
  thread~$\gthread{a}{\prio}{s \tscomp{} \dots \tscomp{} t}$.
  We define the~$a$-strengthening as follows.
  For every strong edge~$(u', u)$ where~$\wanc{u'}{t}$ and~$\sanc{u}{t}$
  and~$\nanc{u}{s}$:
  \begin{enumerate}
  \item Remove the weak edge~$(u', u'')$ that exists by
    Definition~\ref{def:wellformed}.
  \item Add a strong edge~$(u''_p, u)$, where~$u''_p$ is the parent of~$u''$,
    which exists by Definition~\ref{def:wellformed}.
  \item Remove the strong edge~$(u', u)$.
  \end{enumerate}
  
\end{defn}

\begin{theorem}\label{thm:brent}
  Let~$\graph$ be a well-formed graph by Definition~$\ref{def:wellformed}$
  and let~$\gthread{a}{\prio}{s \tscomp \dots \tscomp t} \in \graph$.
  For any admissible prompt schedule of~$\graph$ on~$P$ cores,
  \[
  \resptimeof{a} \leq 
  \frac{1}{P}\left[
  \prioworkof{\compwork{}{a}}{\psnlt{\prio}} +
      (P-1)\longsp{\compwork{}{a}}{a}\right]
      \]
\end{theorem}
\begin{proof}
  We need only prove the two conditions of Theorem~\ref{thm:brent-weak-gen}.
  \iffull
  \begin{enumerate}
  \item Let~$u$ be ready in~$\strengthen{a}{\graph}$.
    By definition, all ancestors of~$u$ have been executed.
    We must show that all strong ancestors of~$u$ in~$\graph$ have been executed.
    The ancestors of~$u$ are unchanged between~$\strengthen{a}{\graph}$
    and~$\graph$ unless~$u$ is~$u''$ or~$u$ of
    Definition~\ref{def:strengthen}.
    If~$u$ is~$u''$ of that definition, there is only an additional weak
    edge in~$\graph$, which does not prevent~$u$ from being ready.
    If~$u$ is~$u$ of that definition, then by construction, the only
    strong parent of~$u$ in~$\strengthen{a}{\graph}$ is~$u''_p$, so this
    vertex has been executed.
    The parent of~$u$ in~$\graph$ is~$u'$ of Definition~\ref{def:strengthen}.
    But because there is a weak edge from~$u'$ to~$u''$ and 
    the only strong parent of~$u''$ has been executed,~$u'$ must also have
    been executed in an admissible schedule.
    Therefore,~$u$ is ready in~$\graph$.
  \item
    Let~$u$ be such that there is a strong path from~$u$ to~$t$
    in~$\strengthen{a}{\graph}$ and~$\nanc{u}{s}$.
    We note that~$\sanc{u}{t}$ in~$\strengthen{a}{\graph}$ because
    if~$\wanc{u}{t}$,
    then there is a strong edge~$(u, u')$ and a strong path from~$u'$ to~$t$,
    so the edge~$(u, u')$ would have been removed in~$\strengthen{a}{\graph}$,
    a contradiction.
    Furthermore, if~$\sanc{u}{t}$ in~$\strengthen{a}{\graph}$,
    then~$\ple{\prio}{\uprio{\graph}{u}}$.
    To see this, it suffices to show that if~$\sanc{u}{t}$
    in~$\strengthen{a}{\graph}$, then~$\sanc{u}{t}$ in~$\graph$.
    Consider a path from~$u$ to~$t$ in~$\strengthen{a}{\graph}$ and
    consider an edge~$(u_1, u_2$ in this path that is not in~$\graph$.
    This must be the edge~$(u_p'', u)$ of Definition~\ref{def:strengthen}.
    But by Definition~\ref{def:wellformed}, we have~$\sanc{u''}{t}$
    in~$\graph$ and therefore~$\sanc{u_p''}{t}$ in~$\graph$.
  \end{enumerate}
  \fi
\end{proof}

\section{Correctness of Priority Inversion Type System}\label{sec:proof}

In this section, we prove that the type system of Section~\ref{sec:lang}
guarantees the absence of priority inversions.
Specifically, we show that a graph corresponding to a well-typed {\calcname}
program is {\em well-formed} by Definition~\ref{def:wellformed}
and can therefore be responsively scheduled by Theorem~\ref{thm:brent}.
We first present a {\em cost semantics} for {\calcname} that evaluates the
program and produces a graph corresponding to the execution.
Next, we present the main result.
Finally, in Section~\ref{sec:extensions}, we briefly discuss several extensions
to {\calcname} that model other common use cases of synchronization.

\subsection{Cost Semantics}

\iffull
\begin{figure}
\langfigsize
\[
\begin{array}{l l r l }
  \mathit{Frames} & \stframe & \bnfdef &
  \shole; s
  \bnfalt \kwwithlocks{v}{\shole}
  \bnfalt \kwwithlocksp{v}{\shole}{a}{\prio}
  \bnfalt \kwlet{x}{\shole}{\block}\\

  \mathit{Stacks} & \stack & \bnfdef & 
  \estack \bnfalt \scp{\stack}{\stframe}\\

  \mathit{States} & \stackstate & \bnfdef &
  \ssend{\stack}{\instr}
  \bnfalt \sreturn{\stack}{v}
  \bnfalt \scsend{\stack}{s}
  \bnfalt \screturn{\stack}
\end{array}
\]
\caption{Stack syntax}
\label{fig:stack-syn}
\end{figure}

\begin{figure}
\langfigsize
  \centering
  \def \MathparLineskip {\lineskip=\mylineskip}
\begin{mathpar}
\Rule{KS-Empty}
     {\prio \in \prios'}
     {\stackacceptsc{\sig}{\estack}{\fctx}{\prios'}{\fctx}{\prio}}
\and
\Rule{KS-Seq}
     {\stackacceptsc{\sig}{\stack}{\fctx_2}{\prios_1}{\fctx_3}{\prio_2}\\
       \forall \prio \in \prios_1.
       \styped{\sig}{\ectx}{\fctx_1}{s}{\prio}{\fctx_2}
     }
     {
       \stackacceptsc{\sig}{\scp{\stack}{\shole; s}}
                     {\fctx_1}{\prios_1}{\fctx_3}{\prio_2}
     }
\and
\Rule{KS-WithLockS}
     {\stackacceptsc{\sig}{\stack}{\fctx_1}{\prios_1}{\fctx_2}{\prio_2}\\
       \etyped{\sig}{\ectx}{v}{\kwmutext{\prio'}}
     }
     {
       \stackacceptsc{\sig}{\scp{\stack}{\kwwithlocks{v}{\shole}}}
                     {\fctx_1}{\prios_1 \cup \{\prio'\}}{\fctx_2}{\prio_2}
     }
\and
\Rule{KS-WithLockSP}
     {\stackacceptsc{\sig}{\stack}{\fctx_1}{\prios_1}{\fctx_2}{\prio}\\
       \etyped{\sig}{\ectx}{v}{\kwmutext{\prio'}}
     }
     {
       \stackacceptsc{\sig}
                     {\scp{\stack}{\kwwithlocksp{v}{\shole}{a}{\prio}}}
                     {\fctx_1}{\prios_1 \cup \{\prio'\}}{\fctx_2}
                     {\prio'}
     }
\and   
\Rule{KS-Let}
     {\stackacceptsc{\sig}{\stack}{\fctx_2}{\prios_1}{\fctx_3}{\prio_2}\\
       \forall \prio \in \prios_1.
       \styped{\sig}{\hastype{x}{\tau}}{\fctx_1}{s}{\prio}{\fctx_2}
     }
     {
       \stackaccepts{\sig}{\scp{\stack}{\kwlet{x}{\shole}{s}}}
                    {\tau}{\fctx_1}{\prios_1}{\fctx_3}{\prio_2}
     }
\end{mathpar}
\\[4ex]
\begin{mathpar}
\Rule{KS-PopInstr}
     {
       \stackaccepts{\sig}{\stack}{\tau}{\fctx_2}{\prios_1}{\fctx_2}{\prio_3}\\
       \forall \prio \in \prios_1.
       \ityped{\sig}{\ectx}{\fctx_1}{\instr}{\tau}{\prio}{\fctx_2}
     }
     {
       \sstyped{\sig}{\fctx_1}{\ssend{\stack}{\instr}}{\prio_2}{\fctx_3}
     }
\and
\Rule{KS-PopStmt}
     {
       \stackacceptsc{\sig}{\stack}{\fctx_2}{\prios_1}{\fctx_3}{\prio_2}\\
       \forall \prio \in \prios_1.
       \styped{\sig}{\ectx}{\fctx_1}{s}{\prio}{\fctx_2}
     }
     {
       \sstyped{\sig}{\fctx_1}{\scsend{\stack}{s}}{\prio_2}{\fctx_3}
     }
\and
\Rule{KS-PushInstr}
     {
       \stackaccepts{\sig}{\stack}{\tau}{\fctx_1}{\prios_1}{\fctx_2}{\prio_2}\\
       \etyped{\sig}{\ectx}{v}{\tau}
     }
     {
       \sstyped{\sig}{\fctx_1}{\sreturn{\stack}{v}}{\prio_2}{\fctx_2}
     }
\and
\Rule{KS-PushStmt}
     {
       \stackacceptsc{\sig}{\stack}{\fctx_1}{\prios_1}{\fctx_2}{\prio_2}
     }
     {
       \sstyped{\sig}{\fctx_1}{\screturn{\stack}}{\prio_2}{\fctx_2}
     }
\end{mathpar}
\caption{Typing rules for stacks and stack states.}
\label{fig:stack-statics}
\end{figure}

\fi

The dynamic semantics of {\langname} evaluate a program, at the same time
producing a graph representing the parallel execution.
For nondeterministic programs, the graph will reflect the
``choices'' made in the run of the program taken, but will also
indicate where parallelism was available and thus admit other schedules
consistent with the same ``choices''.
%
%
Constructing a parallel operational semantics that fully simulates a prompt
and admissible
schedule of the program is outside the scope of this paper; such a semantics
would follow the approach in prior work~\citep{MullerSiGoAcAgLe20}.

The dynamic semantics is specified using control stacks.
%
\iffull
The syntax of stacks is shown in Figure~\ref{fig:stack-syn}.
A stack consists of {\em frames}, each of which is a statement with a
``hole''~$\shole$, reflecting the continuation of the program.
Because of the use of ``2/3-cps'' form in instructions, frames consist only
of statements and not instructions or expressions.
Two frames,~$\kwwithlocks{v}{\shole}$
and~$\kwwithlocksp{v}{\shole}{a}{\prio}$ represent states of critical sections
that arise at runtime and will be described later.
\else
A stack~$\stack$ consists of {\em frames}, each of which is a statement with a
``hole''~$\shole$, reflecting the continuation of the program.
Because of the use of ``2/3-cps'' form in instructions, frames consist only
of statements and not instructions or expressions.
\fi
The empty stack is denoted~$\estack$.
The current execution state of a thread is represented by a stack
state~$\stackstate$;
each state contains the stack representing the current continuation.
The state~$\ssend{\stack}{\instr}$ indicates that the thread is
executing the instruction~$\instr$; $\sreturn{\stack}{v}$ indicates that
an instruction has just returned
the value~$v$; $\scsend{\stack}{s}$ indicates that the thread is
executing the statement~$s$; and~$\screturn{\stack}$ indicates that a
statement has just returned
(no return value is necessary as statements do not return values).
\iffull\else
The syntax of stacks and stack states is presented in the
supplementary material.
\fi

%
\iffull
The judgment~$\stackacceptsc{\sig}{\stack}{\fctx}{\prios_1}{\fctx'}{\prio_2}$
indicates that~$\stack$ requires the permissions~$\fctx$ to continue, and
will leave over~$\fctx'$.
It also requires that the current statement be well-typed at all priorities
in the set~$\prios_1$, and indicates that the entire stack can run at~$\prio_2$.
It is an invariant that~$\prio_2 \in \prios_1$.
The
judgment~$\stackaccepts{\sig}{\stack}{\tau}{\fctx}{\prios_1}{\fctx'}{\prio_2}$
is similar, but for a stack that is expecting a return value of type~$\tau$
from an instruction.
The rules for these judgments are shown in the top of
Figure~\ref{fig:stack-statics}.
Stack states are typed with the
judgment~$\sstyped{\sig}{\fctx}{\stackstate}{\prio}{\fctx'}$,
which indicates that~$\stackstate$ is well-typed at priority~$\prio$
with starting permissions~$\fctx$ and leaves over permissions~$\fctx'$.
The rules for this judgment,
in the bottom of Figure~\ref{fig:stack-statics},
simply match up the requirements of the stack, indicated by the stack typing
judgments, with the statement, instruction, or return value.
\else
Stack states are typed with the
judgment~$\sstyped{\sig}{\fctx}{\stackstate}{\prio}{\fctx'}$,
which indicates that~$\stackstate$ is well-typed at priority~$\prio$
with starting permissions~$\fctx$ and leaves over permissions~$\fctx'$.
This judgment depends on an auxiliary typing judgment for stacks, which
indicates the permissions needed to execute the continuation represented
by the stack, as well as the permissions left over after the program.
The full static semantics are detailed in the supplementary material.
\fi

The dynamic semantics operates over {\em thread pools}, which are collections
of ready threads.
A thread~$a$ running the stack state~$\stackstate$ at priority~$\prio$ with
signature~$\sig$ is written~$\cthread{a}{\prio}{\sig}{\stackstate}$.
Thread pools are composed with the operator~$\tpcp$.
A {\em configuration}~$\twlconfig{\tp}{\mem}{\waiters}{\locked}$ represents
a complete snapshot of a running program.
It includes a thread pool~$\tp$, as well as a memory~$\mem$,
which is a mapping from reference cells, written~$\kwassn$, to the values
they contain.
The configuration contains two additional components:
$\waiters$ is a mapping from mutexes and condition variables to threads waiting
on the mutex or CV,
and~$\locked$ is a mapping from mutexes to either the thread currently holding
the mutex or~$\lunlocked$ indicating the mutex is unlocked.
For a mutex~$\mutname$, an element of~$\waiters(\mutname)$
has the form~$(a, (u_1, u_2), \prio, \sig, \stackstate)$, indicating a
thread~$a$ at priority~$\prio$ in state~$\stackstate$ with signature~$\prio$.
The two vertices~$u_1$ and~$u_2$ represent the mutex acquisition and the
vertex following the mutex acquisition, respectively; these vertices are used
to add appropriate sync and weak edges to the graph later.
An element of~$\waiters(\cvname)$ for a CV~$\cvname$ is similar:
in the pair of vertices~$(u_1, u_2)$,~$u_1$ represents the~\texttt{wait}
operation and~$u_2$ represents the following vertex.
An element of~$\locked(\mutname)$ is of the form~$(a, u_1, u_2)$,
indicating a thread~$a$ holds the lock; as above,~$u_1$ and~$u_2$ represent
the acquire operation and the following vertex, respectively.

We present one additional typing rule for typing configurations:
\begin{mathpar}
\Rule{Global}
     {
     \bigcup_{\cvmutname \in \mathit{Dom}(\waiters)}\waiters(\cvmutname) =
     \{(a_{n+1}, u_{n+1}, \prio_{n + 1}, \sig_{n+1}, s_{n+1}), \dots,
     (a_m, u_m, \prio_m, \sig_m, s_m)\}\\
     \sigcvtype{\cvname}{\prio} \in \sig_1, \dots, \sig_n
     \rightarrow \forall i \in [1, m].
     \forall \plt{\prio'}{\prio}. \permof{\fctx_i}{\cvname}{\prio'} = \fnone\\
        \forall i \in [1, m]. \sstyped{\sig, \sig_i}{\fctx_i}{\stackstate_i}{\prio_i}{\fctx_i'}\\
        \fctxwf{\fctx_0, \dots, \fctx_m}\\
        \fctxwf{\fctx'_0, \dots, \fctx'_m}\\
        \leq 1\\
        \forall i \in [1, m]. \mtyped{\sig_i}{\mem}
     }
     {
        \cfgtyped{\sig}
        {\twlconfig{\cthread{a_1}{\prio_1}{\sig_1}{\stackstate_1} \tpcp \dots \tpcp
        \cthread{a_n}{\prio_n}{\sig_n}{\stackstate_n}}{\mem}{\waiters}{\locked}}
     }
\end{mathpar}
The rule requires each thread in the thread pool
(as well as threads blocked on a mutex or CV, which are not included in the
thread pool)
to be well-typed with its signature at its
priority and appropriate permission mappings.
%
%
It also requires that permissions are ``splittable'':
we define~$\fctxwf{\fctx_0, \dots, \fctx_m}$ to mean that
for all~$\cvname$ and all~$\prio$, if~$\permof{\fctx_i}{\cvname}{\prio} =
\fowned$, then for all~$j\neq i$, we have~$\permof{\fctx_j}{\cvname}{\prio}
= \fnone$.
Finally, the memory must be well-typed with respect to every thread's
signatures.
This is indicated with the judgment~$\mtyped{\sig}{\mem}$,
which requires that for all~$\sigrtype{\kwassn}{\tau} \in \sig$,
we have~$\mem(\kwassn) = v$ and~$\etyped{\sig}{\ectx}{v}{\tau}$.

\begin{figure}
  \footnotesize
  
  \centering
\begin{mathpar}
      \Rule{Spawn}
           {u, b\fresh}
           {
      \lconfig{\cthread{a}{\prio}{\sig}{\ssend{\stack}{
            \kwspawn{\prio'}{\tau}{\vec{\cvname}}{s'}}}
        \tpcp \tp}
              {\mem}
              {\waiters}
              {\locked}
              {\graph}
      \mstep
              \lconfig{\cthread{a}{\prio}{\sig}{\sreturn{\stack}{\kwtriv}} \tpcp
                \cthread{b}{\prio}{\sig}{\ssend{\estack}{s'}}
                \tp}
              {\mem}
              {\waiters}
              {\locked}
              {\graph \scomp{a} \gnode{u}{\sig} \cup
                \{\edgespawn{u}{b}\}}
           }
           \iffull
      \and
      \Rule{NewRef}
           {u, \kwassn\fresh}
           {
      \lconfig{\cthread{a}{\prio}{\sig}{\ssend{\stack}{
            \kwnewref{\tau}{y}}{s}}
        \tpcp \tp}
              {\mem}
              {\waiters}
              {\locked}
              {\graph}
      \mstep
      \lconfig{\cthread{a}{\prio}{\sig,\sigrtype{\kwassn}{\tau}}
        {\sreturn{\stack}{\kwref{\kwassn}}} \tpcp \tp}
              {\memupd{\mem}{\mement{\kwassn}{y}{u}{\sig}}}
              {\waiters}
              {\locked}
              {\graph \scomp{a} \gnode{u}{\sig}}
           }
      \and
      \Rule{Deref}
           {u\fresh\\
             \mem(\kwassn) = \memrent{v}{u'}{\sig'}
           }
           {
      \lconfig{\cthread{a}{\prio}{\sig}{\ssend{\stack}{
            \kwderef{\kwref{\kwassn}}}}
        \tpcp \tp}
              {\mem}
              {\waiters}
              {\locked}
              {\graph}
      \mstep
      \lconfig{\cthread{a}{\prio}{\sig, \sig'}
        {\sreturn{\stack}{v}} \tpcp \tp}
              {\mem}
              {\waiters}
              {\locked}
              {\graph \scomp{a} \gnode{u}{\sig}\}}
           }
      \and
      \Rule{Update}
           {u\fresh}
           {
      \lconfig{\cthread{a}{\prio}{\sig}{\ssend{\stack}{
            \kwassign{\kwref{\kwassn}}{y}}}
        \tpcp \tp}
              {\mem}
              {\waiters}
              {\locked}
              {\graph}
      \mstep
      \lconfig{\cthread{a}{\prio}{\sig}
        {\sreturn{\stack}{\kwtriv}} \tpcp \tp}
              {\memupd{\mem}{\mement{\kwassn}{y}{u}{\sig}}}
              {\waiters}
              {\locked}
              {\graph \scomp{a} \gnode{u}{\sig}}
           }
      \and
      \Rule{NewCV}
           {u, \cvname\fresh}
           {
      \lconfig{\cthread{a}{\prio}{\sig}
                {\ssend{\stack}{\kwnewcv{\prio'}}}
              \tpcp \tp}
              {\mem}
              {\waiters}
              {\locked}
              {\graph}
      \mstep
      \lconfig{\cthread{a}{\prio}{\sig,\sigcvtype{\cvname}{\prio'}}
        {\sreturn{\stack}{\kwcv{\cvname}}}
        \tpcp \tp}
              {\mem}
              {\waiters}
              {\locked}
              {\graph \scomp{a} \gnode{u}{\sig}}
           }
           \fi
      \and
      \Rule{Wait}
           {u_1, u_2\fresh}
           {
      \lconfig{\cthread{a}{\prio}{\sig}
                {\ssend{\stack}{\kwwait{\kwcv{\cvname}}}}
              \tpcp \tp}
              {\mem}
              {\waiters}
              {\locked}
              {\graph}
      \mstep
      \lconfig{\tp}
              {\mem}
              {\waitupd{\waiters}{\cvname}
                {\waitcons{(a, (u_1, u_2), \prio, \sig, \sreturn{\stack}{\kwtriv})}{\waiters(\cvname)}}}
              {\locked}
              {\graph \scomp{a} \gnode{u}{\sig}}
           }
      \and
      \Rule{Signal1}
           {u\fresh\\
             \waiters(\cvname) = \waitcons{(b, (\_, u'), \prio_b, \sig_b, \stackstate_b)}{\ell}
           }
           {
      \lconfig{\cthread{a}{\prio}{\sig}{\ssend{\stack}{
                  {\kwsignal{\kwcv{\cvname}}}}}
              \tpcp \tp}
              {\mem}
              {\waiters}
              {\locked}
              {\graph}\\
      \mstep
      \lconfig{\cthread{a}{\prio}{\sig}{\sreturn{\stack}{\kwtriv}} \tpcp
        \cthread{b}{\prio_b}{\sig_b}{\stackstate_b} \tpcp
        \tp}
              {\mem}
              {\waitupd{\waiters}{\cvname}{\ell}}
              {\locked}
              {\graph \scomp{a} \gnode{u}{\sig}  \cup \{\edgesync{u}{u'}\}}
           }
      \and
      \Rule{Signal2}
           {u\fresh\\ \waiters(\cvname) = \waitemp}
           {
      \lconfig{\cthread{a}{\prio}{\sig}{\ssend{\stack}{
                  {\kwsignal{\kwcv{\cvname}}}}}
              \tpcp \tp}
              {\mem}
              {\waiters}
              {\locked}
              {\graph}
      \mstep
      \lconfig{\cthread{a}{\prio}{\sig}{\sreturn{\stack}{\kwtriv}} \tpcp
        \tp}
              {\mem}
              {\waiters}
              {\locked}
              {\graph \scomp{a} \gnode{u}{\sig}}
           }
      \and
      \iffull
      \Rule{Promote}
           {u\fresh}
           {
      \lconfig{\cthread{a}{\prio}{\sig}
                {\ssend{\stack}{\kwpromote{\kwcv{\cvname}}{\prio'}}}
              \tpcp \tp}
              {\mem}
              {\waiters}
              {\locked}
              {\graph}
      \mstep
      \lconfig{\cthread{a}{\prio}{\sig,\sigcvtype{\cvname}{\prio'}}
        {\sreturn{\stack}{\kwcv{\cvname}}}
        \tpcp \tp}
              {\mem}
              {\waiters}
              {\locked}
              {\graph \scomp{a} \gnode{u}{\sig}}
           }
      \and
      \Rule{NewMutex}
           {u, \mutname\fresh}
           {
      \lconfig{\cthread{a}{\prio}{\sig}
                {\ssend{\stack}{\kwnewmutex{\prio'}}}
              \tpcp \tp}
              {\mem}
              {\waiters}
              {\locked}
              {\graph}
      \mstep
      \lconfig{\cthread{a}{\prio}{\sig, \sigcvtype{\mutname}{\prio'}}
        {\sreturn{\stack}{\kwmutex{\mutname}}}
        \tpcp \tp}
              {\mem}
              {\waiters}
              {\locked}
              {\graph \scomp{a} \gnode{u}{\sig}}
           }
      \and
      \Rule{Let1}
           {\strut}
           {
             \lconfig{\cthread{a}{\prio}{\sig}
               {\scsend{\stack}{\kwlet{x}{\instr}{s}}} \tpcp \tp}
                     {\mem}
                     {\waiters}
                     {\locked}
                     {\graph}
             \mstep
             \lconfig{\cthread{a}{\prio}{\sig}
               {\ssend{\scp{\stack}{\kwlet{x}{\shole}{s}}}{\instr}} \tpcp \tp}
                     {\mem}
                     {\waiters}
                     {\locked}
                     {\graph}
           }
       \and
      \Rule{Let2}
           {\strut}
           {
             \lconfig{\cthread{a}{\prio}{\sig}
               {\sreturn{\scp{\stack}{\kwlet{x}{\shole}{s}}}{v}} \tpcp \tp}
                     {\mem}
                     {\waiters}
                     {\locked}
                     {\graph}
             \mstep
             \lconfig{\cthread{a}{\prio}{\sig}
               {\scsend{\stack}{[v/x]s}} \tpcp \tp}
                     {\mem}
                     {\waiters}
                     {\locked}
                     {\graph}
           }
\end{mathpar}
\caption{Dynamic semantics rules (part 1)}
\label{fig:cost1}
\end{figure}
\begin{figure}
  \footnotesize
  
  \centering
  \begin{mathpar}
    \fi
      \Rule{WithLockS1}
           {u_1, u_2\fresh\\ \locked(\mutname) = \lunlocked}
           {
             \lconfig{\cthread{a}{\prio}{\sig}
               {\scsend{\stack}{\kwwithlock{\kwmutex{\mutname}}{s}}}
              \tpcp \tp}
              {\mem}
              {\waiters}
              {\locked}
              {\graph}\\
              \mstep
              \lconfig{\cthread{a}{\prio}{\sig}
                {\scsend{\scp{\stack}{\kwwithlocks{\kwmutex{\mutname}}{\shole}}}
                  {
                    s}}
                \tpcp \tp}
              {\mem}
              {\waiters}
              {\waitupd{\locked}{\mutname}{(a, u_1, u_2)}}
              {\graph \scomp{a} \gnode{u}{\sig}}
           }
      \and
      \Rule{WithLockS2}
           {u_1, u_2\fresh\\ \locked(\mutname) = (a', u_1', u_2')\\
             \ple{\prio}{\uprio{\graph}{a'}}\\
             \stackstate' =
             \scsend{\scp{\stack}{\kwwithlocks{\kwmutex{\mutname}}{\shole}}}
                    {
                      s}
           }
           {
             \lconfig{\cthread{a}{\prio}{\sig}
               {\scsend{\stack}{\kwwithlock{\kwmutex{\mutname}}{s}}}
              \tpcp \tp}
              {\mem}
              {\waiters}
              {\locked}
              {\graph}\\
      \mstep
      \lconfig{\tp}
              {\mem}
              {\waitupd{\waiters}{\mutname}
                {\waitcons{(a, (u_1, u_2), \prio, \sig, \stackstate')}{\waiters(\mutname)}}}
              {\locked}
              {\graph \scomp{a} \gnode{u_1}{\sig} \scomp{a} \gnode{u_2}{\sig}
                \cup \{\edgeweak{u_1'}{u_2}\}
              }
           }
      \and
      \Rule{WithLockS3}
           {u_1, u_2, u_1'', u_2''\fresh\\ \locked(\mutname) = (b, u_1', u_2')\\
             \nple{\prio}{\prio'}\\
             \stackstate' =
             \scsend{\scp{\stack}{\kwwithlocks{\kwmutex{\mutname}}{\shole}}}
                    {
                      s}\\
             \sigcvtype{\mutname}{\prio_{\mutname}} \in \sig\\
             \reads = \{\edgeweak{u_1''}{u_{2c}} \mid (c, (u_{1c}, u_{2c}),
             \prio_c, \sig_c, s_c) \in \waiters(\mutname)\}\\
             \prioceil{\stackstate_b}{\mutname}{b}{\prio'}
                      {\stackstate''}
           }
           {
             \lconfig{\cthread{a}{\prio}{\sig}
               {\scsend{\stack}{\kwwithlock{\kwmutex{\mutname}}{s}}}
               \tpcp \cthread{b}{\prio'}{\sig'}
               {\stackstate_b}
                \tpcp \tp}
              {\mem}
              {\waiters}
              {\locked}
              {\graph}\\
      \mstep
      \lconfig{\cthread{b'}{\prio_{\mutname}}{\sig'}
        {\stackstate''}
        \tpcp
        \tp}
              {\mem}
              {\waitupd{\waiters}{\mutname}
                {\waitcons{(a, (u_1, u_2), \prio, \sig, \stackstate')}{\waiters(\mutname)}}}
              {\waitupd{\locked}{\mutname}{(b', u_1'', u_2'')}}{}
              \\
              {\graph \scomp{b'} \gnode{u_1''}{\sig'} \scomp{b'} \gnode{u_2''}{\sig'}
                \scomp{a} \gnode{u_1}{\sig} \scomp{a} \gnode{u_2}{\sig}
                \cup \{\edgeweak{u_1''}{u_2}, \edgespawn{\lastv{\graph}{a'}}{u_1''}\}
              }
           }
      \and
      \iffull
      \Rule{WithLockS4}
           {u_1, u_2, u_1'', u_2''\fresh\\ \locked(\mutname) = (b, u_1', u_2')\\
             \waiters(\cvmutname') =
             \waitcons{\ell_1}{\waitcons{(b, \_, \prio', \sig', \stackstate_b)}
               {\ell_2}}\\
             \nple{\prio}{\prio'}\\
             \stackstate' =
             \scsend{\scp{\stack}{\kwwithlocks{\kwmutex{\mutname}}{\shole}}}
                    {
                      s}\\
             \sigcvtype{\mutname}{\prio_{\mutname}} \in \sig\\
             \reads = \{\edgeweak{u_1''}{u_{2c}} \mid (c, (u_{1c}, u_{2c}),
             \prio_c, \sig_c, s_c) \in \waiters(\mutname)\}\\
             \prioceil{\stackstate_b}{\mutname}{b}{\prio'}
                      {\stackstate''}\\
                      \waiters'(\cvmutname') =
             \waitcons{\ell_1}{\waitcons{(b', \_, \prio_\mutname, \sig', \stackstate'')}
               {\ell_2}}
           }
           {
             \lconfig{\cthread{a}{\prio}{\sig}
               {\scsend{\stack}{\kwwithlock{\kwmutex{\mutname}}{s}}}
               \tpcp \tp}
              {\mem}
              {\waiters}
              {\locked}
              {\graph}\\
      \mstep
      \lconfig{\tp}
              {\mem}
              {\waitupd{\waiters'}{\mutname}
                {\waitcons{(a, (u_1, u_2), \prio, \sig, \stackstate')}{\waiters(\mutname)}}}
              {\waitupd{\locked}{\mutname}{(b', u_1'', u_2'')}}{}
              \\
              {\graph \scomp{b'} \gnode{u_1''}{\sig'} \scomp{b'} \gnode{u_2''}{\sig'}
                \scomp{a} \gnode{u_1}{\sig} \scomp{a} \gnode{u_2}{\sig}
                \cup \{\edgeweak{u_1''}{u_2}, \edgespawn{\lastv{\graph}{a'}}{u_1''}\}
              }
           }
      \and
      \Rule{WithLockE1}
           {u\fresh\\ \waiters(\mutname) = \waitemp}
           {
             \lconfig{\cthread{a}{\prio}{\sig}
               {\screturn{\scp{\stack}{\kwwithlocks{\kwmutex{\mutname}}{\shole}}}}
               \tpcp \tp}
                     {\mem}
                     {\waiters}
                     {\locked}
                     {\graph}
      \mstep
      \lconfig{\cthread{a}{\prio}{\sig}{\screturn{\stack}} \tpcp
                \tp}
              {\mem}
              {\waiters}
              {\waitupd{\locked}{\mutname}{\lunlocked}}
              {\graph \scomp{a} \gnode{u}{\sig}}
           }
      \and
      \fi
      \Rule{WithLockE2}
           {u\fresh\\
             \waiters(\mutname) = \ell_1,\waitcons{(b, (u_1, u_2), \prio_b, \sig_b, \stackstate_b)}
                     {\ell_2}\\
                     \forall (\_, \_, \prio_c, \_, \_) \in \ell_1\ell_2.
                     \ple{\prio_c}{\prio_b}\\
             \reads = \{\edgeweak{u_1}{u_2'} \mid (c, (u_1', u_2'),
                \prio_c, \sig_c, s_c) \in \ell_1\ell_2\}
           }
           {
      \lconfig{\cthread{a}{\prio}{\sig}
        {\screturn{\scp{\stack}{\kwwithlocks{\kwmutex{\mutname}}{\shole}}}}
              \tpcp \tp}
              {\mem}
              {\waiters}
              {\locked}
              {\graph}\\
      \mstep
      \lconfig{\cthread{a}{\prio}{\sig}{\screturn{\stack}} \tpcp
        \cthread{b}{\prio_b}{\sig_b}{\stackstate_b} \tpcp
                \tp}
              {\mem}
              {\waitupd{\waiters}{\mutname}{\ell_1\ell_2}}
              {\waitupd{\locked}{\mutname}{(b, u_1, u_2)}}
              {\graph \scomp{a} \gnode{u}{\sig} \cup \{\edgesync{u}{u_2}\}
                \cup \reads
              }
           }
    \and
    \Rule{WithLockE3}
           {u, u'\fresh\\
             \waiters(\mutname) = \waitemp
           }
           {
      \lconfig{\cthread{a}{\prio}{\sig}
        {\screturn{\scp{\stack}{\kwwithlocksp{\kwmutex{\mutname}}{\shole}{a'}{\prio'}}}}
              \tpcp \tp}
              {\mem}
              {\waiters}
              {\locked}
              {\graph}\\
      \mstep
      \lconfig{\cthread{a'}{\prio'}{\sig}{\screturn{\stack}} \tpcp
                \tp}
              {\mem}
              {\waiters}
              {\waitupd{\locked}{\mutname}{\lunlocked}}
              {\graph \scomp{a} \gnode{u}{\sig} \scomp{a'} \gnode{u'}{\sig} \cup
                \{\edgesync{u}{u'}\}
              }
           }
     \and
     \iffull
     \Rule{WithLockE4}
           {u, u'\fresh\\
             \waiters(\mutname) = \ell_1,\waitcons{(b, (u_1, u_2), \prio_b, \sig_b, \stackstate_b)}
                     {\ell_2}\\
                     \forall (\_, \_, \prio_c, \_, \_) \in \ell_1\ell_2.
                     \ple{\prio_c}{\prio_b}\\
             \reads = \{\edgeweak{u_1}{u_2'} \mid (c, (u_1', u_2'),
                \prio_c, \sig_c, \stackstate_c) \in \ell_1\ell_2\}
           }
           {
      \lconfig{\cthread{a}{\prio}{\sig}
        {\screturn{\scp{\stack}{\kwwithlocksp{\kwmutex{\mutname}}{\shole}{a'}{\prio'}}}}
              \tpcp \tp}
              {\mem}
              {\waiters}
              {\locked}
              {\graph}\\
      \mstep
      \lconfig{\cthread{a'}{\prio'}{\sig}{\screturn{\stack}} \tpcp
        \cthread{b}{\prio_b}{\sig_b}{\stackstate_b} \tpcp
                \tp}
              {\mem}
              {\waitupd{\waiters}{\mutname}{\ell_1\ell_2}}
              {\waitupd{\locked}{\mutname}{(b, u_1, u_2)}}
              {}\\
              \graph \scomp{a} \gnode{u}{\sig} \scomp{a'} \gnode{u'}{\sig} \cup
                \{\edgesync{u}{u_2}, \edgesync{u}{u'}\}
                \cup \reads
           }
\end{mathpar}
  \caption{Dynamic semantics rules (part 2)}
\label{fig:cost2}
\end{figure}
\begin{figure}
  \footnotesize
  
  \centering
  \begin{mathpar}
    \Rule{If1}
           {n > 0\\ u \fresh}
           {
             \lconfig{\cthread{a}{\prio}{\sig}{\scsend{\stack}{\kwif{\kwn}{s_1}{s_2}}}
               \tpcp \tp}
              {\mem}
              {\waiters}
              {\locked}
              {\graph}
      \mstep
      \lconfig{\cthread{a}{\prio}{\sig}{\scsend{\stack}{s_1}} \tpcp
                \tp}
              {\mem}
              {\waiters}
              {\locked}
              {\graph \scomp{a} \gnode{u}{\sig}}
           }
      \and
      \Rule{If2}
           {n = 0\\ u \fresh}
           {
             \lconfig{\cthread{a}{\prio}{\sig}{\scsend{\stack}{\kwif{\kwn}{s_1}{s_2}}}
                 \tpcp \tp}
              {\mem}
              {\waiters}
              {\locked}
              {\graph}
      \mstep
      \lconfig{\cthread{a}{\prio}{\sig}{\scsend{\stack}{s_2}} \tpcp
                \tp}
              {\mem}
              {\waiters}
              {\locked}
              {\graph \scomp{a} \gnode{u}{\sig}}
           }
      \and
      \Rule{While}
           {\strut}
           {
      \lconfig{\cthread{a}{\prio}{\sig}{\scsend{\stack}{\kwwhile{v}{s}}} \tpcp \tp}
              {\mem}
              {\waiters}
              {\locked}
              {\graph}
      \mstep
      \lconfig{\cthread{a}{\prio}{\sig}{\scsend{\stack}{\kwif{v}{(s; \kwwhile{v}{s})}{\kwskip}}}
        \tpcp \tp}
              {\mem}
              {\waiters}
              {\locked}
              {\graph \scomp{a} \gnode{u}{\sig}}
           }
      \and
      \Rule{Skip}
           {u \fresh}
           {
      \lconfig{\cthread{a}{\prio}{\sig}{\scsend{\stack}{\kwskip}} \tpcp \tp}
              {\mem}
              {\waiters}
              {\locked}
              {\graph}
      \mstep
      \lconfig{\cthread{a}{\prio}{\sig}{\screturn{\stack}} \tpcp
                \tp}
              {\mem}
              {\waiters}
              {\locked}
              {\graph \scomp{a} \gnode{u}{\sig}}
           }
      \and
       \Rule{Seq1}
           {\strut}
           {
             \lconfig{\cthread{a}{\prio}{\sig}
               {\scsend{\stack}{s_1; s_2}} \tpcp \tp}
                     {\mem}
                     {\waiters}
                     {\locked}
                     {\graph}
             \mstep
             \lconfig{\cthread{a}{\prio}{\sig}
               {\scsend{\scp{\stack}{\shole; s_2}}{s_1}} \tpcp \tp}
                     {\mem}
                     {\waiters}
                     {\locked}
                     {\graph}
           }
      \and
      \Rule{Seq2}
           {\strut}
           {
             \lconfig{\cthread{a}{\prio}{\sig}
               {\screturn{\scp{\stack}{\shole; s_2}}} \tpcp \tp}
                     {\mem}
                     {\waiters}
                     {\locked}
                     {\graph}
             \mstep
             \lconfig{\cthread{a}{\prio}{\sig}
               {\scsend{\stack}{s_2}} \tpcp \tp}
                     {\mem}
                     {\waiters}
                     {\locked}
                     {\graph}
           }
  \fi
  \end{mathpar}
  \caption{Dynamic semantics rules\iffull{} (part 3)\fi}
\label{fig:cost3}
\end{figure}
\iffull
\begin{figure}
  \langfigsize
  \centering
  \begin{mathpar}
    \RuleNolabel
        {\strut}
        {\prioceil{\scp{\stack}{\kwwithlocks{\kwmutex{\mutname}}{\shole}}}
          {\mutname}{a}{\prio}
          {\scp{\stack}{\kwwithlocksp{\kwmutex{\mutname}}{\shole}{a}{\prio}}}
        }
    \and
    \RuleNolabel
        {\prioceil{\stack}{\mutname}{a}{\prio}{\stack'}\\\\
          \stframe \neq \kwwithlocks{\kwmutex{\mutname}}{\shole}}
        {\prioceil{\scp{\stack}{\stframe}}{\mutname}{a}{\prio}
          {\scp{\stack'}{\stframe}}}
    \and
    \RuleNolabel
        {\prioceil{\stack}{\mutname}{a}{\prio}{\stack'}
        }
        {\prioceil{\scsend{\stack}{s}}{\mutname}{a}{\prio}{\scsend{\stack'}{s}}
        }
    \and
    \RuleNolabel
        {\prioceil{\stack}{\mutname}{a}{\prio}{\stack'}}
        {\prioceil{\screturn{\stack}}{\mutname}{a}{\prio}{\screturn{\stack'}}}
    \and
    \RuleNolabel
        {\prioceil{\stack}{\mutname}{a}{\prio}{\stack'}}
        {\prioceil{\ssend{\stack}{\instr}}{\mutname}{a}{\prio}
          {\ssend{\stack'}{\instr}}
        }
    \and
    \RuleNolabel
        {\prioceil{\stack}{\mutname}{a}{\prio}{\stack'}}
        {\prioceil{\sreturn{\stack}{v}}{\mutname}{a}{\prio}
          {\sreturn{\stack'}{v}}
        }
  \end{mathpar}
  \caption{Rules for dynamically changing thread priorities.}
  \label{fig:prio-ceil}
\end{figure}
\fi

The dynamic semantics judgment is
$
\lconfig{\tp}{\mem}{\waiters}{\locked}{\graph} \mstep
\lconfig{\tp'}{\mem'}{\waiters'}{\locked'}{\graph'}
$,
indicating that the configuration steps and and if the DAG is~$\graph$
before the step, it is~$\graph'$ after.
\iffull
The rules are in Figures~\ref{fig:cost1}--\ref{fig:cost3}.
\else
Selected rules are shown in Figure~\ref{fig:cost3}; the full rules are
available in the supplementary material.
\fi
Most rules create a fresh vertex~$u$ and add it to the graph:
the notation~$\graph \scomp{a} \gnode{u}{\sig}$
adds the vertex~$u$ onto the thread~$a$ in~$\graph$.
We also associate with each vertex
the signature of the thread at that point in the program.
We will also use the notations~$\graph \cup E$ to add the edges in~$E$ to
the graph: we will denote a create edge, a sync edge, and a weak edge
as~$\edgespawn{u_1}{u_2}$, $\edgesync{u_1}{u_2}$, and~$\edgeweak{u_1}{u_2}$,
respectively.

The most unusual rules are those dealing with CVs and mutexes.
Rule~\rulename{Wait} adds the current thread to~$\waiters$.
Note that the use of ``2/3-cps'' form means that we don't have to evaluate
the subexpression to a condition variable handle; the syntax and typing rules
ensure this is the only form it can take.
The~\rulename{Wait} rule adds two vertices to the thread; one represents
the~\texttt{wait} operation itself, and the other represents the next operation
in the thread following the~\texttt{wait}; this is the operation that will
be enabled by a~\texttt{signal}, so this vertex is added to~$\waiters$ so that
we may later add a sync edge leading to it.
There are two rules for signal, depending on whether any threads are
currently waiting on the CV.
If there are (\rulename{Signal1}), one is chosen; it is removed
from~$\waiters(\cvname)$ and added back to the thread pool.
We also add the sync edge~$(u, u')$ to the DAG, where~$u$ is the vertex
representing the signal and~$u'$ is the vertex following the~\texttt{wait}
operation in the waiting thread.
If there are no threads waiting (\rulename{Signal2}), the vertex~$u$ is
added to the DAG but otherwise no action is taken.

There are four rules for acquiring the lock to start
a~$\kwwithlock{v}{s}$ critical section.
If the lock is currently unlocked, \rulename{WithLockS1}
transitions it to locked by updating~$\locked(\mutname)$ to the current thread.
As with~\rulename{Wait}, we add two vertices to the thread, representing
the acquire operation and the first operation in the critical section.
The statement is also changed to~$\kwwithlocks{\kwmutex{\mutname}}{s}$
to indicate that the lock has already been acquired.
If the lock is currently held by a higher-priority thread than the current
thread (\rulename{WithLockS2}),
the current thread is added to~$\waiters(\mutname)$
and a weak edge is added from~$u_1'$, the vertex that successfully acquired
the lock, to~$u_2$, the newly added vertex in the current thread.
%
%
%
If the lock is currently held by a thread that is not higher priority than
the current thread,
the thread holding the lock is promoted to priority~$\prio_{\mutname}$,
the priority ceiling of the mutex.
This is done by \rulename{WithLockS3} if the thread holding the lock is
active, and \rulename{WithLockS4} \iffull\else (omitted for space, but
included in the supplementary material) \fi
if it is waiting on some other lock or CV.
The rule creates a new thread~$b'$ at~$\prio_{\mutname}$ to
run the remainder of the critical section.
The auxiliary judgment~$\prioceil{\stackstate_b}{\mutname}{b}{\prio'}{\stackstate''}$
replaces the critical section with one of the
form~$\kwwithlocksp{v}{s}{b}{\prio'}$, which indicates that
the thread should return control to thread~$b$ at~$\prio'$ when the
critical section finishes.
\iffull
The rules for this judgment are defined in Figure~\ref{fig:prio-ceil}.
\else
The formal rules for the judgment are in the supplementary material.
\fi
Two new vertices~$u_1''$ and~$u_2''$ are added to begin thread~$b'$---these
are needed to update~$\locked(\mutname)$ to indicate that~$b'$ now holds the
lock.
Finally,
weak edges are added from~$u_1''$
to every other thread in~$\waiters(\mutname)$.

There are also four rules for ending a critical section.
If the thread has not been promoted to a higher priority and no other threads
are waiting for the lock, rule \rulename{WithLockE1} simply
updates~$\locked(\mutname)$ to~$\lunlocked$, releasing the lock.
If other threads are waiting, the highest priority thread~$b$ is chosen
from~$\waiters(\mutname)$.
Rule \rulename{WithLockE2} assigns~$\locked(\mutname)$ to~$b$ and updates
the list of waiting threads.
The rule also adds a synchronization edge from the vertex representing the
release of the lock to the vertex in~$b$ that acquires the lock.
Finally, the rule adds weak edges to the waiting threads.
If the thread holding the lock has been promoted to the lock's priority
ceiling, corresponding rules (\rulename{WithLockE3} if no other threads are
waiting, otherwise \rulename{WithLockE4}) additionally move the thread back
to its original name and priority.

\subsection{Correctness Proof}

The main goal of this section is to prove that the graph produced by an
execution of a {\langname} program is well-formed, meaning that it contains
no priority inversions (that are not handled by the dynamic priority ceiling
mechanism of the operational semantics).
We prove this together with standard type preservation and the preservation
of a number of additional invariants on the configuration and graph that
aid in the proof of well-formedness.
\iffull
First, Lemma~\ref{lem:prio-change-pres} shows that typing of stacks and
stack states is preserved by a dynamic priority change as accomplished
by the~$\prioceil{\stackstate}{\mutname}{a}{\prio}{\stackstate'}$ judgment
and its auxiliary judgment on
stacks,~$\prioceil{\stack}{\mutname}{a}{\prio}{\stack'}$,



\begin{lemma}\label{lem:prio-change-pres}
\begin{itemize}
\item If~$\stackacceptsc{\sig}{\stack}{\fctx_1}{\prios_1}{\fctx_2}{\prio}$
and~$\sigcvtype{\mutname}{\prio'} \in \sig$
and~$\prioceil{\stack}{\mutname}{a}{\prio}{\stack'}$
then~$\stackacceptsc{\sig}{\stack'}{\fctx_1}{\prios_1}{\fctx_2}{\prio'}$
\item If~$\sstyped{\sig}{\fctx}{\stackstate}{\prio}{\fctx'}$
and~$\sigcvtype{\mutname}{\prio'} \in \sig$
and~$\prioceil{\stackstate}{\mutname}{a}{\prio}{\stackstate'}$
then~$\sstyped{\sig}{\fctx}{\stackstate'}{\prio'}{\fctx'}$
\end{itemize}
\end{lemma}
\begin{proof}
\begin{itemize}
\item By induction on the derivation of~$\prioceil{\stack}{\mutname}{a}{\prio}{\stack'}$.
\iffull
\begin{itemize}
\item~$\stack = \scp{\stack_0}{\kwwithlocks{\kwmutex{\mutname}}{\shole}}$
and~$\stack' = \scp{\stack_0}{\kwwithlocksp{\kwmutex{\mutname}}{\shole}}{\mutname}{\prio}$.
By inversion,~$\stackacceptsc{\sig}{\stack_0}{\fctx_1}{\prios_1'}{\fctx_2}{\prio}$
where~$\prios_1 = \prios_1' \cup \{\prio'\}$
and~$\etyped{\sig}{\ectx}{\kwmutex{\mutname}}{\kwmutext{\prio'}}$.
By~\rulename{KS-WithLockSP},
$\stackacceptsc{\sig}{\stack'}{\prios_1}{\fctx_1}{\prios_1}{\fctx_2}{\prio'}$.
\item By induction.
\end{itemize}
\fi
\item By part (1).
\end{itemize}
\end{proof}
\fi

We now present the invariants that we use to prove well-formedness of graphs.
Invariant~(\ref{inv:typed}) is simply that the configuration is well-typed; the
fact that this invariant is preserved corresponds to a standard Preservation
result.
Invariant~(\ref{inv:wf}) is itself that the graph is well-formed.
The remaining invariants state properties that are required to show that
this invariant is preserved.
Invariants~(\ref{inv:waiters})--(\ref{inv:waiters-prio}) deal with mutexes
and Invariants~(\ref{inv:cv-waiters})--(\ref{inv:no-signal}) deal with CVs.
Two invariants directly ensure facets of well-formedness:
Invariant~(\ref{inv:waiters}) ensures that weak edges are present from
threads that hold locks to threads waiting for the lock; this is required for
the definition of well-formedness.
Invariant~(\ref{inv:no-signal}) ensures that high-priority threads do not
receive low-priority vertices on their critical paths due to sync edges
by stating that such vertices have no permission for CVs which would
be required in order to signal them.
Invariants~(\ref{inv:waiters}),~(\ref{inv:mutex-prio}),
and~(\ref{inv:waiters-prio}) place requirements on priorities of threads
dealing with mutexes: first, every mutex is associated with a priority, its
priority ceiling~$\prio_\mutname$, in every signature in which the mutex is
present; second, every thread waiting on a mutex must have a priority less
than both the ceiling and the thread currently holding the lock.
%
Invariant~(\ref{inv:unlocked-empty}) states that no threads are waiting for an
unlocked lock, which limits the steps at which we have to maintain the
invariants on waiting threads.
Invariant~(\ref{inv:cv-waiters}) states that any threads waiting on a CV
have a lower priority than the CV's handle.
The last invariant,~$(\ref{inv:can-promote})$, as well as the last
part of~$(\ref{inv:waiters})$, are only necessary to prove Progress; they
state that any thread holding a lock at a priority lower than its ceiling
can be promoted to the priority ceiling.

\begin{definition}[Invariants]\label{def:invs}
We require the following invariants on configurations
$\lconfig{\tp}{\mem}{\waiters}{\locked}{\graph}$,
where~$\tp = \cthread{a_1}{\prio_1}{\sig_1}{s_1} \tpcp \dots \tpcp
        \cthread{a_n}{\prio_n}{\sig_n}{s_n}$.
\begin{enumerate}
\item \label{inv:typed} $\cfgtyped{\sig}
        {\twlconfig{\tp}{\mem}{\waiters}{\locked}}$.
        In the below, we will use~$\fctx_i$ to refer to
        the permission used in the typing
        derivation for~$a_i \in \tp$
\item \label{inv:wf} $\graph$ is well-formed
\item \label{inv:waiters} If~$\locked(\mutname) = (a, u_1, u_2)$
and~$(a', u', \prio_{a'}, \sig_{a'}, \stackstate_{a'}) \in \waiters(\mutname)$,
then~$\edgeweak{u_1}{u'} \in \graph$ and
$\ple{\prio_{a'}}{\uprio{\graph}{a}}$ and
if~$\prio_{a'} \neq \prio_\mutname$,
then~$\prioceil{\stackstate_{a'}}{\mutname}{a'}{\prio_{a'}}{\stackstate'}$.
\item \label{inv:unlocked-empty} If~$\locked(\mutname) = \lunlocked$
then~$\waiters(\mutname) = \waitemp$
\item \label{inv:mutex-prio}
For all mutexes~$\mutname$, there is a priority~$\prio_\mutname$
such that for all~$i$,
if~$\sigcvtype{\mutname}{\prio} \in \sig_i$, then $\prio = \prio_\mutname$
\item \label{inv:waiters-prio} For all mutexes~$\mutname$,
if~$(a, (u_1, u_2), \prio, \sig, \stackstate) \in \waiters(\mutname)$,
then~$\ple{\prio}{\prio_\mutname}$.

\item \label{inv:cv-waiters} For all CVs~$\cvname$ and
all~$(a, (u_1, u_2), \prio_a, \sig_a, \stackstate_a') \in \waiters(\cvname)$,
we have~$\sigcvtype{\cvname}{\prio} \in \sig_a$
where~$\ple{\prio_a}{\prio}$
\item \label{inv:no-signal} If~$\not\sanc{(u', \sig')}{(u, \sig)}$
and~$\sigcvtype{\cvname}{\prio} \in \sig$ and the last vertex in~$a_i$
is a descendant of~$u'$, then
\begin{itemize}
\item For all~$\prio'$ such that~$\nple{\prio}{\prio'}$,
$\permof{\fctx_i}{\cvname}{\prio'} = \fnone$ and
\item If~$\nple{\prio}{\uprio{\graph}{u'}}$
and~$\ple{\prio}{\uprio{\graph}{a_i}}$, then for all~$\prio'$,
$\permof{\fctx_i}{\cvname}{\prio'} = \fnone$
\end{itemize}
\item \label{inv:can-promote} If~$\locked(\mutname) = (a, u_1, u_2)$,
then~$\cthread{a}{\prio}{\sig}{\stackstate} \in \tp$
or~$(a, (u_1, u_2), \prio, \sig, \stackstate) \in \waiters(\cvmutname)$
and~$\sigcvtype{\mutname}{\prio_{\mutname}} \in \sig$
and if~$\prio \neq \prio_\mutname$,
then~$\prioceil{\stackstate}{\mutname}{a}{\prio}{\stackstate'}$.
\end{enumerate}
\end{definition}

Theorem~\ref{thm:preservation} (Preservation) states that a step of the dynamic
semantics preserves the invariants.

\begin{theorem}\label{thm:preservation}
Suppose~$\tp = \cthread{a_1}{\prio_1}{\sig_1}{s_1} \tpcp \dots \tpcp
        \cthread{a_n}{\prio_n}{\sig_n}{s_n}$
and~$\lconfig{\tp}{\mem}{\waiters}{\locked}{\graph}$ meets the invariants
of Definition~\ref{def:invs}.
If $\lconfig{\tp}{\mem}{\waiters}{\locked}{\graph}
\mstep
\lconfig{\tp'}{\mem'}{\waiters'}{\locked'}{\graph'}$
then~$\lconfig{\tp'}{\mem'}{\waiters'}{\locked'}{\graph'}$ meets the
invariants.
\end{theorem}

\begin{proof}
By induction on the derivation of
$\lconfig{\tp}{\mem}{\waiters}{\locked}{\graph}
\mstep
\lconfig{\tp'}{\mem'}{\waiters'}{\locked'}{\graph'}$.
\iffull
\begin{itemize}
\item \rulename{Spawn}.
Then~$\tp = \cthread{a}{\prio}{\sig}{\ssend{\stack}{
            \kwspawn{\prio'}{\tau}{\vec{\cvname}}{s'}}}
        \tpcp \tp_0$
and~$\tp' = \cthread{a}{\prio}{\sig}{\sreturn{\stack}{\kwtriv}} \tpcp
                \cthread{b}{\prio}{\sig}{\ssend{\estack}{s'}}
                \tp_0$.
\begin{enumerate}
\item
By inversion on the typing rules~\rulename{KS-PopInstr} and~\rulename{Spawn},
$\stackaccepts{\sig}{\stack}{\kwunit}{\fctx_2}{\prios_1}{\fctx'}{\prio}$
and
$\fctx = \fctx_1, \fctx_2$ and
$\styped{\sig}{\ctx}{\fctx_1}{s'}{\prio'}{\fctx_0}$
and $\forall \cvname, \prio''. \permof{\fctx}{\cvname}{\prio''} \neq \fnone
  \rightarrow
  \permof{\fctx}{\cvname}{\prio} \neq \fnone$.
By~\rulename{unitI} and~\rulename{KS-PushInstr},
$\sstyped{\sig}{\fctx_2}{\sreturn{\stack}{\kwtriv}}{\prio}{\fctx'}$.
For all~$\cvname, \prio$, $\splitsto{\permof{\fctx}{\cvname}{\prio}}{\permof{\fctx_1}{\cvname}{\prio}}{\permof{\fctx_2}{\cvname}{\prio}}$,
so permissions are preserved.
\item No new sync edges are created,
so well-formedness is preserved.
\item No items are added to~$\waiters$ or $\locked$.
\item $\locked$, $\waiters$ are unchanged.
\item $\sig_b' = \sig_a'$, so the priority of any mutexes is preserved.
No items are added to~$\waiters$.
\item No items are added to~$\waiters$.
\item No items are added to~$\waiters$.
\item If~$u$ meets the requirements of condition~(\ref{inv:no-signal}),
  then the last
vertex in~$b$ will as well because if~$\permof{\fctx}{\cvname}{\prio} = 0$,
then~$\permof{\fctx_2}{\cvname}{\prio} = 0$.
If~$\prio_\Sigma$ is the priority of the CV in condition~(\ref{inv:no-signal})
and~$\nple{\prio_\Sigma}{\prio}$ and~$\ple{\prio_\Sigma}{\prio'}$,
then~$\permof{\fctx_2}{\cvname}{\prio_0} = \fnone$ for all~$\prio_0$ because
by assumption,~$\permof{\fctx}{\cvname}{\prio} = \fnone$, so by the
contrapositive of the premise of
the typing rule,~$\permof{\fctx}{\cvname}{\prio_0} = \fnone$.
\item If~$\locked(\mutname) = (a, \_, \_)$,
then we have~$\prioceil{\stack}{\mutname}{a}{\prio}{\stack^p}$
for some~$\stack^p$. This remains unchanged.
\end{enumerate}

\item \rulename{NewRef}
Then~$\tp = \cthread{a}{\prio}{\sig}{\ssend{\stack}{
            \kwnewref{\tau}{y}}}
        \tpcp \tp_0$
and~$\tp' = \cthread{a}{\prio}{\sig,\sigrtype{\kwassn}{\tau}}
        {\sreturn{\stack}{\kwref{\kwassn}}} \tpcp \tp_0$
and~$\mem' = \memupd{\mem}{\mement{\kwassn}{y}{u}{\sig}}$.
\begin{enumerate}
\item By inversion,
$\ityped{\sig}{\ctx}{\fctx}{\kwnewref{\tau}}{\kwreft{\tau}}{\prio}{\fctx}$
and $\stackaccepts{\sig}{\stack}{\kwreft{\tau}}{\fctx}{\prios_1}{\fctx'}$
and~$\etyped{\sig}{\ctx}{y}{\tau}$.
By \rulename{RevVal},~$\etyped{\sig,\sigrtype{\kwassn}{\tau}}
{\ctx}{\kwref{\kwassn}}{\kwreft{\tau}}$.
We have~$\mtyped{\sig,\sigrtype{\kwassn}{\tau}}{\mem'}$.
By~\rulename{KS-PushInstr},
$\sstyped{\sig}{\fctx}{\sreturn{\stack}{\kwref{\kwassn}}}{\prio}{\fctx'}$.
\item No new sync edges are created,
so well-formedness is preserved.
\item No items are added to~$\waiters$ or $\locked$.
\item $\locked$, $\waiters$ unchanged.
\item $\kwassn$ is not a mutex, so the condition is preserved.
\item No items are added to~$\waiters$.
\item No items are added to~$\waiters$.
\item Permissions are not split or reassigned by this step.
\item If~$\locked(\mutname) = (a, \_, \_)$,
then we have~$\prioceil{\stack}{\mutname}{a}{\prio}{\stack^p}$
for some~$\stack^p$. This remains unchanged.
\end{enumerate}

\item \rulename{Deref}
Then~$\tp = \cthread{a}{\prio}{\sig}{\ssend{\stack}{
            \kwderef{\kwref{\kwassn}}}}
        \tpcp \tp_0$
and~$\tp' = \cthread{a}{\prio}{\sig, \sig'}
        {\sreturn{\stack}{v}} \tpcp \tp_0$
where~$\mem(\kwassn) = \memrent{v}{u'}{\sig'}$.
\begin{enumerate}
\item By inversion,
$\ityped{\sig}{\ctx}{\fctx}{\kwderef{\kwref{\kwassn}}}{\tau}{\prio}{\fctx}$
and~$\sigrtype{\kwassn}{\tau} \in \sig$
and~$\stackaccepts{\sig}{\stack}{\tau}{\fctx}{\prios_1}{\fctx'}{\prio}$.
Because~$\mtyped{\sig}{\mem}$, we
have~$\etyped{\sig}{\ctx}{v}{\tau}$.
By~\rulename{KS-PushInstr},
$\sstyped{\sig}{\fctx}{\sreturn{\stack}{v}}{\prio}{\fctx'}$.
\item No new sync edges are created,
so well-formedness is preserved.
\item No items are added to~$\waiters$ or $\locked$.
\item $\locked$, $\waiters$ unchanged.
\item The condition held of~$\sig$ and~$\sig'$, so it holds of~$\sig, \sig'$.
\item No items are added to~$\waiters$.
\item No items are added to~$\waiters$.
\item Permissions are not split or reassigned by this step.
\item If~$\locked(\mutname) = (a, \_, \_)$,
then we have~$\prioceil{\stack}{\mutname}{a}{\prio}{\stack^p}$
for some~$\stack^p$. This remains unchanged.
\end{enumerate}

\item \rulename{Update}
Then~$\tp = \cthread{a}{\prio}{\sig}{\ssend{\stack}{
            \kwassign{\kwref{\kwassn}}{y}}}
        \tpcp \tp_0$
and~$\tp' = \cthread{a}{\prio}{\sig}
        {\sreturn{\stack}{\kwtriv}} \tpcp \tp_0$
and~$\mem' = \memupd{\mem}{\mement{\kwassn}{y}{u}{\sig}}$.
\begin{enumerate}
\item By inversion,
$\ityped{\sig}{\ctx}{\fctx}{\kwassign{\kwref{\kwassn}}{y}}{\kwunit}{\prio}{\fctx}$
and~$\etyped{\sig}{\ctx}{\kwref{\kwassn}}{\kwreft{\tau}}$
and~$\etyped{\sig}{\ctx}{y}{\tau}$
and~$\stackaccepts{\sig}{\stack}{\kwunit}{\fctx}{\prios_1}{\fctx'}{\prio}$.
By~\rulename{KS-PushInstr},
$\sstyped{\sig}{\fctx}{\sreturn{\stack}{\kwtriv}}{\prio}{\fctx'}$.
We still have~$\mtyped{\sig}{\mem'}$ because~$y$ has type~$\tau$.
\item No new sync edges are created,
so well-formedness is preserved.
\item No items are added to~$\waiters$ or $\locked$.
\item $\locked$, $\waiters$ unchanged.
\item Signatures are not changed.
\item No items are added to~$\waiters$.
\item No items are added to~$\waiters$.
\item Permissions are not split or reassigned by this step.
\item If~$\locked(\mutname) = (a, \_, \_)$,
then we have~$\prioceil{\stack}{\mutname}{a}{\prio}{\stack^p}$
for some~$\stack^p$. This remains unchanged.
\end{enumerate}

\item \rulename{NewCV}
Then~$\tp = \cthread{a}{\prio}{\sig}
                {\ssend{\stack}{\kwnewcv{\prio'}}}
              \tpcp \tp_0$
and~$\tp' = \cthread{a}{\prio}{\sig,\sigcvtype{\cvname}{\prio'}}
        {\sreturn{\stack}{\kwcv{\cvname}}}
        \tpcp \tp_0$.
\begin{enumerate}
\item By inversion,
$\ityped{\sig}{\ctx}{\fctx}{\kwnewcv{\prio'}}{\kwcvt{\cvname}{\prio'}}{\prio}{\fctx'}$
and~$\stackaccepts{\sig}{\stack}{\kwcvt{\cvname}{\prio'}}{\fctx'}{\prios_1}{\fctx''}{\prio}$.
By \rulename{CV},
$\etyped{\sig, \sigcvtype{\cvname}{\prio'}}{\ctx}{\kwcv{\cvname}}
{\kwcvt{\cvname}{\prio'}}$.
By~\rulename{KS-PushInstr},
$\sstyped{\sig}{\fctx'}{\sreturn{\stack}{\kwcv{\cvname}}}{\prio}{\fctx''}$.
\item No new sync edges are created,
so well-formedness is preserved.
\item No items are added to~$\waiters$ or $\locked$.
\item $\locked$, $\waiters$ are unchanged.
\item $\kwassn$ is not a mutex, so this property is preserved.
\item No items are added to~$\waiters$.
\item No items are added to~$\waiters$.
\item We have that~$\permof{\fctx'}{\cvname}{\prio_0} = \fnone$
  for all~$\prio_0$ such that~$\nple{\prio'}{\prio_0}$ by typing rule
  ~\rulename{NewCV}.
  Vertex~$u$ has no proper descendants, so if the last vertex of~$a_i$ is a
  descendant of~$u$, then~$a_i = a$
  and~$\uprio{\graph'}{u} = \uprio{\graph'}{a_i}$,
  so the premise of the second part of invariant~\ref{inv:no-signal},
  $\nple{\prio}{\uprio{\graph'}{u}}$ and~$\ple{\prio}{\uprio{\graph'}{a_i}}$,
  is a contradiction.
\item If~$\locked(\mutname') = (a, \_, \_)$,
then we have~$\prioceil{\stack}{\mutname'}{a}{\prio}{\stack^p}$
for some~$\stack^p$. This remains unchanged.
\end{enumerate}

\item \rulename{Wait}
Then~$\tp = \cthread{a}{\prio}{\sig}
                {\ssend{\stack}{\kwwait{\kwcv{\cvname}}}}
              \tpcp \tp'$
and~$\waiters' = \waitupd{\waiters}{\cvname}
                {\waitcons{(a, (u_1, u_2), \prio, \sig, [\kwtriv/x]s)}{\waiters(\cvname)}}$.
\begin{enumerate}
\item By inversion,
$\ityped{\sig}{\ctx}{\fctx}{\kwwait{\kwcv{\cvname}}}{\kwunit}{\prio}{\fctx}$
and~$\stackaccepts{\sig}{\stack}{\kwunit}{\fctx}{\prios_1}{\fctx'}{\prio}$.
By~\rulename{KS-PushInstr},
$\sstyped{\sig}{\fctx}{\sreturn{\stack}{\kwtriv}}{\prio}{\fctx'}$.
\item  No new weak edges or sync edges are created,
so well-formedness is preserved.
\item $\cvname$ is not a mutex, so~$\cvname \not\in \mathit{Dom}(\locked)$.
\item $\cvname$ is not a mutex.
\item $\cvname$ is not a mutex, so this condition does not apply.
\item $\cvname$ is not a mutex, so this condition does not apply.
\item By inversion on \rulename{Wait} and~\rulename{CV},
$\sigcvtype{\cvname}{\prio'} \in \sig$
where~$\meetc[\worlds]{\ctx}{\ple{\prio}{\prio'}}$.

\item Permissions are not split or reassigned by this step.
\item If~$\locked(\mutname') = (a, \_, \_)$,
then we have~$\prioceil{\stack}{\mutname'}{a}{\prio}{\stack^p}$
for some~$\stack^p$. This remains unchanged, although~$a$ is moved from~$\tp$
to~$\waiters'(\mutname)$.
\end{enumerate}

\item \rulename{Signal1}
Then~$\tp = \cthread{a}{\prio}{\sig}{\ssend{\stack}{
                  {\kwsignal{\kwcv{\cvname}}}}}
              \tpcp \tp_0$
and~$\tp' = \cthread{a}{\prio}{\sig}{\sreturn{\stack}{\kwtriv}} \tpcp
        \cthread{b}{\prio_b}{\sig_b}{\stackstate_b}
        \tp_0$
and~$(b, (\_, u'), \prio_b, \sig_b, \stackstate_b) \in \waiters(\cvname)$
and by inversion,~$\permof{\fctx}{\cvname}{\prio} \neq \fnone$.
\begin{enumerate}
\item By inversion,
$\sstyped{\sig_b}{\fctx_b}{\stackstate_b}{\prio_b}{\fctx_b'}$
and
$\ityped{\sig}{\ctx}{\fctx}{\kwsignal{\kwcv{\cvname}}}{\kwunit}{\prio}{\fctx}$
and~$\stackaccepts{\sig}{\stack}{\kwunit}{\fctx}{\prios_1}{\fctx'}{\prio}$.
By~\rulename{KS-PushInstr},
$\sstyped{\sig}{\fctx}{\sreturn{\stack}{\kwtriv}}{\prio}{\fctx'}$.
\item We need to show that the edge $(u, u')$ does not violate well-formedness.
It suffices to show that for all strong ancestors~$u_0$ of~$u$,
$\ple{\prio_b}{\uprio{\graph'}{u_0}}$.
By assumption,~$\sigcvtype{\cvname}{\prio_\cvname} \in \sig_b$
and~$\ple{\prio_b}{\prio_{\cvname}}$.
If~$\sanc{u_0}{u'}$, then this must already hold by well-formedness.
If~$\not\sanc{u_0}{u'}$, then $\ple{\prio_{\cvname}}{\uprio{\graph'}{u_0}}$
because otherwise, $\permof{\fctx}{\cvname}{\prio} = \fnone$.
\item No items are added to~$\waiters$ or~$\locked$.
\item $\cvname$ is not a mutex, so the condition is unchanged.
\item $\cvname$ is not a mutex, so this condition doesn't apply.
\item $\cvname$ is not a mutex, so this condition doesn't apply.
\item No items are added to~$\waiters$.
\item Permissions are not split or reassigned by this step.
\item If~$\locked(\mutname') = (a, \_, \_)$,
then we have~$\prioceil{\stack}{\mutname'}{a}{\prio}{\stack^p}$
for some~$\stack^p$. This remains unchanged.
\end{enumerate}

\item \rulename{Signal2}. Similar to \rulename{Signal1}.

\item \rulename{Promote}
Then~$\tp = \cthread{a}{\prio}{\sig}
                {\ssend{\stack}{\kwpromote{\kwcv{\cvname}}{\prio'}}}
              \tpcp \tp_0$
and~$\tp' = \cthread{a}{\prio}{\sig,\sigcvtype{\cvname}{\prio'}}
        {\sreturn{\stack}{\kwcv{\cvname}}}
        \tpcp \tp_0$.
\begin{enumerate}
\item By inversion,
$\ityped{\sig}{\ctx}{\fctx}{\kwpromote{\kwcv{\cvname}}{\prio'}}
{\kwcvt{\cvname}{\prio'}}{\prio}{\fctx'}$
and
$\stackaccepts{\sig}{\stack}{\kwcvt{\cvname}{\prio'}}{\fctx'}{\prios_1}
{\fctx''}{\prio}$
and
$\forall \prio_0. \nple{\prio_2}{\prio_0} \rightarrow
  \permof{\fctx}{\cvname}{\prio_0} = \fowned$ and
  $\forall \prio_0. \nple{\prio_2}{\prio_0} \rightarrow
  \permof{\fctx'}{\cvname}{\prio_0} = \fnone$ and
  $\forall \prio_0. \ple{\prio_2}{\prio_0} \rightarrow
  \permof{\fctx'}{\cvname}{\prio_0} = \permof{\fctx}{\cvname}{\prio_0}$.
By~\rulename{CV} and~\rulename{KS-PushInstr},
$\sstyped{\sig,\sigcvtype{\cvname}{\prio'}}{\fctx'}
{\sreturn{\stack}{\kwcv{\cvname}}}{\prio}{\fctx''}$.
\item No new weak edges or sync edges are created,
so well-formedness is preserved.
\item No items are added to~$\waiters$ or~$\locked$.
\item $\cvname$ is not a mutex, so this condition is unchanged.
\item $\cvname$ is not a mutex, so this condition doesn't apply.
\item $\cvname$ is not a mutex, so this condition doesn't apply.
\item No items are added to~$\waiters$.
\item We now have~$\sigcvtype{\cvname}{\prio'}$ in~$\sig$, so
we must check the conditions for all~$u'$ such that~$\not\sanc{u'}{u}$.
Let~$u_e$ be a descendant of~$u'$ that is also the last vertex of a thread
in~$\graph$.
If~$u_e = u$,
then its priority will be~$\prio$ and it will be typed with $\fctx''$.
For all~$\prio''$ such that~$\nple{\prio'}{\prio''}$, we have
$\permof{\fctx''}{\cvname}{\prio''} = \fnone$, as required.
If~$u_e$ is not a descendant of~$u$,
then it must have permission~$\fnone$ for~$\cvname$ at any priority~$\prio''$
such that~$\nple{\prio'}{\prio''}$ because
$\permof{\fctx}{\cvname}{\prio''} = \fowned$.
\item If~$\locked(\mutname') = (a, \_, \_)$,
then we have~$\prioceil{\stack}{\mutname'}{a}{\prio}{\stack^p}$
for some~$\stack^p$. This remains unchanged.
\end{enumerate}

\item \rulename{NewMutex}
Then~$\tp = \cthread{a}{\prio}{\sig}
                {\ssend{\stack}{\kwnewmutex{\prio'}}}
              \tpcp \tp_0$
and~$\tp' = \cthread{a}{\prio}{\sig, \sigcvtype{\mutname}{\prio'}}
        {\sreturn{\stack}{\kwmutex{\mutname}}}
        \tpcp \tp_0$.
\begin{enumerate}
\item By inversion,
$\ityped{\sig}{\ctx}{\fctx}{\kwnewmutex{\prio'}}{\kwmutext{\prio'}}{\prio}{\fctx}$
and~$\stackaccepts{\sig}{\stack}{\kwmutext{\prio'}}{\fctx}{\prios_1}{\fctx'}{\prio}$.
By~\rulename{Mutex} and~\rulename{KS-PushInstr},
$\sstyped{\sig,\sigcvtype{\mutname}{\prio'}}{\fctx}
{\sreturn{\stack}{\kwmutex{\mutname}}}{\prio}{\fctx'}$.
\item  No new sync edges are created, so well-formedness is preserved.
\item No items are added to~$\waiters$ or~$\locked$.
\item $\locked$, $\waiters$ are unchanged.
\item $\mutname$ is fresh, so~$\sigcvtype{\mutname}{\prio'}$ is the only
entry for~$\mutname$.
\item No items are added to~$\waiters$.
\item No items are added to~$\waiters$.
\item Permissions are not split or reassigned by this step.
\item If~$\locked(\mutname') = (a, \_, \_)$,
then we have~$\prioceil{\stack}{\mutname'}{a}{\prio}{\stack^p}$
for some~$\stack^p$. This remains unchanged.
\end{enumerate}

\item \rulename{Let1}
Then~$\tp = \cthread{a}{\prio}{\sig}
               {\scsend{\stack}{\kwlet{x}{\instr}{s}}{v}} \tpcp \tp_0$
and~$\tp' = \cthread{a}{\prio}{\sig}
               {\ssend{\scp{\stack}{\kwlet{x}{\shole}{s}}}{\instr}} \tpcp \tp_0$.
\begin{enumerate}
\item By inversion,
$\stackacceptsc{\sig}{\stack}{\fctx_3}{\prios_1}{\fctx_4}{\prio}$
and for all~$\prio' \in \prios_1$,
we have~$\ityped{\sig}{\ectx}{\fctx}{\instr}{\tau}{\prio'}{\fctx_2}$
and~$\styped{\sig}{\hastype{x}{\tau}}{\fctx_2}{s}{\prio'}{\fctx_3}$.
By~\rulename{KS-Let},
$\stackaccepts{\sig}{\scp{\stack}{\kwlet{x}{\shole}{s}}}{\tau}
{\fctx_2}{\prios_1}{\fctx_4}{\prio}$.
By~\rulename{KS-PopInstr},
$\sstyped{\sig}{\fctx}{\ssend{\scp{\stack}{\kwlet{x}{\shole}{s}}}{\instr}}
{\prio}{\fctx_4}$.
\item  No new sync edges are created, so well-formedness is preserved.
\item No items are added to~$\waiters$ or~$\locked$.
\item $\locked$ is unchanged.
\item Signatures are not changed.
\item No items are added to~$\waiters$.
\item No items are added to~$\waiters$.
\item Permissions are not split or reassigned by this step.
\item If~$\locked(\mutname) = (a, \_, \_)$,
then we have~$\prioceil{\stack}{\mutname}{a}{\prio}{\stack^p}$
for some~$\stack^p$. This remains unchanged.
\end{enumerate}.

\item \rulename{Let2}
Then~$\tp = \cthread{a}{\prio}{\sig}
               {\sreturn{\scp{\stack}{\kwlet{x}{\shole}{s}}}{v}} \tpcp \tp_0$
and~$\tp' = \cthread{a}{\prio}{\sig}
               {\scsend{\stack}{[v/x]s}} \tpcp \tp_0$.
\begin{enumerate}
\item By inversion,
$\stackacceptsc{\sig}{\stack}{\fctx''}{\prios_1}{\fctx'}{\prio}$
and~$\etyped{\sig}{\ectx}{v}{\tau}$
and for all~$\prio' \in \prios_1$,
we have~$\styped{\sig}{\hastype{x}{\tau}}{\fctx}{s}{\prio'}{\fctx''}$.
By substitution,~$\styped{\sig}{\ectx}{\fctx}{[v/x]s}{\prio'}{\fctx''}$.
By~\rulename{KS-PopStmt},
$\sstyped{\sig}{\fctx}{\scsend{\stack}{[v/x]s}}{\prio}{\fctx'}$.
\item  No new sync edges are created, so well-formedness is preserved.
\item No items are added to~$\waiters$ or~$\locked$.
\item $\locked$, $\waiters$ are unchanged.
\item Signatures are not changed.
\item No items are added to~$\waiters$.
\item No items are added to~$\waiters$.
\item Permissions are not split or reassigned by this step.
\item If~$\locked(\mutname) = (a, \_, \_)$,
then we have~$\prioceil{\stack}{\mutname}{a}{\prio}{\stack^p}$
for some~$\stack^p$. This remains unchanged.
\end{enumerate}.

\item \rulename{WithLockS1}
Then~$\tp = \cthread{a}{\prio}{\sig}
               {\scsend{\stack}{\kwwithlock{\kwmutex{\mutname}}{s}}}
              \tpcp \tp_0$
and~$\tp' = \cthread{a}{\prio}{\sig}
                {\scsend{\scp{\stack}{\kwwithlocks{\kwmutex{\mutname}}{\shole}}}
                  {
                    s}}
        \tpcp \tp_0$.
and~$\locked' = \waitupd{\locked}{\mutname}{(a, u_1, u_2)}$.
\begin{enumerate}
\item By inversion,
$\stackacceptsc{\sig}{\stack}{\fctx''}{\prios_1}{\fctx'}{\prio}$
and for all~$\prio' \in \prios_1$,
we have~$\styped{\sig}{\ectx}{\fctx}{\kwwithlock{\kwmutex{\mutname}}{s}}{\prio'}{\fctx''}$
By inversion on~\rulename{WithLock},
we have~$\sigcvtype{\mutname}{\prio_\mutname} \in \sig$
and~$\styped{\sig}{\ectx}{\fctx}{s}{\prio'}{\fctx''}$
and~$\styped{\sig}{\ectx}{\fctx}{s}{\prio_\mutname}{\fctx''}$.
By~\rulename{KS-WithLockS},
$\stackacceptsc{\sig}{\scp{\stack}{\kwwithlocks{\kwmutex{\mutname}}{\shole}}}
{\fctx''}{\prios_1 \cup \{\prio_\mutname\}}{\fctx'}{\prio}$.
By~\rulename{Seq} and~\rulename{KS-PopStmt},
$\sstyped{\sig}{\fctx}
{\scsend{\scp{\stack}{\kwwithlocks{\kwmutex{\mutname}}{\shole}}}
  {
    s}}{\prio}{\fctx'}$.
\item No sync edges are added, so well-formedness is preserved.
\item By assumption, $\waiters(\mutname) = \waitemp$.
\item $\locked'(\mutname) \neq \lunlocked$.
\item Signatures are not changed.
\item No items are added to~$\waiters$.
\item $\mutname$ is not a CV, so the invariant is preserved.
\item Permissions are not split or reassigned by this step.
\item Suppose~$\locked(\mutname') = (a, \_, \_)$.
If~$\mutname = \mutname'$,
then we have
\[\prioceil{\scp{\stack}{\kwwithlocks{\kwmutex{\mutname}}{\shole}}}
{\mutname}{a}{\prio}{\kwwithlocksp{\kwmutex{\mutname}}{\shole}{a}{\prio}}\]
Otherwise, the invariant is unchanged.
\end{enumerate}

\item \rulename{WithLockS2}
Then~$\tp = \cthread{a}{\prio}{\sig}
               {\scsend{\stack}{\kwwithlock{\kwmutex{\mutname}}{s}}}
              \tpcp \tp'$
and~$\waiters' =\waitupd{\waiters}{\mutname}
                {\waitcons{(a, (u_1, u_2), \prio, \sig, \stackstate')}{\waiters(\mutname)}}$
                where~$\stackstate' =
             \scsend{\scp{\stack}{\kwwithlocks{\kwmutex{\mutname}}{\shole}}}
                    {
                      s}$
and~$\locked(\mutname) = (a', u_1', u_2')$
and~$\ple{\prio}{\uprio{\graph}{a'}}$.
\begin{enumerate}
\item
By inversion,
$\stackacceptsc{\sig}{\stack}{\fctx''}{\prios_1}{\fctx'}{\prio}$
and for all~$\prio' \in \prios_1$,
we have~$\styped{\sig}{\ectx}{\fctx}{\kwwithlock{\kwmutex{\mutname}}{s}}{\prio'}{\fctx''}$
By inversion on~\rulename{WithLock},
we have~$\sigcvtype{\mutname}{\prio_\mutname} \in \sig$
and~$\styped{\sig}{\ectx}{\fctx}{s}{\prio'}{\fctx''}$
and~$\styped{\sig}{\ectx}{\fctx}{s}{\prio_\mutname}{\fctx''}$.
By~\rulename{KS-WithLockS},
$\stackacceptsc{\sig}{\scp{\stack}{\kwwithlocks{\kwmutex{\mutname}}{\shole}}}
{\fctx''}{\prios_1 \cup \{\prio_\mutname\}}{\fctx'}{\prio}$.
By~\rulename{Seq} and~\rulename{KS-PopStmt},
$\sstyped{\sig}{\fctx}
{\scsend{\scp{\stack}{\kwwithlocks{\kwmutex{\mutname}}{\shole}}}
  {
    s}}{\prio}{\fctx'}$.
\item No sync edges are added, so well-formedness is preserved.
\item By assumption, for all~$(c, (u_1'', u_2''), \prio_c,
\sig_c, s_c) \in \waiters(\mutname)$, 
we have a weak edge~$(u_1', u_2'')$ and~$\ple{\prio_c}{\uprio{\graph}{a'}}$.
We add a weak edge~$(u_1', u_2)$ to~$\graph'$, so the invariant on weak edges
is preserved.
Because~$\ple{\prio}{\uprio{\graph}{a'}}$, the invariant on priorities is
also preserved.
We have
\[\prioceil{\scp{\stack}{\kwwithlocks{\kwmutex{\mutname}}{\shole}}}
{\mutname}{a}{\prio}{\kwwithlocksp{\kwmutex{\mutname}}{\shole}{a}{\prio}}\]
so this requirement is also preserved.
\item $\locked'(\mutname) \neq \lunlocked$.
\item Signatures are not changed.
\item By inversion on \rulename{WithLock} and \rulename{Mutex},
$\sigcvtype{\mutname}{\prio_\mutname} \in \sig$
and~$\ple{\prio}{\prio_\mutname}$.

\item $\mutname$ is not a CV, so the invariant is preserved.
\item Permissions are not split or reassigned by this step.
\item If~$\locked(\mutname') = (a, \_, \_)$,
then we have~$\prioceil{\stack}{\mutname'}{a}{\prio}{\stack^p}$
for some~$\stack^p$. This remains unchanged.
\end{enumerate}

\item \rulename{WithLockS3}
Then~$\tp = \cthread{a}{\prio}{\sig}
               {\scsend{\stack}{\kwwithlock{\kwmutex{\mutname}}{s}}}
               \tpcp \cthread{b}{\prio'}{\sig'}{\stackstate_b}
                \tpcp \tp
              \tpcp \tp_0$
and~$\tp' = \cthread{b'}{\prio_{\mutname}}{\sig'}
        {\stackstate''}
        \tpcp \tp_0$
and~$\waiters' = \waitupd{\waiters}{\mutname}
                {\waitcons{(a, (u_1, u_2), \prio, \sig, \stackstate')}{\waiters(\mutname)}}$
where~$\stackstate' =
             \scsend{\scp{\stack}{\kwwithlocks{\kwmutex{\mutname}}{\shole}}}
                    {
                      s}$
and~$\locked(\mutname) = (a', u_1', u_2')$.
\begin{enumerate}
\item By inversion,
$\stackacceptsc{\sig}{\stack}{\fctx''}{\prios_1}{\fctx'}{\prio}$
and for all~$\prio' \in \prios_1$,
we have~$\styped{\sig}{\ectx}{\fctx}{\kwwithlock{\kwmutex{\mutname}}{s}}{\prio'}{\fctx''}$
By inversion on~\rulename{WithLock},
we have~$\sigcvtype{\mutname}{\prio_\mutname} \in \sig$
and~$\styped{\sig}{\ectx}{\fctx}{s}{\prio'}{\fctx''}$
and~$\styped{\sig}{\ectx}{\fctx}{s}{\prio_\mutname}{\fctx''}$.
By~\rulename{KS-WithLockS},
$\stackacceptsc{\sig}{\scp{\stack}{\kwwithlocks{\kwmutex{\mutname}}{\shole}}}
{\fctx''}{\prios_1 \cup \{\prio_\mutname\}}{\fctx'}{\prio}$.
By~\rulename{Seq} and~\rulename{KS-PopStmt},
$\sstyped{\sig}{\fctx}
{\scsend{\scp{\stack}{\kwwithlocks{\kwmutex{\mutname}}{\shole}}}
  {
    s}}{\prio}{\fctx'}$.
Also by inversion,~$\sstyped{\sig'}{\fctx_b}{\stackstate_b}{\prio'}{\fctx_b'}$.
By Lemma~\ref{lem:prio-change-pres},
$\sstyped{\sig'}{\fctx_b}{\stackstate''}{\prio_\mutname}{\fctx_b'}$.
\item No sync edges are added, so well-formedness is preserved.
\item By assumption, for all~$(c, (u_{1c}, u_{2c}), \prio_c,
\sig_c, s_c) \in \waiters(\mutname)$,
$\ple{\prio_c}{\ple{\uprio{\graph}{a'}}{\prio_{\mutname}}}$.
By inversion,~$\ple{\prio}{\prio_{\mutname}}$.
We add all of the relevant weak edges from~$u_1''$.
We have~$\prioceil{\scp{\stack}{\kwwithlocks{\kwmutex{\mutname}}{\shole}}}
{\mutname}{a}{\prio}{\kwwithlocksp{\kwmutex{\mutname}}{\shole}{a}{\prio}}$,
so this requirement is also preserved.
\item $\locked'(\mutname) \neq \lunlocked$.
\item Signatures are not changed.
\item By inversion on \rulename{WithLock} and \rulename{Mutex},
$\sigcvtype{\mutname}{\prio_\mutname} \in \sig$
and~$\ple{\prio}{\prio_\mutname}$.
\item $\mutname$ is not a CV, so the invariant is preserved.
\item Permissions are not split or reassigned by this step.
\item If~$\locked(\mutname') = (a, \_, \_)$,
then we have~$\prioceil{\stack}{\mutname'}{a}{\prio}{\stack^p}$
for some~$\stack^p$. This remains unchanged.
The priority of~$b'$ is~$\prio_\mutname$, so this property no longer applies
to~$b'$.
\end{enumerate}

\item \rulename{WithLockS4}
Then~$\tp = \cthread{a}{\prio}{\sig}
               {\scsend{\stack}{\kwwithlock{\kwmutex{\mutname}}{s}}}
                \tpcp \tp
              \tpcp \tp'$
and~$\waiters' = \waitupd{\waitupd{\waiters}{\cvmutname'}
               {\waitcons{\ell_1}{\waitcons{(b', \_, \prio_\mutname, \sig', \stackstate'')}
               {\ell_2}}}}
               {\mutname}
                {\waitcons{(a, (u_1, u_2), \prio, \sig, \stackstate')}{\waiters(\mutname)}}$
where~$\waiters(\cvmutname') = 
\waitcons{\ell_1}{\waitcons{(b, \_, \prio', \sig', \stackstate_b)}
               {\ell_2}}$
and $\stackstate' =
             \scsend{\scp{\stack}{\kwwithlocks{\kwmutex{\mutname}}{\shole}}}
                    {
                      s}$
and~$\locked(\mutname) = (a', u_1', u_2')$.
\begin{enumerate}
\item By inversion,
$\stackacceptsc{\sig}{\stack}{\fctx''}{\prios_1}{\fctx'}{\prio}$
and for all~$\prio' \in \prios_1$,
we have~$\styped{\sig}{\ectx}{\fctx}{\kwwithlock{\kwmutex{\mutname}}{s}}{\prio'}{\fctx''}$
By inversion on~\rulename{WithLock},
we have~$\sigcvtype{\mutname}{\prio_\mutname} \in \sig$
and~$\styped{\sig}{\ectx}{\fctx}{s}{\prio'}{\fctx''}$
and~$\styped{\sig}{\ectx}{\fctx}{s}{\prio_\mutname}{\fctx''}$.
By~\rulename{KS-WithLockS},
$\stackacceptsc{\sig}{\scp{\stack}{\kwwithlocks{\kwmutex{\mutname}}{\shole}}}
{\fctx''}{\prios_1 \cup \{\prio_\mutname\}}{\fctx'}{\prio}$.
By~\rulename{Seq} and~\rulename{KS-PopStmt},
$\sstyped{\sig}{\fctx}
{\scsend{\scp{\stack}{\kwwithlocks{\kwmutex{\mutname}}{\shole}}}
  {
    s}}{\prio}{\fctx'}$.
Also by inversion,~$\sstyped{\sig'}{\fctx_b}{\stackstate_b}{\prio'}{\fctx_b'}$.
By Lemma~\ref{lem:prio-change-pres},
$\sstyped{\sig'}{\fctx_b}{\stackstate''}{\prio_\mutname}{\fctx_b'}$.
\item No sync edges are added, so well-formedness is preserved.
\item By assumption, for all~$(c, (u_{1c}, u_{2c}), \prio_c,
\sig_c, s_c) \in \waiters(\mutname)$,
$\ple{\prio_c}{\ple{\uprio{\graph}{a'}}{\prio_{\mutname}}}$.
By inversion,~$\ple{\prio}{\prio_{\mutname}}$.
We add all of the relevant weak edges from~$u_1''$.
We have~$\prioceil{\scp{\stack}{\kwwithlocks{\kwmutex{\mutname}}{\shole}}}
{\mutname}{a}{\prio}{\kwwithlocksp{\kwmutex{\mutname}}{\shole}{a}{\prio}}$,
so this requirement is also preserved.
\item $\locked'(\mutname) \neq \lunlocked$.
\item Signatures are not changed.
\item By inversion on \rulename{WithLock} and \rulename{Mutex},
$\sigcvtype{\mutname}{\prio_\mutname} \in \sig$
and~$\ple{\prio}{\prio_\mutname}$.
\item $\mutname$ is not a CV, so the invariant is preserved.
\item Permissions are not split or reassigned by this step.
\item If~$\locked(\mutname') = (a, \_, \_)$,
then we have~$\prioceil{\stack}{\mutname'}{a}{\prio}{\stack^p}$
for some~$\stack^p$. This remains unchanged.
The priority of~$b'$ is~$\prio_\mutname$, so this property no longer applies
to~$b'$.
\end{enumerate}

\item \rulename{WithLockE1}
Then~$\tp = \cthread{a}{\prio}{\sig}
               {\screturn{\scp{\stack}{\kwwithlocks{\kwmutex{\mutname}}{\shole}}}}
              \tpcp \tp_0$
and~$\tp' = \cthread{a}{\prio}{\sig}{\screturn{\stack}} \tpcp
                \tp_0$.
\begin{enumerate}
\item By inversion,
$\stackacceptsc{\sig}{\stack}{\fctx}{\prios_1}{\fctx'}{\prio}$.
By~\rulename{KS-PushStmt},
$\sstyped{\sig}{\fctx}{\screturn{\stack}}{\prio}{\fctx'}$.
\item No sync edges are added, so well-formedness is preserved.
\item $\locked'(\mutname) = \lunlocked$.
\item $\waiters' = \waiters$ and, by assumption,~$\waiters(\mutname) = \waitemp$.
\item Signatures are not changed.
\item No items are added to~$\waiters$.
\item $\mutname$ is not a CV, so the invariant is preserved.
\item Permissions are not split or reassigned by this step.
\item If~$\locked(\mutname') = (a, \_, \_)$,
then we have~$\prioceil{\stack}{\mutname'}{a}{\prio}{\stack^p}$
for some~$\stack^p$. This remains unchanged.
\end{enumerate}

\item \rulename{WithLockE2}
Then~$\tp = \cthread{a}{\prio}{\sig}
        {\screturn{\scp{\stack}{\kwwithlocks{\kwmutex{\mutname}}{\shole}}}}
              \tpcp \tp_0$
and~$\tp' = \cthread{a}{\prio}{\sig}{\screturn{\stack}} \tpcp
        \cthread{b}{\prio_b}{\sig_b}{\stackstate_b}\tpcp
                \tp_0$.
\begin{enumerate}
\item By inversion,
$\stackacceptsc{\sig}{\stack}{\fctx}{\prios_1}{\fctx'}{\prio}$.
By~\rulename{KS-PushStmt},
$\sstyped{\sig}{\fctx}{\screturn{\stack}}{\prio}{\fctx'}$.
We also have~$\sstyped{\sig_b}{\fctx_b}{\stackstate_b}{\prio_b}{\fctx_b'}$.
\item With the added sync edge, we now have~$\sanc{u}{u_2}$.
By inversion on \rulename{Unlock} and \rulename{Mutex},
$\sigcvtype{\mutname}{\prio} \in \sig$.
By invariant~\ref{inv:waiters-prio},~$\ple{\prio_b}{\prio}$.
If~$\sanc{u_0}{u}$ and~$u_0$ is not
an ancestor of the start of~$u$'s thread,
then~$\ple{\uprio{\graph'}{u}}{\uprio{\graph'}{u_0}}$ by well-formedness,
so~$\ple{\prio_b}{\uprio{\graph'}{u_0}}$.
This addresses the first component of well-formedness.
We also have the strong edge~$(u_a, u_b)$, where~$\locked(\mutname) =
(a, u_a, u_b)$, and~$\wanc{u_a}{u_2}$ by assumption.
Also by assumption,~$\edgeweak{u_a}{u_2} \in \graph$, so this satisfies the
second component of well-formedness.
\item $\locked'(\mutname) = (b, u_1, u_2)$ and we add all of the requisite
edges to~$\graph'$.
By assumption,
$\forall (\_, \_, \prio_c, \_, \_) \in \ell_1\ell_2 = \waiters'(\mutname),
\ple{\prio_c}{\prio_b}$,
so the invariant on priorities is preserved.
\item $\locked'(\mutname) \neq \lunlocked$.
\item Signatures are not changed.
\item No items are added to~$\waiters$.
\item $\mutname$ is not a CV, so the invariant is preserved.
\item Permissions are not split or reassigned by this step.
\item If~$\locked'(\mutname') = (a, \_, \_)$, then~$\mutname' \neq \mutname$.
We have~$\prioceil{\stack}{\mutname'}{a}{\prio}{\stack^p}$
for some~$\stack^p$. This remains unchanged.
In the case of~$\locked'(\mutname) = (b, u_1, u_2)$,
we have~$\prioceil{\stackstate_b}{\mutname'}{a}{\prio}{\stackstate^p}$
by invariant~\ref{inv:waiters}.
\end{enumerate}

\item \rulename{WithLockE3}
Then~$\tp = \cthread{a}{\prio}{\sig}
        {\screturn{\scp{\stack}{\kwwithlocksp{\kwmutex{\mutname}}{\shole}{a'}{\prio'}}}}
              \tpcp \tp_0$
and~$\tp' = \cthread{a'}{\prio'}{\sig}{\screturn{\stack}} \tpcp
                \tp_0$.
\begin{enumerate}
\item By inversion,
$\stackacceptsc{\sig}{\stack}{\fctx}{\prios_1}{\fctx'}{\prio'}$.
By~\rulename{KS-PushStmt},
$\sstyped{\sig}{\fctx}{\screturn{\stack}}{\prio'}{\fctx'}$.

\item
The edge~$(u, u')$ preserves the well-formedness because,
by assumption,~$\ple{\uprio{\graph'}{u'} = \prio'}{\prio = \uprio{\graph'}{u}}$.

\item $\locked'(\mutname) = \lunlocked$.
\item $\waiters' = \waiters$ and, by assumption,~$\waiters(\mutname) = \waitemp$.
\item Signatures are not changed.
\item No items are added to~$\waiters$.
\item $\mutname$ is not a CV, so the invariant is preserved.
\item Permissions are not split or reassigned by this step.
\item If~$\locked(\mutname') = (a, \_, \_)$,
then we have~$\prioceil{\stack}{\mutname'}{a}{\prio}{\stack^p}$
for some~$\stack^p$. This remains unchanged.
\end{enumerate}

\item \rulename{WithLockE4}
Then~$\tp = \cthread{a}{\prio}{\sig}
        {\screturn{\scp{\stack}{\kwwithlocksp{\kwmutex{\mutname}}{\shole}{a'}{\prio'}}}}
              \tpcp \tp_0$
and~$\tp' = \cthread{a'}{\prio'}{\sig}{\screturn{\stack}} \tpcp
        \cthread{b}{\prio_b}{\sig_b}{\stackstate_b} \tpcp
                \tp_0$.
\begin{enumerate}
\item By inversion,
$\stackacceptsc{\sig}{\stack}{\fctx}{\prios_1}{\fctx'}{\prio'}$.
By~\rulename{KS-PushStmt},
$\sstyped{\sig}{\fctx}{\screturn{\stack}}{\prio'}{\fctx'}$.
We also have~$\sstyped{\sig_b}{\fctx_b}{\stackstate_b}{\prio_b}{\fctx_b'}$.
\item With the added sync edge~$(u, u_2)$, we now have~$\sanc{u}{u_2}$.
By inversion on \rulename{Unlock} and \rulename{Mutex},
$\sigcvtype{\mutname}{\prio} \in \sig$.
By assumption,~$\ple{\prio_b}{\prio}$. If~$\sanc{u_0}{u}$ and~$u_0$ is not
an ancestor of the start of~$u$'s thread,
then~$\ple{\uprio{\graph'}{u}}{\uprio{\graph'}{u_0}}$ by well-formedness,
so~$\ple{\prio_b}{\uprio{\graph'}{u_0}}$.
This addresses the first component of well-formedness.
The edge~$(u, u')$ also preserves the invariant because,
by assumption~$\ple{\uprio{\graph'}{u'} = \prio'}{\prio = \uprio{\graph'}{u}}$,
and the same argument as above applies to strong ancestors of~$u$.
We also have the strong edge~$(u_a, u_b)$, where~$\locked(\mutname) =
(a, u_a, u_b)$, and~$\wanc{u_a}{u_2}$ by assumption.
Also by assumption,~$\edgeweak{u_a}{u_2} \in \graph$, so this satisfies the
second component of well-formedness.
\item $\locked'(\mutname) = (b, u_1, u_2)$ and we add all of the requisite
edges to~$\graph'$.
By assumption,
$\forall (\_, \_, \prio_c, \_, \_) \in \ell_1\ell_2 = \waiters'(\mutname),
\ple{\prio_c}{\prio_b}$,
so the invariant on priorities is preserved.
\item $\locked'(\mutname) \neq \lunlocked$.
\item Signatures are not changed.
\item No items are added to~$\waiters$.
\item $\mutname$ is not a CV, so the invariant is preserved.
\item Permissions are not split or reassigned by this step.
\item If~$\locked'(\mutname') = (a, \_, \_)$, then~$\mutname' \neq \mutname$.
We have~$\prioceil{\stack}{\mutname'}{a}{\prio}{\stack^p}$
for some~$\stack^p$. This remains unchanged.
In the case of~$\locked'(\mutname) = (b, u_1, u_2)$,
we have~$\prioceil{\stackstate_b}{\mutname'}{a}{\prio}{\stackstate^p}$
by assumption.
\end{enumerate}

\item \rulename{If1}.
Then~$\tp = \cthread{a}{\prio}{\sig}{\scsend{\stack}{\kwif{\kwn}{s_1}{s_2}}} \tpcp \tp_0$
and~$\tp' = \cthread{a}{\prio}{\sig}{\scsend{\stack}{s_1}} \tpcp \tp_0$.
\begin{enumerate}
\item By inversion,
$\stackacceptsc{\sig}{\stack}{\fctx''}{\prios_1}{\fctx'}{\prio}$
and for all~$\prio' \in \prios_1$,
we have~$\styped{\sig}{\ectx}{\fctx}{\kwif{\kwn}{s_1}{s_2}}{\prio'}{\fctx''}$.
By inversion on~\rulename{If},
$\styped{\sig}{\ectx}{\fctx}{s_1}{\prio'}{\fctx''}$.
By \rulename~\rulename{KS-PopStmt},
$\sstyped{\sig}{\fctx}{\scsend{\stack}{s_1}}{\prio}{\fctx'}$.
\item No sync edges are added, so well-formedness is preserved.
\item No items are added to~$\waiters$ or~$\locked$.
\item $\locked$, $\waiters$ are not changed.
\item No signatures are altered.
\item No items are added to~$\waiters$.
\item No items are added to~$\waiters$.
\item Permissions are not split or reassigned by this step.
\item If~$\locked(\mutname) = (a, \_, \_)$,
then we have~$\prioceil{\stack}{\mutname}{a}{\prio}{\stack^p}$
for some~$\stack^p$. This remains unchanged.
\end{enumerate}

\item \rulename{If2}. Similar to \rulename{If1}.

\item \rulename{While}.
Then~$\tp = \cthread{a}{\prio}{\sig}{\scsend{\stack}{\kwwhile{v}{s}}} \tpcp \tp_0$
and~$\tp' = \cthread{a}{\prio}{\sig}
{\scsend{\stack}{\kwif{v}{(s; \kwwhile{v}{s})}{\kwskip}}} \tpcp \tp_0$.
\begin{enumerate}
\item By inversion,
$\stackacceptsc{\sig}{\stack}{\fctx}{\prios_1}{\fctx'}{\prio}$
and for all~$\prio' \in \prios_1$,
we have~$\styped{\sig}{\ectx}{\fctx}{\kwwhile{v}{s}}{\prio'}{\fctx}$.
By inversion on~\rulename{While},~$\etyped{\sig}{\ectx}{v}{\kwnat}$
and $\styped{\sig}{\ectx}{\fctx}{s}{\prio}{\fctx}$.
By~\rulename{Seq},
$\styped{\sig}{\ectx}{\fctx}{s; \kwwhile{v}{s}}{\prio}{\fctx}$.
By~\rulename{Skip} and~\rulename{If},
$\styped{\sig}{\ectx}{\fctx}{\kwif{v}{(s; \kwwhile{v}{s})}{\kwskip}}{\prio}{\fctx}$.
By~\rulename{KS-PopStmt},
$\sstyped{\sig}{\fctx}{\scsend{\stack}{\kwif{v}{(s; \kwwhile{v}{s})}{\kwskip}}}
{\prio}{\fctx'}$.
\item No sync edges are added, so well-formedness is preserved.
\item No items are added to~$\waiters$ or~$\locked$.
\item $\locked$, $\waiters$ are not changed.
\item No signatures are altered.
\item No items are added to~$\waiters$.
\item No items are added to~$\waiters$.
\item Permissions are not split or reassigned by this step.
\item If~$\locked(\mutname) = (a, \_, \_)$,
then we have~$\prioceil{\stack}{\mutname}{a}{\prio}{\stack^p}$
for some~$\stack^p$. This remains unchanged.
\end{enumerate}

\item \rulename{Skip}.
\begin{enumerate}
\item By inversion,
$\stackacceptsc{\sig}{\stack}{\fctx}{\prios_1}{\fctx'}{\prio}$.
By~\rulename{KS-PushStmt},
$\sstyped{\sig}{\fctx}{\screturn{\stack}}{\prio}{\fctx'}$.
\item No sync edges are added, so well-formedness is preserved.
\item No items are added to~$\waiters$ or~$\locked$.
\item $\locked$, $\waiters$ are not changed.
\item No signatures are altered.
\item No items are added to~$\waiters$.
\item No items are added to~$\waiters$.
\item Permissions are not split or reassigned by this step.
\item If~$\locked(\mutname) = (a, \_, \_)$,
then we have~$\prioceil{\stack}{\mutname}{a}{\prio}{\stack^p}$
for some~$\stack^p$. This remains unchanged.
\end{enumerate}


\item \rulename{Seq1}.
Then~$\tp = \cthread{a}{\prio}{\sig}
               {\scsend{\stack}{s_1; s_2}} \tpcp \tp_0$
and~$\tp' = \cthread{a}{\prio}{\sig}
               {\scsend{\scp{\stack}{\shole; s_2}}{s_1}}\tpcp \tp_0$.
\begin{enumerate}
\item By inversion,
$\stackacceptsc{\sig}{\stack}{\fctx_2}{\prios_1}{\fctx'}{\prio}$
and for all~$\prio' \in \prios_1$,
we have~$\styped{\sig}{\ectx}{\fctx}{s_1}{\prio'}{\fctx_1}$
and~$\styped{\sig}{\ectx}{\fctx_1}{s_2}{\prio'}{\fctx_2}$.
By~\rulename{KS-Seq},
$\stackacceptsc{\sig}{\scp{\stack}{\shole; s_2}}{\fctx_1}{\prios_1}{\fctx'}{\prio}$.
By~\rulename{KS-PopStmt},
$\sstyped{\sig}{\fctx}{\scsend{\scp{\stack}{\shole; s_2}}{s_1}}{\prio}{\fctx'}$.
\item No sync edges are added, so well-formedness is preserved.
\item No items are added to~$\waiters$ or~$\locked$.
\item $\locked$, $\waiters$ are not changed.
\item No signatures are altered.
\item No items are added to~$\waiters$.
\item No items are added to~$\waiters$.
\item Permissions are not split or reassigned by this step.
\item If~$\locked(\mutname) = (a, \_, \_)$,
then we have~$\prioceil{\stack}{\mutname}{a}{\prio}{\stack^p}$
for some~$\stack^p$. This remains unchanged.
\end{enumerate}

\item \rulename{Seq2}.
Then~$\tp = \cthread{a}{\prio}{\sig}
               {\screturn{\scp{\stack}{\shole; s_2}}} \tpcp \tp_0$
and~$\tp' = \cthread{a}{\prio}{\sig}{\scsend{\stack}{s_2}} \tpcp \tp$.
\begin{enumerate}
\item By inversion,
$\stackacceptsc{\sig}{\stack}{\fctx''}{\prios_1}{\fctx'}{\prio}$
and for all~$\prio' \in \prios_1$,
we have~$\styped{\sig}{\ectx}{\fctx}{s_2}{\prio'}{\fctx''}$.
By~\rulename{KS-PopStmt},
$\sstyped{\sig}{\fctx}{\scsend{\stack}{s_2}}{\prio}{\fctx'}$.
\item No sync edges are added, so well-formedness is preserved.
\item No items are added to~$\waiters$ or~$\locked$.
\item $\locked$, $\waiters$ are not changed.
\item No signatures are altered.
\item No items are added to~$\waiters$.
\item No items are added to~$\waiters$.
\item Permissions are not split or reassigned by this step.
\item If~$\locked(\mutname) = (a, \_, \_)$,
then we have~$\prioceil{\stack}{\mutname}{a}{\prio}{\stack^p}$
for some~$\stack^p$. This remains unchanged.
\end{enumerate}

\end{itemize}
\else
A full proof is available in the supplementary material.
\fi
\end{proof}

Next, we prove a Progress result stating that if a configuration meets the
invariants, then every non-blocked thread can take a step according to the
dynamic semantics.

\begin{theorem}\label{thm:progress}
Suppose~$\lconfig{\tp}{\mem}{\waiters}{\locked}{\graph}$ meets the invariants
of Definition~\ref{def:invs}.
Then for all~$\cthread{a}{\prio}{\sig}{\stackstate} \in \tp$,
either~$\stackstate = \screturn{\estack}$ or
$\lconfig{\cthread{a}{\prio}{\sig}{\stackstate} \tpcp \tp_0}{\mem}{\waiters}{\locked}
{\graph}
\mstep
\lconfig{\tp'}{\mem'}{\waiters'}{\locked'}{\graph'}$
where~$\cthread{a}{\prio}{\sig}{\stackstate} \tpcp \tp_0 = \tp$.
\end{theorem}
\begin{proof}
By the invariants and inversion on \rulename{Global}, we have
$\sstyped{\sig}{\fctx}{\stackstate}{\prio}{\fctx'}$.
Proceed by induction on this derivation.
A full proof is available in the supplementary material.
\iffull
\begin{itemize}
\item \rulename{KS-PopInstr}.
Then~$\stackstate = \ssend{\stack}{\instr}$
and~$\stackaccepts{\sig}{\stack}{\tau}{\fctx''}{\prios_1}{\fctx'}{\prio}$
and~$\ityped{\sig}{\ectx}{\fctx}{\instr}{\tau}{\prio}{\fctx''}$.
Proceed by nested induction on the latter derivation.
\begin{itemize}
\item \rulename{Spawn}. Apply \rulename{Spawn}.
\item \rulename{NewRef}. Apply \rulename{NewRef}.
\item \rulename{Deref}. Then because the context is empty, by inversion
on~\rulename{RefVal}, $\instr = \kwderef{\kwref{\kwassn}}$
and~$\sigrtype{\kwassn}{\tau} \in \sig$.
By memory typing,~$\kwassn \in \dom{\mem}$.
Apply \rulename{Deref}.
\item \rulename{Update}. Apply \rulename{Update}.
\item \rulename{Wait}. Then because the context is empty, by inversion
on~\rulename{CV}, $\instr = \kwwait{\kwcv{\cvname}}$.
Apply~\rulename{Wait}.
\item \rulename{Signal}. Then because the context is empty, by inversion
on~\rulename{CV}, $\instr = \kwsignal{\kwcv{\cvname}}$.
Apply~\rulename{Signal1} or~\rulename{Signal2} depending on
whether~$\waiters(\cvname)$ is empty.
\item \rulename{Promote}. Then because the context is empty, by inversion
on~\rulename{CV}, $\instr = \kwpromote{\kwcv{\cvname}}{\prio_2}$.
Apply~\rulename{Promote}.
\item \rulename{NewCV}. Apply \rulename{NewCV}.
\item \rulename{NewMutex}. Apply \rulename{NewMutex}.
\end{itemize}

\item \rulename{KS-PopStmt}.
Then~$\stackstate = \scsend{\stack}{s}$
and~$\stackacceptsc{\sig}{\stack}{\fctx''}{\prios_1}{\fctx'}{\prio}$
and~$\styped{\sig}{\ectx}{\fctx}{s}{\prio}{\fctx''}$.
Proceed by nested induction on the latter derivation.
\begin{itemize}
\item \rulename{WithLock}. Then because the context is empty, by inversion
on~\rulename{Mutex}, $s = (\kwwithlock{\kwmutex{\mutname}}{s'})$.
If~$\locked(\mutname) = \lunlocked$, apply \rulename{WithLockS1}.
Otherwise, if~$\locked(\mutname) = (a', u_1', u_2')$
and~$\ple{\prio}{\uprio{\graph}{a'}}$, then apply \rulename{WithLockS2}.
Otherwise, by invariant~\ref{inv:mutex-prio}, there exists~$\prio_{\mutname}$ such that
$\sigcvtype{\mutname}{\prio_\mutname} \in \sig$.
We must have~$\uprio{\graph}{a'} \neq \prio_\mutname$
because~$\ple{\prio}{\prio_\mutname}$ and~$\nple{\prio}{\uprio{\graph}{a'}}$.
Therefore, by the invariants,
there is a~$\cthread{b}{\prio'}{\sig'}{\stackstate_b}$
in~$\tp_0$ or~$\waiters$
such that~$\prioceil{\stackstate_b}{\mutname}{b}{\prio'}{\stackstate''}$
and we can apply \rulename{WithLockS3} or \rulename{WithLockS4} respectively.
\item \rulename{Let}. Apply~\rulename{Let1}.
\item \rulename{If}. Then~$s = (\kwif{v}{s_1}{s_2})$.
Because the context is empty, by inversion on~\rulename{$\kw{nat}$I},
$v = \kwn$ for some~$n$. Apply \rulename{If1} or \rulename{If2}.
\item \rulename{While}. Apply \rulename{While}.
\item \rulename{Skip}. Apply \rulename{Skip}.
\item \rulename{Seq}. Apply \rulename{Seq1}.
\end{itemize}

\item \rulename{KS-PushInstr}.
Then~$\stackstate = \sreturn{\stack}{v}$
and~$\stackaccepts{\sig}{\stack}{\tau}{\fctx}{\prios_1}{\fctx'}{\prio}$
and~$\etyped{\sig}{\ectx}{v}{\tau}$.
By inversion,~$\stack = \scp{\stack_0}{\kwlet{x}{\shole}{s}}$.
Apply \rulename{Let2}.

\item \rulename{KS-PushStmt}.
Then~$\stackstate = \screturn{\stack}$
and~$\stackacceptsc{\sig}{\stack}{\fctx}{\prios_1}{\fctx'}{\prio}$.
Proceed by nested induction on this derivation.
\begin{itemize}
\item \rulename{KS-Empty}. Then~$\stackstate = \screturn{\estack}$.
\item \rulename{KS-Seq}. Apply \rulename{Seq2}.
\item \rulename{KS-WithLockS}.
Then because the context is empty, by inversion
on~\rulename{Mutex},
$\stack = \scp{\stack_0}{\kwwithlocks{\kwmutex{\mutname}}{\shole}}$.
Apply \rulename{WithLockE1} or \rulename{WithLockE2} depending on
whether~$\waiters(\mutname)$ is empty.
\item \rulename{KS-WithLockSP}.
Then because the context is empty, by inversion
on~\rulename{Mutex},
$\stack = \scp{\stack_0}{\kwwithlocksp{\kwmutex{\mutname}}{\shole}}{a'}{\prio'}$.
Apply \rulename{WithLockE3} or \rulename{WithLockE4} depending on
whether~$\waiters(\mutname)$ is empty.
\end{itemize}
\end{itemize}
\fi
\end{proof}

It is a direct consequence of Progress and Preservation that any configuration
reachable from an initial configuration can continue to make progress (or has
finished executing) and contains a well-formed graph.

\begin{theorem}[Soundness]
Suppose~$\sstyped{\esig}{\fctx}{\stackstate}{\prio}{\fctx'}$,
where~$\fctx$ and~$\fctx'$ are valid permission mappings,
and~$\lconfig{\cthread{a}{\prio}{\esig}{\stackstate}}{\emem}{\emptyset}{\emptyset}{\egraph}
\mstep^* \lconfig{\tp'}{\mem'}{\waiters'}{\locked'}{\graph'}$.
Then~$\graph'$ is well-formed and
for all~$\cthread{a}{\prio}{\sig}{\stackstate} \in \tp$,
either~$\stackstate = \screturn{\estack}$ or
$\lconfig{\cthread{a}{\prio}{\sig}{\stackstate} \tpcp \tp'_0}{\mem'}{\waiters'}
          {\locked'}
{\graph'}
\mstep
\lconfig{\tp''}{\mem''}{\waiters''}{\locked''}{\graph''}
$,
where~$\tp' = \cthread{a}{\prio}{\sig}{\stackstate} \tpcp \tp'_0$.
\end{theorem}
\begin{proof}
By \rulename{Global}, the initial configuration is well-typed.
All other invariants are met by an initial condition.
The result holds by inductive application of Theorems~\ref{thm:preservation}
and~\ref{thm:progress}.
\end{proof}

\subsection{Extensions}\label{sec:extensions}
Below, we discuss two synchronization operations that are not
currently modeled in {\calcname} but could be added without difficulty.
We have excluded them thus far in the interest of
keeping the semantics and proofs as simple as possible and focusing on the
key points.

\newcommand{\kwtrywith}{\kw{trywith}}
\paragraph{Trylock.} In many implementations of mutex-based synchronization,
\texttt{trylock} is a nonblocking construct that attempts to acquire a mutex;
if the mutex is already locked, \texttt{trylock} returns immediately with a
return value or error code indicating that it failed to acquire the mutex.
We could model a variant of~$\kwwithlock{v}{s}$ that does something similar;
we can call it~$\kwtrywith$.
In the context of our syntax, on failure,~$\kwtrywith$ would run an alternative
statement or set a designated variable to indicate failure.
The restrictions on the use of~$\kwtrywith$ would be identical to those
for~\texttt{with}.
Its dynamic semantics would include a rule corresponding
to~\rulename{WithLockS1}, as well as an additional rule to perform the
desired alternative action if~$\locked(\mutname) \neq \lunlocked$.

\paragraph{Broadcast.} The \texttt{signal} operation
wakes up one thread waiting on the CV.
Many implementations of CVs also include a \texttt{broadcast}
primitive that wakes up all waiting threads.
This operation would be straightforward to include in {\calcname}; its typing
restrictions would be identical to those of \texttt{signal}.
The dynamic semantics rule (analogous to \rulename{Signal1}) would add all
threads from~$\waiters(\cvname)$ back to the thread pool and add sync edges
from the current thread to all waiting threads.

\fi
\section{Implementations and Case Studies}

This section describes our Rust and C++ implementations of the type system
rules and the corresponding case studies we conducted to 
evaluate the usability of our type system rules.  
Both implementations expose a library with the ability to spawn threads at
various priorities and perform synchronization, and both statically enforce
versions of the restrictions on synchronization primitives
discussed in Section~\ref{sec:overview}.
We have conducted case studies using three application benchmarks, two
developed from scratch and one real-world interactive application,
the Memcached object caching server~\cite{Memcached09} (v1.5.13), which we
ported to use our C++ type system. 
We discuss the Rust implementation first, followed by the C++ implementation,
focusing on where it differs from the Rust one, and then the case studies
illustrating our experience working the type systems.

\subsection{Rust Implementation}\label{sec:rust}

The Rust implementation focuses on restricting the use of CVs using Rust's
rich type system, which includes notions of ownership and affinity; it does
not enforce the restrictions on mutexes.
%
%
%
Our library provides wrappers for threading and CV operations and checks the
relevant restrictions on priorities and ownership before invoking equivalent
operations of the standard Rust threading library.
%
%
The full library consists of 551 lines of code.\footnote{LoC figures are
measured by \texttt{tokei}.}

To encode the seven restrictions of Section~\ref{sec:overview}, our library
needs to represent priorities and ownership.
Priorities are represented as types that implement the trait \code{Priority}.
Priorities are simply markers that are never instantiated; they are used
only as type parameters for other generic
constructs.  Another trait, \code{Ge} (greater-or-equal), is used to establish
a type-level partial ordering between these types.  The priority of a thread
is indicated with a \code{Token} type, parameterized over a priority.  Each
thread has a token, and some library functions require a reference to a token
as proof of the current thread's priority.  The \code{Token} type is defined
such that safe code cannot access another thread's token, and a new token can
only be created by spawning a thread to go with it.

The library also defines types for each ownership level
(\code{Owned}, \code{Shared}, and \code{NoAccess})
with the~\code{Ownership} trait.
These ownership levels should not be confused with the ownership features of
Rust's type system itself (which they resemble); these types make ownership
explicit in the type of the CV.
%
%
%
%
%
Each CV, of type \code{Condvar}, is parameterized with its priority level,
as well as one of these Ownership levels for each priority.
This gives rise to the limitation that the
number of priorities must be fixed by the library; for our examples, three
(\code{Low}, \code{Medium}, and \code{High}) are sufficient.
We define a \code{split} method that, using Rust's affine-by-default type system,
ensures each of the three ownership priorities follows the appropriate rules
(e.g., \code{Owned} status implies that no other ``handle'' to this CV has
\code{Owned} or \code{Shared} access at that same priority).
%
%
Rust's affine type system also ensures that the CV handle returned by these
operations is used in the future and others (e.g., old handles that have
been split) are invalid.
The priority and ownership(s) of a CV allow us to enforce the appropriate
restrictions on uses of CVs.
%
An important aspect of enforcing these restrictions is that sharing of
\code{Condvar}s across threads must be restricted.
We implement this using another trait, \code{PrioSend}, and various functions
that allow for passing CVs across threads in safe ways.

As an example, Figure~\ref{fig:rust-ex} shows simplified Rust code for the
ill-typed producer-consumer example of Figure~\ref{fig:pc-terr} (only the body
of the main function is shown; the code of the producer and consumer threads
have been stripped to include only the signal and wait).
The main thread is initialized with a token~\code{token} of type
\code{token<Low>}.
On line 1, we use this token to create a CV and immediately promote it to
priority~\code{High}.
Rust is able to infer the type
\code{CV<High, NoAccess, NoAccess, Owned>} for~\code{cv}, indicating
a High-priority CV owned at priority High with no ownership at Low or Medium,
as required by restriction~4.
On line 2, we split the CV into two handles to be passed to the two threads.
The~\code{hollow\_split} method gives one of these handles,~\code{cv\_to\_c},
no ownership.
We then pass this CV to the consumer thread (at priority
High as indicated by the type of its token), which waits on it.
The other CV handle is passed (on lines 5-6) to the producer thread, also
at priority \code{High}, which signals it using the~\code{notify\_one} method.
This spawn should be ill-typed by restriction~3 because the main thread
has no ownership at priority \code{Low}, and indeed, when the code is compiled,
line 5 triggers the error shown at the bottom of the figure.

\begin{figure}
  \begin{lstlisting}[morekeywords={let,new,::}]
let cv = Condvar::new(token).promote::<High>();
let (cv_to_p, cv_to_c) = cv.hollow_split();
let c_th = prio::spawn_and_pass_hollow(cv_to_c, |cv, token: &Token<High>|
                                                  {cv.dummy_wait(token);});
let (cv2, p_th) = prio::spawn_and_split(cv_to_p, token, |cv, token: &Token<High>|
                                                          {cv.notify_one(token);});
\end{lstlisting}
\code{error[E0277]: the trait bound `NoAccess: Partial` is not satisfied}
\caption{A skeleton of the ill-typed producer-consumer example in Rust,
  and the resulting error message.}
  \label{fig:rust-ex}
\end{figure}

\paragraph{Discrepancies and Limitations.}
Our Rust implementation diverges in small ways from the type system discussed
in \secref{lang}.  First, the Rust library is limited to a fixed number
of priorities.
%
In addition, ownership
is explicit in the type of the CV (although Rust is able to
infer some of these types automatically) and the
programmer must manually invoke methods to appropriately split and transfer
ownership.    
Finally, it is possible to use \code{unsafe} code to sidestep the thread safety
mechanisms used to enforce some of the library's restrictions.

\subsection{C++ Implementation}\label{sec:cpp}

We also implemented a C++
library (consisting of 1,252 lines of code) that approximates the features and restrictions of {\calcname};
in addition to the features of the Rust library, the C++ implementation also
enforces our restrictions on mutexes.
The C++ library defines wrappers around threading and synchronization features
provided by I-Cilk~\cite{MullerSiGoAcAgLe20, SingerGoMuAgLeAc20}.
The C++ library represents priorities using a strategy drawn from
prior work~\cite{MullerSiGoAcAgLe20}:  each priority is represented
as a \code{class} and the relationship between two priorities is captured
through class hierarchy via inheritance.
Every thread has a priority
known at compile time, initialized by invoking a parameterized function via
\spawn
with the appropriate priority type.
Similarly, every CV is also
initialized with its own priority type and every mutex is initialized with its
ceiling priority type.  When a thread operates on a mutex or CV, the type system
checks for priority inversions according to the restrictions discussed in
\secref{overview} by using static asserts and \code{is\_base\_of} on the
priority types of the thread and the mutex or CV.
Because every critical section needs to be type checked with both the priority
of the current thread and the priority ceiling of the mutex, we
require the programmer to lift any critical section into its own parameterized
function so it can be explicitly instantiated with both priority levels.

Encoding the notion of ownership into the type system is more complex, as C++
has a non-affine type system and does not inherently support static ownership
checking.
Ownerships are represented as \texttt{enums} (\code{none},
\code{shared}, or \code{owned}).  When a CV is created, the CV is initialized
with a priority and a list of ownership states, one for each priority level
used in the program. CVs are implemented with variadic templates that allow
different programs to initialize CVs with varying max priorities so, unlike
in the Rust library, the number of priorities is not fixed by the library
(though it must still be known at compile time).
%
This templated CV serves as a wrapper, henceforth referred to as the
\code{wrapper\_CV}, with the appropriate priority and ownership types.

As in Rust, we define explicit \emph{ownership-changing
operations} on \code{wrapper\_CV}s such as ownership splitting, transferring,
and promotion on the CV, and such operations invalidate old \code{wrapper\_CV}s
that represent previous ownership levels.
Doing so is a major challenge of the C++ implementation, as it is not
automatically enforced as in Rust.
%
Change of ownership types is achieved by statically invalidating the
original \code{wrapper\_CV} variable in the current lexical scope of thread $a$
and creating a new \code{wrapper\_CV} variable associated with the same
underlying CV to be used after the \spawn statement.  To invalidate the
original \code{wrapper\_CV} variable, we utilize a compile-time
counter \code{counter} that can be incremented at compile time.  The
\code{counter} comes with an increment method \code{next} that evaluates to a integer 
literal at compile time, and the value is incremented by one each time 
\code{next} is encountered in the current translation unit.

We use these compile-time counters as ``poison flags'' to indicate the validity of a
given \code{wrapper\_CV} variable.  Whenever a new \code{wrapper\_CV} variable
is introduced (whether via explicit declaration by the programmer or via
ownership-changing macros), a corresponding local variable of type
\code{counter} associated with the \code{wrapper\_CV} is also introduced in
the same lexical scope.  The counter starts out with value $0$, indicating a
valid \code{wrapper\_CV}.  When any ownership-changing macro is invoked that
invalidates the \code{wrapper\_CV}, the macro also invokes \code{next} on the
\code{counter}, bringing its value to $1$, indicating an invalid
\code{wrapper\_CV}.  Any operations on the \code{wrapper\_CV}---signaling or
invoking ownership-changing macros---use static assert to ensure that the
\code{wrapper\_CV} is still valid and to perform the necessary checks
discussed in \secref{overview}.

\paragraph{Discrepancies and Limitations.}
Many of the discrepancies between {\calcname} and the Rust implementation
(e.g., explicit ownership changes),
are also present in the C++ implementation.
In addition, in {\calcname}, \code{if}
statements may alter the ownership map provided that both the conditional and
else branches alter it the same way (\rulename{If} in \secref{lang}
\figref{stmt-statics}), whereas in C++ the programmer must explicitly invoke
the ownership-changing macros before the \code{if} statement if the branches
require ownership change.  This is because we have no good ways to check
whether the two branches alter the ownership map in the same way otherwise.

To enforce typing rules, we provide additional facilities and impose
certain programming restrictions.  First, to ensure a CV is passed with
an ownership transfer, operations invoked on the \code{wrapper\_CV} perform
name mangling.  Thus, the programmer must use the appropriate macros to
transfer ownership when spawning a thread with a CV and invoke a macro upon
function entry to enable the use of the CV.  Without either operation, the
code will fail to compile.  Second, to allow for a CV to be split within a
loop, the library provides an alternate macro for splitting that does not
invalidate the input \code{wrapper\_CV} and asserts that the
split does not change the ownership of the input.  Finally, to ensure that the
if-else and looping constructs interact with CVs appropriately, the library
requires the programmer to use specialized constructs provided.
These specialized constructs ensure that no CVs with scope beyond that of the
construct's body are invalidated within the constructs, meaning no changes to
the ownership permission levels.  The library macro-defines away ordinary C++
if / looping constructs so that the compiler outputs an error if ordinary C++
constructs are used in a translation unit that uses CVs.
Lastly, as in prior
work~\cite{MullerSiGoAcAgLe20}, the programmer should not use unsafe type
casts to modify priorities or ownership.

\subsection{Application Case Studies}

To evaluate the usability of our type systems, we developed 
a chat server in Rust (356 lines of code) and an email client in C++
(1,252 lines of code), and we ported the Memcached server, a large 
interactive application (~20,100 lines of C code) to
use our C++ type system.  We describe each application in turn and discuss our
experience.

\paragraph{Chat Server in Rust.}
Our simple multi-user chat server utilizes a tiered set of threads.
Threads at one priority handle each
connected user's TCP stream, and accept and handle new users.  Threads at a
lower priority write recent messages to the channel's metadata. Each channel is
associated with a CV that is signaled whenever a new message arrives from any
connected client.  A new thread is spawned for each connected user, which takes
in the TCP stream and the CV.  Each of these threads acts as both a producer and
consumer, both sending messages from the connected user and routinely waiting on
the CV for new messages from other connected clients.

%

\paragraph{Email Client in C++.}
Using our C++ implementation, we implemented a multi-user shared email server
based on
prior work~\citep{MullerSiGoAcAgLe20}.  The server utilizes 5 priority levels.  At the
highest priority, the server accepts new client connections (one thread per
connection), listens to requests from connected clients, and spawns off other
types of threads to perform requests received: \code{send} threads process
requests to send emails (highest priority); \code{sort} threads handle
requests for sorting emails (2nd highest); and \code{compress} threads
compress emails and \code{print} threads handle requests to print emails (3rd
highest).  There is also a \code{dequeue} thread that dequeues compression
tasks that arise from sending emails and generates threads to perform the
compressions (4th highest).  The main thread, which performs setup and
tear-down tasks, is at the lowest priority.

The interesting interaction is how \code{compress} threads are generated.  Two
different kinds of threads interact in a multi-producer/single consumer model:
the (multiple) \code{send} threads act as producers, and the \code{dequeue}
thread acts as the consumer.  Since the number of \code{send} threads is not
known at compile time (it depends on the number of send requests), we utilize
the alternative splitting macro (that does not invalidate the input) within a
loop to split off CV signal permissions as we spawn off each \code{send}
thread.
When a \code{send} thread pushes the number of emails in an inbox over a
threshold of uncompressed messages, it enqueues information on what to
compress and signals the \code{dequeue} thread to wake up and generate a
\code{compress} thread.

\paragraph{The Memcached Object Cache Server in C.}
%
The Memcached object caching server~\cite{Memcached09} acts as a distributed
in-memory cache. It is a key-value store, with the core purpose of maintaining a
hash table of objects that can be updated or retrieved by clients.  Memcached is
written in C, and uses pthreads and I/O multiplexing to handle many
clients.   

The main thread in Memcached performs most of the setup code and spawns off other
long-running threads, such as a resize thread that resizes the hash table
when necessary, an LRU maintenance thread that maintains the (approximately)
least-recently-used (LRU) ordering of items in buckets in the hash table, an LRU
crawler thread that frees up items that have gone ``cold'' in the cache, a slab
maintainer thread for balancing free memory between different size classes of
Memcached's internal memory pool, and a logger thread that aggregates and logs
messages and statistics from the other threads.

We split the application into three different priority levels ---
whenever the actual resizing occurs, it is done as a high-priority task (spawned
off by the resizing thread); on the other hand, any logging related activities
are done as low-priority tasks.  Activities relating to handling client requests,
maintaining LRUs of the cache, and memory pool management, are done as 
medium-priority tasks.
This is convenient, as these tasks perform signaling and waiting on condition
variables within critical sections of common locks.\footnote{As explained in
  \secref{overview}, threads acquiring the same
lock must have the same priority level to avoid priority inversions.}
We chose to keep the single medium priority, as all the tasks at this
priority are inter-dependent (e.g., a client request may run out of
memory if the slab thread runs at a lower priority or if the LRU crawler thread
does not free the memory in time).
If we desired more fine-grained priority levels, we could have changed the
priority
levels of certain critical sections and potentially promoted CVs from lower to
higher priorities at more places.

Porting the Memcached server required modifying about 11,400 lines of memcached
code, or about 57\%.
The work was completed by one graduate student, who was also designing and
maintaining the C++ implementation of the type system.
Including time to adapt and extend the C++ library implementation, the
conversion required around 80 person-hours.
Many of the tasks involved in the conversion were mechanical changes,
including converting code from C to C++, as well as mechanical changes required
by our library (e.g. converting functions to commands, converting if
statements), and could have been significantly sped up by a refactoring tool.
We estimate that the conversion itself would take less than 40
person-hours without a refactoring tool, and on the order of 8 to 16
person-hours with a refactoring tool.

\paragraph{Discussion.}
Being able to type check the Memcached server using our C++ type system gives us
some confidence that our typing rules are not overly restrictive.  For the most
part, incorporating the priority annotations as we developed the applications
from scratch (i.e., the chat server and the email client) was fairly
straightforward.  Compilation errors that we encountered due to mishandling of
types (e.g., forgetting to use the right split function at the necessary code
point when transferring ownerships of CVs) are easily fixed and, indeed, helped
determine correct priorities for the threads.
As an example, the initial design of the Rust chat server
had the threads handling connected users at a higher
priority than those handling new users, but this design had a subtle priority
inversion, which was caught by the static restrictions.

The process of
converting the Memcached server was less straightforward, as the code base is
quite mature and uses coding patterns that do not work as naturally with the way
we encoded the typing rules.  For instance, the lock-acquire and lock-release
operations are not always well-nested (e.g., the lock is acquired in one
function and released in a different one), which required us to refactor the
code so that we can lift the critical section into its own parameterized
function.  Another example is the use of CVs involving control constructs (i.e.,
if-else and loops).
While transforming code to use the specialized macros for these constructs is
quite mechanical and straightforward, it can be error-prone if there are
multiple levels of nesting of such constructs (as occurs in Memcached).
If we were to write the code from scratch, we would have rewritten the control
flow to avoid complex nesting of control constructs.

The implementations and case studies also provided evidence that our system is
usable and is not overly complex.
None of the implementors are experts in the theory of type systems.
The implementation of the Rust library was performed by an advanced
undergraduate who was not involved in the design of the type system, after
reading the presentation in Section~\ref{sec:overview} and one or two
discussion meetings with the first author.
The implementation of the Rust chat server was performed by a different
advanced undergraduate, who was not involved in the design or implementation
of the type system, with similar preparation.

\section{Related Work}\label{sec:related}

\paragraph{Cooperative and Competitive Threading.}
Work on {\em cooperative} threading, in which parallel workers cooperate to
complete a (largely computational) task
dates back to the late 1970s and 1980s
with systems such as Id~\citep{id-78} and Multilisp~\citep{halstead84, halstead85}.
Since then, this model has been studied in the context of many languages,
including parallel dialects of
ML~\citep{manticore-damp07, manticore-implicit-11, rmab-mm-2016, AroraWeAc21},
Haskell~\citep{spjlekech08,haskell-dp-2007},
and Java~\citep{x10-2005,is-habanero-14}.
{\em Competitive threading}, or concurrency, in which threads are used
to improve a program's responsiveness, latency, or compositionality,
has been
the subject of research and practice for decades, as have the problems of
scheduling and synchronization that result.
%
%
The body of related research in this area is too large to review in detail,
but we refer the interested reader to an excellent summary by~\citet{sgg-os-05}.

%
Prior work by some of the authors
\citep{mah-responsive-2017,mah-priorities-2018} observed that
adding provable responsiveness guarantees to cooperative threading requires
adding language constructs to prioritize interactive threads, and careful
control to avoid priority inversions (which will be discussed later in this
section).
This work was later extended to handle fairness~\citep{mwa-fairness-2019}
as well as mutable state~\citep{MullerSiGoAcAgLe20}.

\paragraph{Efficient Scheduling and Cost Semantics.}
\citet{brent74}, and
later~\citet{eagerzala89}, bounded the lengths of schedules of parallel
programs in terms of {\em work} and {\em span}
Cost semantics~\citep{rosendahl89,sands90} extend traditional operational
semantics to track resource usage.
Cost semantics for parallel programs using a graph or DAG representation have
been used since the early work on the NESL language~\citep{blellochgr95, blellochgr96}.

The cost semantics of this paper builds most directly on our recent model for
parallel programs with mutable state~\citep{MullerSiGoAcAgLe20}.
That work introduced {\em weak edges} to represent happens-before dependencies
induced by writes and reads to memory, and {\em strengthening} to convert
weak edges to more standard dependency edges in calculating the work and span.
Given these definitions, we were able to show a scheduling bound (including
both throughput and response time for interactive programs) in the style
of Brent, assuming graphs are {\em well-formed}, that is, lack priority
inversions, which they enforce statically.
In this paper, we use weak edges to represent happens-before relations induced
by contention on mutexes, and use a combination of static and dynamic
techniques to avoid priority inversions.
Our soundness result builds on and generalizes these theorems in several ways.
In addition, our extensions model changing priorities at runtime (to implement
the priority ceiling protocol), which is not handled by prior work on
responsive parallelism.

\paragraph{Priority Inversions.}
%
%
Priority inversions have been studied since the 1980s
in a variety of languages and systems
(e.g.,~\citep{lr-mesa-1980, cs-priority-1987}), and a number of
solutions have been proposed, including {\em priority inheritance} and
{\em priority ceiling}~\citep{srl-priority-1990}; we implement the
latter because it often results in fewer priority promotions than
priority inheritance and is simpler to reason about, though it
requires foreknowledge of each mutex's priority ceiling.

The problem of priority inversion control with condition variables is much
less well-studied.
As discussed in Sections~\ref{sec:intro} and~\ref{sec:overview},
standard dynamic priority inheritance
techniques are not sufficient to avoid priority inversions in programs with
CVs.
%
This problem was observed by \citet{Cucinotta2013}, who proposes allowing
programmers to specify such dependencies between tasks, so that
the runtime scheduler can perform priority inheritance as necessary.
We discover these dependencies automatically
using our type system rather than requiring programmers to specify them,
and rule out resulting priority inversions at compile time.
To the best of our knowledge, this paper presents the first type system
for avoiding priority inversions in programs using mutexes and CVs.

\citet{bms-formal-1993} present an elegant graph-based formalization of
priority inversions and techniques for avoiding them, which is in some
ways similar to the graph-based model used in this paper.
However, their formulation is based on snapshots of threads blocking on
each other during execution and would not suffice to represent complex
priority inversions involving the spawn histories of threads, such as the
one that arises in the code of Figure~\ref{fig:pc-terr}.

\paragraph{Ownership and Permissions.}
Our typing restrictions for condition variables heavily rely on a notion
of ownership.
Ownership has been used in a number of systems, often to prevent
data races in concurrent programming or to enable safe memory management.
A number of mechanisms for tracking ownership and related properties have
arisen, largely derived from linear logic~\citep{system-f-reynolds-1974,Girard1987}.
These include 
capabilities~(e.g., \citep{afm07-l3,crarywamo99}) and alias types~\citep{swm00},
and have been incorporated into languages
such as Cyclone~\citep{Cyclone02} and Rust (thus inspiring our Rust implementation).
%
%
We are most closely inspired by fractional
permissions~\citep{Boyland03}, altough we use a fairly common relaxation of
fractional permissions (e.g.,~\citep{heule2011fractional,NadenBoAlBi12})
which is less restrictive.

In the language~$L^3$, \citet{afm07-l3}
used capabilities in a type system to allow {\em strong
  updates}, which change the type of a memory cell (ownership is required for
strong updates because changing the type will invalidate other references to
the cell).
Our~\texttt{promote} operation may be seen as a kind of strong update, as
it changes the priority of a condition variable handle, which in turn changes
its type, and the restrictions on promotion were loosely based on those
of~$L^3$.
Promotions differ from strong updates, however, in that they do not invalidate
all other handles to the condition variable; other threads may still use a
handle to signal at a higher priority or wait at a lower priority.
We allow this using fractional permissions, which are not considered in~$L^3$.

\section{Conclusion}
In this paper, we have presented {\calcname}, a calculus for responsive parallel
programming with synchronization in the form of mutexes and condition variables.
The type system of {\calcname} statically prevents priority inversions which
might arise from the use of CVs, and guarantees that the graphs produced by
{\calcname}'s cost semantics  obey strong
scheduling bounds which preclude priority inversions.
The type system can be approximately encoded in Rust and
C++, and these encodings can be used to write substantial programs.
An exciting direction of future research would be to develop a
compiler to support the typing restrictions directly; such a custom compiler
could track ownership of CVs with less or no programmer
intervention, and could discover priorities of CVs and priority ceilings
of mutexes automatically through type inference.

\bibliography{main,new,thisbib}
\end{document}